\documentclass[10pt]{article}

\usepackage[T1]{fontenc}
\usepackage[letterpaper,textwidth=6.8in,textheight=9in,centering]{geometry}
\usepackage[parfill]{parskip}
\usepackage{amsmath,amssymb,amsthm,mathtools}
\usepackage[tt=false,type1=true]{libertine}
\usepackage[varqu]{zi4}
\usepackage[libertine]{newtxmath}
\usepackage{microtype,graphicx}
\usepackage{enumitem}
\usepackage{booktabs,tabularx,needspace}
\usepackage{algorithm,algpseudocode}
\usepackage[dvipsnames]{xcolor}
\definecolor{PaperLinkBlue}{HTML}{2A7FB8}
\definecolor{PaperBlue}{HTML}{1F4E79}
\usepackage{titlesec}
\titleformat{\section}{\normalfont\Large\bfseries}{\thesection}{1em}{}
\titleformat{\subsection}{\normalfont\large\bfseries\color{PaperBlue}}{\thesubsection}{1em}{}
\titleformat{\subsubsection}{\normalfont\normalsize\bfseries\color{PaperBlue}}{\thesubsubsection}{1em}{}
\titleformat{\paragraph}[runin]{\normalfont\normalsize\bfseries\color{PaperBlue}}{\theparagraph}{1em}{}
\titlespacing*{\paragraph}{0pt}{1.7mm}{0.6em}
\usepackage{hyperref}
\hypersetup{colorlinks=true,linkcolor=PaperLinkBlue,citecolor=PaperLinkBlue,urlcolor=PaperLinkBlue,
 pdftitle={A Walk From Free Probability to Matrix Discrepancy I: Matrix Spencer},
 pdfauthor={Tarun Kathuria}}
\numberwithin{equation}{section}
\newtheorem{theorem}{Theorem}[section]
\newtheorem{lemma}[theorem]{Lemma}
\newtheorem{proposition}[theorem]{Proposition}
\newtheorem{corollary}[theorem]{Corollary}
\theoremstyle{definition}
\newtheorem{definition}[theorem]{Definition}
\theoremstyle{remark}
\newtheorem{remark}[theorem]{Remark}
\DeclareMathOperator{\Tr}{Tr}
\DeclareMathOperator{\ran}{ran}
\DeclareMathOperator{\rank}{rank}
\DeclareMathOperator{\diag}{diag}
\newcommand{\HS}{\mathrm{HS}}
\newcommand{\op}{\mathrm{op}}
\newcommand{\Id}{\mathrm{Id}}
\newcommand{\R}{\mathbb R}

\newcommand{\E}{\mathbb E}

\allowdisplaybreaks[2]

\newcommand{\A}{\widehat A}
\newcommand{\norm}[1]{\lVert#1\rVert}

\newcommand{\proc}[1]{\textnormal{\textsc{#1}}}

\title{A Walk From Free Probability to Matrix\\
Discrepancy I: Matrix Spencer}
\author{Tarun Kathuria\\[0.6ex]
\includegraphics[width=1.4in]{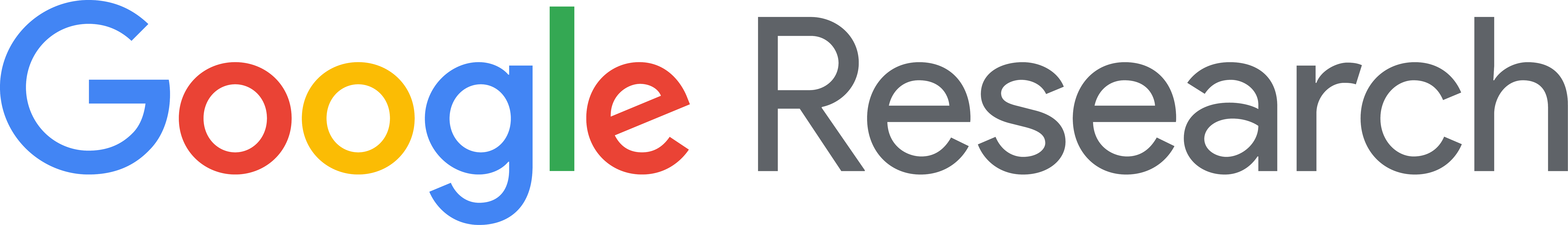}\\[0.6ex]
{\normalsize\href{mailto:tarunkathuria@google.com}{\texttt{tarunkathuria@google.com}}}}
\date{September 15, 2026}
\hypersetup{pdftitle={A Walk From Free Probability to Matrix Discrepancy I: Matrix Spencer},pdfauthor={Tarun Kathuria}}
\begin{document}
\maketitle
\begin{abstract}
The Matrix Spencer conjecture asks whether any $n$ real symmetric
matrices $A_1,\ldots,A_n\in\R^{m\times m}$ of operator norm at most
one admit a signing $x\in\{-1,1\}^n$ such that
\[
 \left\|\sum_{i=1}^n x_iA_i\right\|
 \le O\!\left(\sqrt{n\max\{1,\log(m/n)\}}\right).
\]
In the real-arithmetic model, we give a randomized algorithm running in
polynomial time that establishes this bound. We first
prove the $O(\sqrt n)$ bound for $m\le n$, resolving the square case,
and then obtain the rectangular bound by changing the regularizer.
After the results of this paper were completed, a preprint by Akbas
and Sra \cite{akbasSra2026} resolved the square Matrix Spencer conjecture
but gives a suboptimal dependence on $\log(m/n)$ in the rectangular
case. Our approach improves the logarithmic exponent from $2$ to the
optimal $1/2$.

As in earlier algorithmic discrepancy methods
\cite{lovettmeka2012,bansalLaddhaVempala2022,pesentivladu2026}, we run a
covariance-controlled random walk from the origin of the hypercube,
rounding coordinates near its faces and keeping them fixed. Our potential measures a
soft spectral edge of the evolving discrepancy matrix perturbed by an
operator-valued free semicircular element. Inspired by the free
interpolation approach of Bandeira, Boedihardjo, and van Handel
\cite{bbvh2023}, we combine Lehner's variational formula for the free
edge \cite{lehner1999} with spectral Tsallis regularization
\cite{allenZhuLiaoOrecchia2015,pesentivladu2026}. This puts the discrepancy
and remaining covariance in a single smooth optimization problem. The
potential has a finite-dimensional semidefinite formulation. We analyze
the optimizer's stability through the linearized Karush--Kuhn--Tucker
(KKT) system of a regularized min--max problem, whose stationarity
equations are related to the matrix Dyson equation \cite{erdos2019}.
This lets us find a large subspace in which to move while controlling
discrepancy. The square case uses the
Tsallis--$1/2$ regularizer; the rectangular case uses a suitable
generalized Tsallis power regularizer.

Our companion paper \cite{kathuria2026ks} applies these ideas to give an algorithmic proof of
Weaver's discrepancy theorem, whose existence proof by Marcus,
Spielman, and Srivastava resolved the Kadison--Singer conjecture
\cite{mss2015}.
Lean formalizations of our main discrepancy theorems have been
completed and will be released shortly.
\end{abstract}

\section{Introduction}
Let $A_1,\ldots,A_n$ be real symmetric $m\times m$ matrices with
$\norm{A_i}\le1$, where $\norm{\cdot}$ denotes the operator norm.
We seek signs $\sigma_i\in\{-1,1\}$ for which
$\norm{\sum_i\sigma_iA_i}$ is small. Throughout, $n$ counts matrices
and $m$ is their dimension.\footnote{A nonsymmetric real matrix $B_i$
can be replaced by its symmetric dilation
$\left(\begin{smallmatrix}0&B_i\\B_i^T&0\end{smallmatrix}\right)$.
This preserves signed-sum norms and doubles the dimension.}
For diagonal inputs this is a vector discrepancy problem. General
symmetric inputs require control of spectral directions that can change
as the coefficients move.
Our algorithm runs in polynomial time in the real-arithmetic
model.\footnote{We allow scalar real arithmetic, comparisons,
nonnegative square roots, exact symmetric eigendecomposition, and
uniform random draws. The required feasible SDP solutions are computed
internally to certified objective accuracy by an ellipsoid algorithm.
Definition~\ref{model:main} specifies the operation counts.}

\Needspace{9\baselineskip}
\begin{theorem}\label{thm:main}\label{main:combined}
There are signs satisfying
\[
 \left\|\sum_{i=1}^n\sigma_iA_i\right\|\le
 \begin{cases}
 10^8\sqrt n,&m\le n,\\
 10^8\sqrt{n(1+\log(2m/n))},&m\ge n\ge1.
 \end{cases}
\]
For every integer $b\ge1$, a randomized algorithm returns such a signing
or failure, with failure probability at most $2^{-b}$. Every execution
runs in time polynomial in $n,m,b$ in this real-arithmetic model,
with constants independent of the input matrix entries.
\end{theorem}

Section~\ref{impl:runtime} bounds every execution, including failed trials;
the existence conclusion
follows from positive success probability and has no computational
assumptions. Empty inputs or zero dimension are immediate.

The walk begins at $x=0$ and gradually fixes the coefficients at signs.
Alongside $x$ it maintains a positive semidefinite matrix $C$ on coefficient
space. This matrix is a reserve: decreasing it can compensate for the
spectral cost of moving $x$. The main analytic estimate tells us when
that compensation is useful. It requires a cap on the potential's
covariance derivative, so the algorithm first removes covariance in
directions that violate the cap. It then keeps covariance eigenvalues at
least $1/2$, intersects their subspace with the frozen-coordinate and
radial constraints, and walks uniformly in an orthonormal basis of the
intersection. The matching withdrawal is a small multiple of half that
subspace's projection. Preparation pays for the
curvature estimate; movement makes progress toward a vertex.

Spencer's theorem gives the corresponding $O(\sqrt n)$ scale for
bounded scalar entries in the square regime \cite{spencer1985}.
Constructive discrepancy methods build signs through correlated or
subspace-restricted random walks, beginning with Bansal's semidefinite
method and the edge walk of Lovett and Meka
\cite{bansal2010,lovettmeka2012}. For matrices, Bansal, Jiang, and Meka
prove an $O(\sqrt n)$ bound when every input has rank at most
$n/\log^3 n$; their argument uses the refined matrix concentration
inequality of Bandeira, Boedihardjo, and van Handel \cite{bjm2022}.
The theorem here imposes no rank hypothesis. Our computational claim is
in this real-arithmetic model, so its comparison with algorithms in
other computation models is limited to their discrepancy guarantees.

After the results of this paper were completed, a preprint by Akbas
and Sra \cite{akbasSra2026} resolved the square Matrix Spencer conjecture
through hereditary matrix small-ball estimates but gives a suboptimal
dependence on $\log(m/n)$ in the rectangular case.
For $m\ge n$, their Theorem~1.3 and the partial-coloring argument in
Section~3.5 yield the rectangular bound
$O\!\left(\sqrt n\,(1+\log(2m/n))^2\right)$.
Our approach gives the bound $O\!\left(\sqrt{n\log(2m/n)}\right)$,
improving the logarithmic exponent from $2$ to the optimal $1/2$.
The square-root dependence on $\log(2m/n)$ is optimal already for
diagonal matrices: classical vector-discrepancy lower bounds give
$\Omega\!\left(\sqrt{n\log(2m/n)}\right)$ for
$n\le m\le2^{cn}$ with a sufficiently small absolute constant $c>0$
\cite[Section~4.1, Exercise~1(c)]{matousek2010}.
For larger dimensions, the trivial discrepancy bound $n$ provides the
natural saturation scale.

Akbas and Sra obtain their algorithm by applying the shifted
partial-coloring theorem of Reis and Rothvoss
\cite[Theorem~6]{reisRothvoss2023} to their small-ball estimate.
This theorem uses Rothvoss's Gaussian-projection method
\cite{rothvoss2017}: project a Gaussian point onto a discrepancy body
intersected with a box. For spectral discrepancy bodies, these
nearest-point problems are SDPs.\footnote{For radius $r>0$ and coefficient
vector $z$, the spectral constraint is
$-rI\preceq\sum_i z_iA_i\preceq rI$. The quadratic projection objective
has a semidefinite epigraph by the Schur complement.}
Our proof gives an end-to-end random walk and also implies that these
projection SDPs succeed. To see this, for $k$ active matrices set
\[
 r_m(k)=\sqrt{k\max\{1,\log(2m/k)\}},\qquad
 K_k=\left\{z\in\R^k:
       \left\|\sum_{i=1}^k z_iA_i\right\|\le c_0r_m(k)\right\},
\]
where $c_0$ is a sufficiently large universal constant and the active
matrices have been relabeled. Since $r_m(k)$ is increasing in $k$,
Theorem~\ref{thm:main} gives the same bound for every subfamily.
Reis and Rothvoss \cite[Remark~3]{reisRothvoss2023} convert this hereditary
discrepancy bound into
$\gamma_k(K_k)\ge2^{-O(k)}$, where $\gamma_k$ is standard Gaussian
measure on $\R^k$. Their shifted partial-coloring theorem now applies
at each stage. The stage costs are summable,
$\sum_{j\ge0}r_m(2^{-j}n)=O(r_m(n))$, so the projection method also
achieves our square and rectangular bounds up to universal constants.

The construction is motivated by the interpolation approach of
Bandeira, Boedihardjo, and van Handel, which compares random matrices
with their operator-valued free counterparts through moments and
resolvent statistics \cite[Section 1.4.2]{bbvh2023}. Lehner's formula
provides a finite-dimensional variational description of the free
spectral edge \cite{lehner1999}. We smooth that description using the
density regularizers developed by Allen-Zhu, Liao, and Orecchia for
spectral sparsification framed as matrix optimization
\cite{allenZhuLiaoOrecchia2015}, and by Pesenti and Vladu for discrepancy
minimization \cite{pesentivladu2026}. The resulting
potential treats the discrepancy center and covariance in one optimization.

The maximizing density and minimizing transport form a saddle point
of a regularized min--max problem. Their stationarity and normalization
equations are its Karush--Kuhn--Tucker (KKT) conditions. Differentiating
these conditions gives the linearized KKT system governing the
optimizer's response to changes in the center and covariance.
Its operator structure is related to the analysis of
matrix Dyson equations by Ajanki, Erd\H{o}s, and Kr\"uger
\cite{aek2019}, and to the regular-edge stability analysis of Alt,
Erd\H{o}s, Kr\"uger, and Schr\"oder \cite[Section 4]{aeks2020}.
The regularizer and trace constraint change the equations. Our task is
to bound their response to the particular center and covariance changes
made by the walk. The discussion below explains this task before
introducing the operators used to solve it.

The companion paper \cite{kathuria2026ks} develops the same
free-probability variational approach into a deterministic algorithm
for Weaver's discrepancy problem and the Kadison--Singer theorem.
Lean formalizations of the main discrepancy theorems proved here have
also been completed; the formal proofs will be released shortly.

\paragraph{Reading the proof.}
Section~\ref{sec:motivation} develops the potential from a scalar example,
and Section~\ref{sq:overview} explains the algorithm through the estimate
it needs. We then derive the variational formulas and the response bound,
and use them to prove the square signing theorem in
Section~\ref{walk:phases}. The rectangular argument begins in
Section~\ref{rect:extension}; it uses a suitable generalized Tsallis
power regularizer while keeping the same walk. The last three sections
construct the numerical reports and bound the work of the resulting algorithm.
Appendix~\ref{reg:section} supplies the uniform derivative estimates
needed for finite steps. The numerical constants are chosen for a
transparent proof, without an attempt at optimization.

\section{A spectral potential with a covariance reserve}\label{sec:motivation}
\subsection{The state and its two matrix spaces}
Set $D=2m$ and define
\[
 \A_i=\diag(A_i,-A_i),\qquad H(x)=\sum_i x_i\A_i.
\]
Then $\lambda_{\max}(H(x))=\norm{\sum_i x_iA_i}$.
A coordinate is \emph{live} when $|x_i|<1$ and \emph{frozen} when
$|x_i|=1$. Frozen coordinates remain in $H(x)$.
An epoch retains an index set $I_0$ of $\ell$ initially live labels,
even if some freeze during that epoch. Its covariance satisfies
$0\preceq C\preceq I_\ell$. For a coefficient vector $a\in\R^{I_0}$,
write $\mathcal A(a)=\sum_{i\in I_0}a_i\A_i$.

The source associated with $C$ is the linear map on physical matrices
\begin{equation}\label{sq:fo:eta}
 \eta_C(X)=\sum_{i,j\in I_0}C_{ij}\A_iX\A_j.
\end{equation}
Factoring $C=RR^T$ and setting $T_a=\sum_iR_{ia}\A_i$ gives
$\eta_C(X)=\sum_aT_aXT_a$. In particular it preserves positive
semidefiniteness and is self-adjoint for the trace pairing. Also,
\begin{equation}\label{sq:fo:sourcebudget}
 0\preceq\eta_C(I_D)\preceq\sum_{i\in I_0}\A_i^2\preceq\ell I_D.
\end{equation}
The matrix $C$ is selected by the algorithm; it is not the covariance
of the random point $x$ or a prescribed function of the coordinate margins.

All traces are unnormalized. A \emph{density} is a positive semidefinite
$D\times D$ matrix of trace one. The Hilbert--Schmidt pairing is
$\langle X,Y\rangle=\Tr(X^TY)$, with norm
$\|X\|_{\mathrm{HS}}=\sqrt{\Tr(X^TX)}$.
The trace norm $\|X\|_1$ is the sum of the singular values of $X$. For a subspace $W$, $P_W$ is its orthogonal projection. A superscript
$T$ or $*$ denotes the transpose or the adjoint between the stated real
Hilbert spaces. A positive matrix is called faithful if it is positive
definite on its stated space. For a linear map $T$, its range is
$\ran T=\{Tv:v\text{ is in the domain of }T\}$ and its kernel is
$\ker T=\{v:Tv=0\}$. The support of a positive semidefinite matrix is
its range, the orthogonal complement of its kernel. Inverses on a support
are inverses of the corresponding restricted matrices.

\subsection{What the density measures}
For positive semidefinite matrices $S,M$, define their fidelity by
\[
 F(S,M)=\Tr\sqrt{S^{1/2}MS^{1/2}}.
\]
We first use the square-root regularizer and retain a parameter
$\theta>0$ to display its role. Define
\begin{align}
 f_\theta(H)&=\max_{S\succeq0,\,\Tr S=1}
       \{\Tr(HS)+2\theta\Tr\sqrt S\},\label{sq:fo:f}\\
 E_\theta(H,C)&=\max_{S\succeq0,\,\Tr S=1}
       \{\Tr(HS)+2F(S,\eta_C(S))+2\theta\Tr\sqrt S\}.
       \label{sq:fo:energy}
\end{align}
Theorem~\ref{origin:formula} will show that this potential is also
\[
 E_\theta(H,C)=\inf_{Z\succ0}
 f_\theta\bigl(H+Z^{-1}+\eta_C(Z)\bigr).
\]
Here $Z$ is a positive transport matrix. The two formulas play different
roles: the density form defines a concave optimization, and the transport
form makes its relation to the free spectral edge visible.
The density gives a spectral test for $H$. The square-root term smooths
that test, while fidelity measures the remaining source through the
same density. The square argument takes $\theta=1$; we usually suppress this fixed
parameter and write $E$ and $f$.
The density is optimized anew at each state, so its response must be
included when differentiating the value.

These quantities satisfy
\begin{equation}\label{eq:reserve}
 \lambda_{\max}(H)\le f_\theta(H)\le E_\theta(H,C)
 \le f_\theta(H)+2\sqrt\ell.
\end{equation}
Indeed, $F(S,M)=\|S^{1/2}M^{1/2}\|_1
\le\sqrt{\Tr S\Tr M}$ by the Schatten Cauchy--Schwarz inequality.
Self-adjointness of $\eta_C$ gives
$\Tr\eta_C(S)=\Tr(S\eta_C(I_D))\le\ell$.
Maximizing the resulting bound proves the last inequality; nonnegativity
of fidelity and the rank-one spectral test give the others.
At zero center, $f_\theta(0)=2\theta\sqrt D$ by scalar
Cauchy--Schwarz on the density eigenvalues. Thus the initial potential
has order $\sqrt n$ in the square regime.

The last inequality in \eqref{eq:reserve} also has an algorithmic use.
A fresh identity covariance increases the potential by at most
$2\sqrt\ell$ above its source-free value. Replenishing the reserve is
therefore possible at a controlled cost, which must be included whenever
an epoch restarts.

\subsection{Why retain a map rather than a variance matrix?}
The ordinary remaining variance associated with $C$ is
$\eta_C(I_D)=\sum_{ij}C_{ij}\A_i\A_j$. With diagonal weights
$w_i=1-x_i^2$, this becomes $\sum_i(1-x_i^2)\A_i^2$.
It records the map at the identity. The potential uses its action on
other matrices because the spectral test is itself changing.

The rank-one example in Section~\ref{sec:rank-one-example} will show
that even a fixed family of rank-one inputs can have covariance states
with identical ordinary variance and different reserve terms in the
density objective. The proof uses these interactions with the changing
spectral test through the full sandwich map.

The distinction from classical Gaussian covariance is separate. A
Gaussian matrix sum can have the same map $\eta_C$. Choosing a free
semicircular model specifies higher moments as well, and leads to the
Lehner variational structure. The algorithm uses the finite-dimensional form of that structure.
For standard freely independent semicircular variables $s_a$ and
self-adjoint coefficient matrices $T_a$, the operator
$H\otimes1+\sum_a T_a\otimes s_a$ has matrix-valued covariance
$X\mapsto\sum_a T_aXT_a$. Lehner's variational description of its
upper edge is the unregularized version of the displayed transport
formula \cite{lehner1999,bbvh2023}. Its optimization over $Z$ retains
the spectral effect of the free perturbation after the free variables
have disappeared from the formula.

\subsection{The diagonal example}
Suppose the $\A_i$ are diagonal. For each physical coordinate
$r\in\{1,\ldots,D\}$, let $a_r\in\R^{I_0}$ collect their diagonal
entries, put $h_r=H_{rr}$, and set $v_r=a_r^TCa_r$.
Averaging diagonal sign conjugations replaces a density by its diagonal
without decreasing the objective: the source commutes with these
conjugations, fidelity is jointly concave, and the trace square root is
concave. For $S=\diag(s_1,\ldots,s_D)$ this gives
\begin{equation}\label{eq:diagonal}
 E_\theta(H,C)=\max_{s_r\ge0,\,\sum_rs_r=1}
 \left\{\sum_rs_r(h_r+2\sqrt{v_r})+
                    2\theta\sum_r\sqrt{s_r}\right\}.
\end{equation}
Each diagonal entry acquires a shift equal to twice its remaining
standard deviation. Reducing covariance lowers these shifts. The
weights $s_r$ then readjust to the new spectral landscape. This
readjustment is already visible in a commuting example and is the
source of the optimized curvature in the general proof.

The diagonal formula already suggests the derivative comparison that
will drive the walk. Let $\Gamma$ be the derivative with respect to
covariance, defined by $D_CE[\Delta C]=\Tr(\Gamma\Delta C)$ on the
current covariance support. Write $b_r=P_{\ran C}a_r$ and omit terms
with $v_r=0$. Envelope differentiation gives
\[
 \Gamma=\sum_r\frac{s_r}{\sqrt{v_r}}b_rb_r^T.
\]
If $\Gamma\preceq tI$ on $\ran C$, then each positive summand is
bounded by $tI$. Since $C\preceq I$, this implies
$s_r\sqrt{v_r}\le t$. The inverse curvature of the square-root
regularizer is diagonal with entries $2s_r^{3/2}/\theta$.
Dropping the trace constraint in its inverse quadratic form therefore
gives, for the half-Hessian $J$ of the optimized value,
\[
 \Tr(CJ)\le\frac1\theta\sum_r v_rs_r^{3/2}
 \le\frac{\sqrt t}{\theta}\sum_r v_r^{3/4}s_r
 \le\frac{\sqrt t}{\theta}\ell^{3/4}.
\]
The last inequality uses $v_r\le\ell$ and $\sum_rs_r=1$.
Thus a covariance derivative cap $t=L/\sqrt\ell$ yields an average
curvature bound $\sqrt{L\ell}/\theta$ in the diagonal example.
The noncommutative proof seeks the same scale, with a larger absolute
constant and a matrix-space response calculation in place of this
scalar computation. Section~\ref{sq:an:section} derives the constrained
inverse formula that justifies this calculation as well.

\subsection{A rank-one example: what ordinary variance misses}\label{sec:rank-one-example}
Suppose $A_i=v_iv_i^T$ and first work in the original $m$-dimensional
matrix space. For a symmetric $m\times m$ matrix $Z$, write
\[
 \widetilde\eta_C(Z)=\sum_{i,j\in I_0}C_{ij}A_iZA_j.
\]
For diagonal $C=\diag(c_i)$ with $0\le c_i\le1$, the rank-one identity
$A_iZA_i=(v_i^TZv_i)v_iv_i^T$ gives
\begin{equation}\label{eq:rank-one-source}
 \widetilde\eta_C(Z)=\sum_i c_i(v_i^TZv_i)v_iv_i^T,
 \qquad
 \widetilde\eta_C(I_m)=\sum_i c_i\norm{v_i}^2v_iv_i^T.
\end{equation}
At the identity, each term is weighted by the squared length of its
vector. At a general $Z$, that weight becomes the quadratic form
$v_i^TZv_i$, which also depends on the vector's position relative to $Z$.

For a concrete calculation, take $I_0=\{1,2,3,4\}$ and $m=2$.
Let $e_1,e_2$ be the standard basis, put
$u_\pm=(e_1\pm e_2)/\sqrt2$, and take
\[
 A_1=e_1e_1^T,\quad A_2=e_2e_2^T,\quad
 A_3=u_+u_+^T,\quad A_4=u_-u_-^T.
\]
These are rank-one projections of norm one; the family is
noncommutative, since $A_1A_3\ne A_3A_1$. Compare the two admissible
covariances
\[
 C^{(0)}=\diag(1,1,0,0),\qquad
 C^{(1)}=\diag(0,0,1,1).
\]
Both have the same variance:
$\widetilde\eta_{C^{(0)}}(I_2)=
\widetilde\eta_{C^{(1)}}(I_2)=I_2$.
At $Z=\diag(2,1)$, however, \eqref{eq:rank-one-source} gives
\[
 \widetilde\eta_{C^{(0)}}(Z)=\begin{pmatrix}2&0\\0&1\end{pmatrix},
 \qquad
 \widetilde\eta_{C^{(1)}}(Z)=\frac32I_2,
\]
since $u_+^TZu_+=u_-^TZu_-=3/2$.

The difference also appears directly in the density objective.
For $0<p<1$, use the density $S_p=\diag(p,1-p)$. Then
\[
 \widetilde\eta_{C^{(0)}}(S_p)=S_p,\qquad
 \widetilde\eta_{C^{(1)}}(S_p)=\frac12I_2,
\]
so the two fidelity terms are
\begin{equation}\label{eq:rank-one-fidelity}
 F(S_p,\widetilde\eta_{C^{(0)}}(S_p))=1,
 \qquad
 F(S_p,\widetilde\eta_{C^{(1)}}(S_p))
       =\frac{\sqrt p+\sqrt{1-p}}{\sqrt2}.
\end{equation}
The latter is strictly less than one when $p\ne1/2$.
For the same center and this same density, the linear and regularizer
terms are unchanged. What changes is the value assigned to the
covariance reserve by an anisotropic spectral test.

The intermediate covariance states
$C_t=(1-t)C^{(0)}+tC^{(1)}$, for $0\le t\le1$, all retain variance
$I_2$. Their source at $S_p$ is $(1-t)S_p+(t/2)I_2$, so
\[
 F(S_p,\widetilde\eta_{C_t}(S_p))
 =\sqrt{p\bigl((1-t)p+t/2\bigr)}
  +\sqrt{(1-p)\bigl((1-t)(1-p)+t/2\bigr)}.
\]
For $0<t<1$, both pairs of projections contribute to the source,
including noncommuting matrices. Comparing these covariance states
isolates the dependence of the reserve on the directions in which
covariance remains available.

These are also evaluations of the signed potential used in the paper.
Indeed, embedding $S_p$ as the $4\times4$ density
$\widehat S_p=\diag(S_p,0_2)$ gives
\[
 \eta_C(\widehat S_p)
   =\diag(\widetilde\eta_C(S_p),0_2),\qquad
 F(\widehat S_p,\eta_C(\widehat S_p))
   =F(S_p,\widetilde\eta_C(S_p)).
\]
The zero block is allowed in the density maximization.
Consequently, retaining only $\eta_C(I_D)$ would discard precisely the
interaction displayed in \eqref{eq:rank-one-fidelity}, even for
rank-one inputs.

\subsection{Matching movement with covariance withdrawal}
There is a useful classical analogy for the update we will use.
For a smooth function $g$ on coefficient space and a centered Gaussian
vector $g_C$ of covariance $C\succ0$, Gaussian convolution satisfies
\[
 u(x,C)=\mathbb E g(x+g_C),\qquad
 D_Cu[Q]=\tfrac12\Tr\bigl(Q\nabla_x^2u\bigr).
\]
For example, the identity holds for $g$ with bounded derivatives through
order two: differentiate the Gaussian density and integrate by parts
twice. For a rank-one direction $Q=vv^T$, the derivative of that density
is one half of its second derivative in direction $v$; linearity gives
the general formula. This calculation explains why a centered increment
of covariance $h^2Q$ naturally goes with a withdrawal of $h^2Q$ from a
remaining covariance.

Our potential uses the free variational structure instead of Gaussian
convolution. The exact heat-equation identity is replaced by a bound on
the optimized response, and covariance preparation makes that bound
applicable. The next section derives the resulting movement rule.
The same principle will be used with the power profile, for
$0<q\le1/2$, $\theta>0$, and $\kappa\ge0$,
\begin{equation}\label{intro:potential}
 E(H,C)=\max_{S\succeq0,\Tr S=1}
 \left\{\Tr(HS)+2F(S,\eta_C(S))+
 \frac\theta{1-q}\Tr S^{1-q}+2\kappa\Tr\sqrt S\right\}.
\end{equation}
The square choice is
\begin{equation}\label{sq:fo:theta}
 (q,\theta,\kappa)=(1/2,1,0),\qquad D=2m\le2n.
\end{equation}
For the rectangular regime, Section~\ref{rect:parameters} chooses
$q=1/p$ with $p$ a power of two, and balances $\theta$ against the
initial smoothing cost. The additional $\kappa$ term supplies uniform
numerical conditioning and is included in that regime's response proof.
\space

\section{How the potential determines the walk}\label{sq:overview}
\subsection{The local estimate to aim for}
At a smooth state, define the covariance derivative $\Gamma$ and the
half-Hessian $J$ on coefficient directions by
\[
 D_CE_\theta[\Delta C]=\Tr(\Gamma\Delta C),\qquad
 J[a,a]=\tfrac12D_H^2E_\theta[\mathcal A(a),\mathcal A(a)].
\]
Covariance derivatives are taken within the current support of $C$.
For a centered sample $\xi$ with $\E\xi\xi^T=Q\preceq C$, the
matched update is
\[
 x_{I_0}\gets x_{I_0}+h\xi,\qquad C\gets C-h^2Q.
\]
A symmetric sample cancels odd terms, giving the local expansion
\begin{equation}\label{eq:drift}
 \E[\Delta E_\theta\mid x,C]
   =h^2\Tr\bigl(Q(J-\Gamma)\bigr)+O(h^4).
\end{equation}
Lemma~\ref{walk:step} proves this expansion; Appendix~\ref{reg:section}
bounds its remainder uniformly on the states reached by the algorithm. The two quadratic terms explain the construction:
center curvature is paid against covariance withdrawal.

For the square profile, the key estimate is
\begin{equation}\label{eq:cap}
 \Gamma|_{\ran C}\preceq\frac{4096}{\sqrt\ell}I
 \quad\Longrightarrow\quad \Tr(CJ)\le770\sqrt\ell.
\end{equation}
This is Theorem~\ref{sq:te:cap-response-theorem}, specialized to
$L=4096$ and $\theta=1$. It concerns the fully optimized potential. Its proof linearizes the density and transport stationarity equations
\eqref{sq:an:nonlinear-stationarity}, balances the resulting matrix
operators, and controls the response forces through their relation to
the source. The analogy with MDE stability is at this level: a coupled
self-consistent system must be analyzed before its response can be
bounded. The regularized $Z$ is generally not an ordinary edge resolvent.
Near-edge MDE stability provides context for this response problem
\cite[Section 4]{aeks2020}; the broader deterministic theory includes
edges and cusps \cite{alt2020}. Here the required output is the
specific quadratic-form estimate \eqref{eq:cap}.

\subsection{Preparation, legal movement, and progress}
The cap in \eqref{eq:cap} is enforced by preparation. If a unit vector
$v\in\ran C$ has large $v^T\Gamma v$, a sufficiently small reduction
$C\gets C-\alpha vv^T$ lowers the potential by approximately
$\alpha v^T\Gamma v$. During this operation $x$ stays fixed.
Preparation trades covariance for potential decrease until the response
bound applies. Its trace loss is recorded separately from movement;
the numerical implementation also records discarded small eigenvalues.

At the prepared state, let $V$ be the coefficient subspace defined by
vanishing on frozen labels and orthogonality to $x_{I_0}$. Retain the
spectral subspace $U$ of $C$ with eigenvalues at least $1/2$, and put
\[
 W=U\cap V,\qquad Q=\tfrac12P_W,
\]
where $P_W$ is the orthogonal projection onto $W$. The cutoff gives
$Q\preceq C$. The trace ledger shows that $W$ still has dimension
proportional to the epoch's live count: little trace has been spent on
preparation or movement, and only a few new coordinate constraints
have appeared. Constructing $Q$ leaves the stored $C$ in place; the
update withdraws only $h^2Q$. Eigenvalues between the cleanup threshold
$2\delta$ and the movement cutoff $1/2$ remain in $C$: the cutoff
selects movement directions, while cleanup has its own recorded trace cost.

Let $k=\dim W>0$ and choose an orthonormal basis $u_1,\ldots,u_k$ of $W$
from the EVD of $P_W$. The movement has a simple geometric description:
choose a signed unit vector $\pm u_j$ uniformly among these $2k$
possibilities and take a small step of length $s_{\rm step}=h\sqrt{k/2}$.
The matching covariance withdrawal gives
\begin{equation}\label{eq:unit-direction-update}
 x_{I_0}'=x_{I_0}\pm s_{\rm step}u_j,\qquad
 C'=C-\frac{s_{\rm step}^2}{k}P_W.
\end{equation}
The covariance of the displacement is exactly
$(s_{\rm step}^2/k)P_W$, the amount withdrawn from $C$.
The parameter $h$ is fixed by the numerical recipe; the physical step
length $s_{\rm step}$ varies with the dimension of $W$.

For the drift calculation we write the same displacement as $h\xi$, where
$\xi=\pm\sqrt{k/2}\,u_j$. This normalization gives
$\E\xi=0$, $\E\xi\xi^T=Q$, and $\norm\xi^2=k/2=\Tr Q$.
The radial constraint therefore gives, for either sign,
\[
 \norm{x_{I_0}'}^2-\norm{x_{I_0}}^2=s_{\rm step}^2=h^2\Tr Q.
\]
Thus every movement makes the same radial progress at a given state.
The step-size choice keeps the movement inside the cube; sufficiently
near-face coordinates are then rounded and frozen.

\subsection{Why epochs can be accepted and restarted}
Save the initial state of an epoch as $x_*$, write $H_*=H(x_*)$,
and let $S_*$ maximize $f_\theta(H_*)$. Convexity makes
\[
 \Psi(x,C)=E_\theta(H(x),C)-f_\theta(H_*)
                 -\Tr(S_*(H(x)-H_*))
\]
nonnegative. At identity covariance it starts at most $2\sqrt\ell$.
The drift estimate controls its growth and pays for preparation loss.
The fixed supporting-plane contribution from centered movements is
controlled separately by a second-moment estimate, with rounding charged
by its deterministic norm bound.

A trial stops after enough coordinates freeze, its time budget expires,
or preparation consumes too much covariance. The acceptance test checks
the potential excess, the supporting-plane contribution, and the recorded
loss. With constant probability the trial is accepted and makes enough
radial progress or freezes enough labels. Failed trials restart from
$x_*$; accepted trials keep their endpoint.

The next epoch replenishes covariance on the remaining live labels.
This can raise $E_\theta$, but \eqref{eq:reserve} bounds the cost.
A bounded number of accepted epochs halves the live count, and summing
the square-root costs over these phases gives the square discrepancy
bound. The rectangular construction changes the density regularizer to
a power $1-1/p$, with $p$ on the scale of $1+\log(2m/n)$, and changes
the response and phase budgets accordingly. The preparation and movement
mechanism remains the same.

\begin{algorithm}[H]
\caption{Organization of one epoch; quantitative version in Algorithm~\ref{alg:epoch}}
\label{alg:overview}
\begin{algorithmic}[1]
\Require Saved coefficient vector $x_*$ and its $\ell$ live labels
\State Start from $x=x_*$ and identity covariance $C=I_\ell$
\While{the time, freezing, and preparation-loss budgets permit movement}
 \State Hold $x$ fixed and reduce $C$ until its covariance derivative is capped
 \State Let $W$ be the intersection of the $C$-eigenspace with eigenvalues at least $1/2$ and the frozen and radial constraints; set $Q=P_W/2$
 \State Let $k=\dim W$; choose $u$ uniformly from a signed orthonormal basis of $W$
 \State Set $s_{\rm step}=h\sqrt{k/2}$, $x_{I_0}\gets x_{I_0}+s_{\rm step}u$, and $C\gets C-(s_{\rm step}^2/k)P_W$; round near-face coordinates
\EndWhile
\State Accept if the covariance loss, potential excess, and supporting-plane term pass their tests
\State On rejection, retry from $x_*$; on acceptance, keep $x$ and pay for fresh covariance
\end{algorithmic}
\end{algorithm}

\begin{figure}[ht]
\centering
\includegraphics[width=\linewidth]{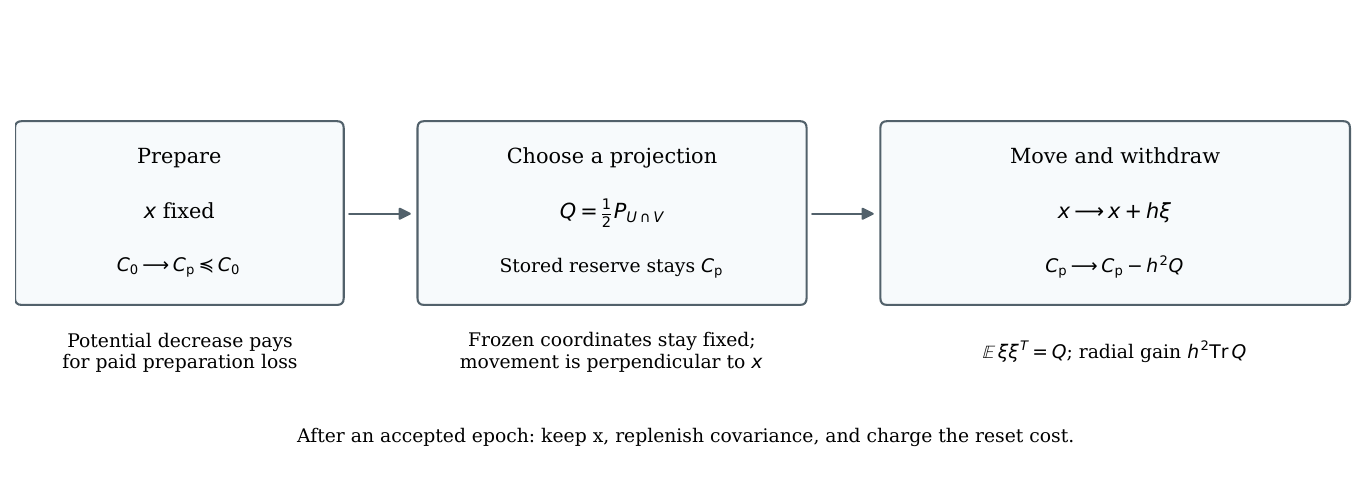}
\caption{One covariance cycle. Preparation changes $C_0$ to $C_p$ at a
fixed coefficient vector. Intersecting its high spectral subspace with
$V$ enforces frozen signs and radial orthogonality, and $Q$ is half the
projection onto that intersection. Computing $Q$ does not overwrite $C_p$; only
the sampled movement withdraws $h^2Q$. Paid preparation loss is charged
to potential decrease, while small discarded eigenvalues have a separate
trace budget. An accepted endpoint keeps its coefficients when the
covariance is reset.}
\label{fig:walk-roadmap}
\end{figure}
The full finite algorithm, including stopping rules, numerical reports,
and retry counts, is specified in Section~\ref{intro:algorithm}.
Its analysis uses the square response estimate proved next and the
contracts in Section~\ref{walk:contracts}, discharged in
Sections~\ref{query:section}--\ref{impl:main}.
\space

\section{The variational toolkit}\label{sq:sec:potential}

We now prove the matrix facts used by the potential. Fidelity has two
representations: its transport minimization is suited to differentiation,
and its semidefinite lift is suited to computation. The square-root
regularizer supplies both uniqueness of the density and a comparison
under compression, needed when the source has a kernel.

The full real matrix space $\R^{D\times D}$ has inner product
$\langle X,Y\rangle=\Tr(X^TY)$ and norm $\|X\|_{\HS}$.
Its symmetric subspace has inner product $\Tr(XY)$. The density and its
variations lie in that subspace. Products, vectorized frames, and tensor
operators will also act on the full matrix space. We write $\Tr_{\HS}$
for the trace of an operator on the specified matrix Hilbert space.
For a real symmetric physical matrix $X$, define
\[
 L_X(Y)=XY,\qquad R_X(Y)=YX,\qquad \operatorname{Ad}_X(Y)=XYX.
\]
Products such as $L_XR_Y$ act on the full real matrix space; individual
left or right multiplication need not preserve symmetric matrices.
All inverses are taken on the support explicitly stated in the proof.
We repeatedly use $\Tr(XY)\ge0$ for $X,Y\succeq0$, because
$\Tr(XY)=\Tr(Y^{1/2}XY^{1/2})$.

For Hilbert-space vectors $U,V$, the notation $|U\rangle\langle V|$
denotes the rank-one operator $X\mapsto U\langle V,X\rangle$.
If $\mathsf F$ is the linear map whose columns are $F_a$, their
\emph{frame operator} is $\mathsf F\mathsf F^*=
\sum_a|F_a\rangle\langle F_a|$, and their \emph{Gram matrix} is
$\mathsf F^*\mathsf F=(\langle F_a,F_b\rangle)_{a,b}$.
The two operators have the same nonzero eigenvalues: if
$\mathsf F^*\mathsf Fv=\lambda v$ with $\lambda>0$, then
$\mathsf Fv\ne0$ is a corresponding eigenvector of
$\mathsf F\mathsf F^*$, and the reverse correspondence uses
$\mathsf F^*$; these maps are inverse up to the factor $\lambda$.
We will use this elementary correspondence to turn covariance caps
into inequalities on matrix space.

\subsection{Square roots and fidelity}

This section supplies the matrix-function and variational facts used later.
We use the finite-dimensional spectral theorem, ordinary scalar calculus,
and the Cauchy--Schwarz inequality throughout.

\begin{lemma}[Square roots and their curvature]\label{sq:fo:sqrt}
For positive semidefinite matrices the square root is order preserving.
If $V$ is an isometry, then
\begin{equation}\label{sq:fo:sqrtcompression}
 V^T\sqrt S\,V\preceq\sqrt{V^TSV}.
\end{equation}
The function $S\mapsto\Tr\sqrt S$ is strictly concave. On the positive
definite cone it is smooth, and, in an eigenbasis $S=\operatorname{diag}(s_i)$,
\begin{align}
 D(\Tr\sqrt S)[X]&=\tfrac12\Tr(S^{-1/2}X),\label{sq:fo:sqrtfirst}\\
 -D^2(2\theta\Tr\sqrt S)[X,X]
 &=\theta\sum_{i,j}
 \frac{|X_{ij}|^2}{\sqrt{s_i}\sqrt{s_j}(\sqrt{s_i}+\sqrt{s_j})}.
 \label{sq:fo:sqrthessian}
\end{align}
The last quadratic form is strictly positive for $X\ne0$.
\end{lemma}
\begin{proof}
For $a\ge0$ elementary integration gives
\[
 \sqrt a=\frac1\pi\int_0^\infty
              \frac{a}{a+t}\,t^{-1/2}\,dt.
\]
The spectral theorem gives the same norm-convergent integral for a matrix.
If $0<X\preceq Y$, congruence by $X^{-1/2}$ and diagonalization show
$Y^{-1}\preceq X^{-1}$. The same comparison with $X+tI,Y+tI$ holds when
$X$ is singular. Consequently
$X(X+tI)^{-1}=I-t(X+tI)^{-1}$ is order preserving for each $t>0$.
Integrating proves square-root monotonicity.

To prove the compression assertion, extend $V$ to an orthogonal basis and
write a positive definite matrix $T$ in blocks with upper left block
$T_{11}=V^TTV$. Block elimination gives
\[
 V^TT^{-1}V=(T_{11}-T_{12}T_{22}^{-1}T_{21})^{-1}
 \succeq T_{11}^{-1}.
\]
Apply this to $T=S+tI$, subtract $t$ times the result from the identity,
and integrate the preceding square-root formula. This proves
\eqref{sq:fo:sqrtcompression} for all $S\succeq0$.

For $S\succ0$, put $Q=\sqrt S$. The linear map $Y\mapsto QY+YQ$
is invertible, since its entries in a $Q$-eigenbasis are multiplied by
$\sqrt{s_i}+\sqrt{s_j}>0$. Differentiating $Q^2=S$ therefore gives
\[
 (D\sqrt S[X])_{ij}=\frac{X_{ij}}{\sqrt{s_i}+\sqrt{s_j}}.
\]
This also proves smoothness by the finite-dimensional implicit function
theorem (Theorem~\ref{reg:ift}), applied on the open positive cone.
Taking the trace proves \eqref{sq:fo:sqrtfirst}. Differentiating
$Q^{-1}$ gives $D(Q^{-1})[X]=-Q^{-1}(DQ[X])Q^{-1}$, and substitution in
\eqref{sq:fo:sqrtfirst} gives \eqref{sq:fo:sqrthessian}.
Thus the trace of the square root has strictly negative second derivative
along any nonconstant line in the positive definite cone.
For two distinct semidefinite matrices $S,T$, restrict to
$\ran(S+T)$. Every strict convex combination of their restrictions is
positive definite there. The same negative second derivative on that
open segment, and continuity at the endpoints, prove strict concavity.
\end{proof}

\begin{lemma}[Fidelity, including its variational and SDP formulas]
\label{sq:fo:fidelity}
For $S,M\succeq0$ define
\[
 F(S,M)=\Tr\sqrt{S^{1/2}MS^{1/2}}.
\]
This continuous function is symmetric, jointly concave, order preserving
in both arguments, and homogeneous in the form $F(aS,aM)=aF(S,M)$ for
$a\ge0$. It obeys
\begin{equation}\label{sq:fo:fidelitytrace}
 0\le F(S,M)\le\sqrt{\Tr S\,\Tr M}.
\end{equation}
In the second formula below, $W$ ranges over all real matrices of the
same size as $S$ and need not be symmetric. The two useful variational
formulas are
\begin{align}
 2F(S,M)&=\inf_{Z\succ0}\{\Tr(SZ^{-1})+\Tr(MZ)\},
 \label{sq:fo:fidinf}\\
 F(S,M)&=\max_W\left\{\Tr W:
 \begin{pmatrix}S&W\\W^T&M\end{pmatrix}\succeq0\right\}.
 \label{sq:fo:fidsdp}
\end{align}
When $S,M\succ0$, the infimum in \eqref{sq:fo:fidinf} has a unique minimizer
$Z\succ0$, characterized by $ZMZ=S$.
If $M$ is supported on a subspace $K$ and $\pi S$ denotes compression to
$K$, then
\begin{equation}\label{sq:fo:fidcompression}
 F(S,M)=F(\pi S,M|_K).
\end{equation}
\end{lemma}
\begin{proof}
The nonzero eigenvalues, with multiplicities, of $XX^T$ and $X^TX$
coincide: multiplication by $X^T$ gives inverse isomorphisms on each
positive eigenspace after division by its eigenvalue. Apply this to
$X=S^{1/2}M^{1/2}$ to obtain symmetry and
\[
 F(S,M)=\|S^{1/2}M^{1/2}\|_1,
 \qquad \|X\|_1:=\Tr\sqrt{X^TX}.
\]
For an arbitrary real matrix $X$, the singular-value decomposition gives
\[\|X\|_1=\max_{\|K\|_{\op}\le1}\Tr(XK).\]
Here $K$ ranges over all real matrices of the required size.
In singular coordinates each diagonal entry of a real contraction is at most one, and an orthogonal matrix
from the singular-value decomposition attains the sum of singular values.
Hilbert--Schmidt Cauchy--Schwarz then gives
\[
 |\Tr(S^{1/2}M^{1/2}K)|
 \le\|S^{1/2}\|_{\HS}\|M^{1/2}K\|_{\HS}
 \le\sqrt{\Tr S\,\Tr M},
\]
proving \eqref{sq:fo:fidelitytrace}. Continuity follows from continuity of
the scalar square root under spectral functional calculus (equivalently,
uniform polynomial approximation on a bounded spectral interval).

Suppose first that $S,M\succ0$. Let
$A=M^{1/2}SM^{1/2}$ and $Q=M^{1/2}ZM^{1/2}$. Cyclicity of trace gives
\begin{align*}
 \Tr(SZ^{-1})+\Tr(MZ)-2\Tr\sqrt A
 &=\Tr(AQ^{-1})+\Tr Q-2\Tr\sqrt A\\
 &=\|Q^{-1/2}\sqrt A-Q^{1/2}\|_{\HS}^2\ge0.
\end{align*}
Equality holds exactly at $Q=\sqrt A$, equivalently at the unique
$Z=M^{-1/2}\sqrt A\,M^{-1/2}$, and this is equivalent to $ZMZ=S$.
This proves \eqref{sq:fo:fidinf} in the faithful case. For semidefinite
arguments replace both by $S+\epsilon I,M+\epsilon I$. For each fixed
$Z\succ0$, passage to the limit gives the inequality with $2F(S,M)$ on
the left. Conversely the infimum for $S,M$ is at most the infimum for
the enlarged arguments, since both trace terms increase when an
argument increases. Letting $\epsilon\downarrow0$ proves equality.
The infimum of functions jointly linear in $(S,M)$ is concave: at a
convex combination each such function is at least the corresponding
convex combination of the two infima. The same formula proves
monotonicity because $Z,Z^{-1}\succ0$. Homogeneity is immediate.

For the SDP formula first suppose $S,M\succ0$. Congruence by
$\operatorname{diag}(S^{-1/2},M^{-1/2})$ shows that the block constraint
is equivalent to $W=S^{1/2}KM^{1/2}$ with $\|K\|_{\op}\le1$.
Indeed the resulting block matrix is
$\left(\begin{smallmatrix}I&K\\K^T&I\end{smallmatrix}\right)$;
the Schur complement says it is positive semidefinite exactly when
$I-K^TK\succeq0$; see \cite[Appendix~A.5.5]{boydvandenberghe2004}.
Maximizing its trace objective by the trace-norm
formula proves \eqref{sq:fo:fidsdp}.
For singular $S,M$, positivity of the block matrix implies
$|u^TWv|^2\le(u^TSu)(v^TMv)$: restrict its quadratic form to the two
vectors $(u,0),(0,v)$ and use the nonnegative determinant of the resulting
$2$-by-$2$ matrix. Thus $W$ vanishes on the relevant kernels and the
same contraction factorization follows by restriction to the supports.
Conversely any such contraction produces a positive block matrix by
congruence. The maximization proof therefore still applies. The feasible
set is compact by this factorization, so the maximum is attained.

Finally, symmetry writes the fidelity as
$\Tr\sqrt{M^{1/2}SM^{1/2}}$. If $M$ is supported on $K$, this matrix is
supported on $K$ and depends only on $\pi S$, proving
\eqref{sq:fo:fidcompression}.
\end{proof}

For the remainder of the square argument, $k$ denotes any retained
label count and $0\preceq C\preceq I_k$. The source and potential are
\eqref{sq:fo:eta} and \eqref{sq:fo:energy}. Factoring $C$ gives the
symmetric Kraus matrices
\[
 B'_a=\sum_i(C^{1/2})_{ia}\A_i,
 \qquad \eta_C(X)=\sum_aB'_aXB'_a.
\]
Thus the source is completely positive: every block amplification is
again a sum of matrix congruences and preserves positive semidefiniteness.
It is also self-adjoint for the Hilbert--Schmidt pairing. The budget
\eqref{sq:fo:sourcebudget} follows by factoring $I-C$ in the same way.
We retain a general $\theta>0$ in the analytic formulas and set
$\theta=1$ when proving the square theorem.

\begin{proposition}[The density semidefinite program]\label{sq:fo:sdp}
For fixed $H,C$, \eqref{sq:fo:energy} is the value of the finite-dimensional
semidefinite program
\begin{align*}
 \text{maximize}\quad &\Tr(HS)+2\Tr W+2\theta\Tr T,\\
 \text{subject to}\quad
 &S,T\succeq0,\quad\Tr S=1,\\
 &\begin{pmatrix}S&W\\W^T&\eta_C(S)\end{pmatrix}\succeq0,
 \qquad\begin{pmatrix}S&T\\T&I\end{pmatrix}\succeq0.
\end{align*}
Here $W$ is an arbitrary real matrix and $S,T$ are real symmetric.
\end{proposition}
\begin{proof}
Lemma~\ref{sq:fo:fidelity} gives the maximum over $W$ for each $S$.
The second block constraint is equivalent by block elimination to
$S-T^2\succeq0$. For $T\succeq0$, square-root monotonicity gives
$T=\sqrt{T^2}\preceq\sqrt S$, so its largest trace is $\Tr\sqrt S$,
attained at $T=\sqrt S$. The two auxiliary optimizations are independent
once $S$ is fixed, proving the assertion. Every displayed constraint is
affine in real variables for fixed $C$, and the objective is linear.
Thus this is an SDP in the standard sense of
\cite[Section~4.6.2]{boydvandenberghe2004}.
Section~\ref{solver:main} supplies the equivalent normalized formulation
used for certified optimization, including singular physical sources.
\end{proof}

\begin{lemma}[The maximizing density]\label{sq:fo:optimizer}
For every real symmetric $H$ and every real covariance $C\succeq0$,
the maximum in \eqref{sq:fo:energy} is attained by a unique density
$S(H,C)\succ0$. Both its value and its maximizing density are continuous
in $(H,C)$.
\end{lemma}
\begin{proof}
The density set is compact and the objective is continuous. The fidelity
term is concave in $S$ because both its arguments are linear in $S$ and
fidelity is jointly concave. The square-root trace is strictly concave
by Lemma~\ref{sq:fo:sqrt}, so their sum with the linear term is strictly
concave. This proves existence and uniqueness.

Suppose a maximizing $S$ has a zero eigenvalue. For $0<\epsilon<1$, put
$S_\epsilon=(1-\epsilon)S+\epsilon I/D$. These matrices commute with $S$.
Every zero eigenvalue contributes $\sqrt{\epsilon/D}$ to
$\Tr\sqrt{S_\epsilon}$; the change from the finitely many positive
eigenvalues is bounded in absolute value by a constant times $\epsilon$.
The linear term changes by at most a constant times $\epsilon$.
Joint concavity gives
\[
 F(S_\epsilon,\eta_C(S_\epsilon))
 \ge(1-\epsilon)F(S,\eta_C(S))
       +\epsilon F(I/D,\eta_C(I/D)),
\]
so the loss of the fidelity term is also at most a constant times
$\epsilon$. The positive order-$\sqrt\epsilon$ gain from the regularizer
therefore contradicts optimality for sufficiently small $\epsilon$.
Thus $S\succ0$.

For continuity, along any convergent parameter sequence the objective
converges uniformly on the compact density set, by continuity on the
compact set consisting of the sequence and its limit times the density
set. The corresponding maximum values converge. Every subsequential
limit of the optimizing densities maximizes the limiting objective;
uniqueness forces that limit to be $S(H,C)$. Compactness then gives
convergence of the entire sequence.
\end{proof}

\begin{lemma}[Elementary barrier bounds]\label{sq:fo:barrier}
If $0\preceq C\preceq I_k$ is a covariance on $k$ retained labels, then
\begin{equation}\label{sq:fo:excess}
 0\le E(H,C)-f(H)\le2\sqrt k.
\end{equation}
Both $E(\cdot,C)$ and $f$ are convex and $1$-Lipschitz for operator norm.
The value $E(H,C)$ is order preserving in $C$. In square mode
($\theta=1$ and $D\le2n$), the lifted pencil satisfies
\begin{equation}\label{sq:fo:initialbound}
 E(H(x),C)\ge\|H(x)\|_{\op},\qquad
 E(0,I_n)\le5\sqrt n.
\end{equation}
At the same center $H$, discarding any current covariance and installing the
identity covariance on $\ell$ live labels increases the energy by at most
$2\sqrt \ell$.
\end{lemma}
\begin{proof}
For every density $S$, self-adjointness and \eqref{sq:fo:sourcebudget} give
\[
 \Tr\eta_C(S)=\Tr(S\eta_C(I))\le k.
\]
Thus Lemma~\ref{sq:fo:fidelity} bounds the fidelity term between $0$ and
$\sqrt k$, proving \eqref{sq:fo:excess} by maximizing. A supremum of affine
functions of $H$ is convex. For any density,
$|\Tr((H_1-H_2)S)|\le\|H_1-H_2\|_{\op}$, by diagonalizing $S$ into a
convex combination of rank-one projections. Taking maxima proves
the Lipschitz statements. If $C_1\preceq C_2$, factoring $C_2-C_1$ shows
$\eta_{C_1}(S)\preceq\eta_{C_2}(S)$ for each $S\succeq0$.
Fidelity monotonicity proves energy monotonicity.

All regularizer terms are nonnegative, so a rank-one density on a top
eigenvector gives $E(H,C)\ge\lambda_{\max}(H)$. The two-sided encoding in the introduction
proves the first assertion of \eqref{sq:fo:initialbound}. At $H=0$, the
fidelity bound and scalar Cauchy--Schwarz applied to the eigenvalues of
$S$ give
\[
 E(0,I_n)\le2\sqrt n+2\sqrt{2n}<5\sqrt n.
\]
For the final assertion, the old energy is at least $f(H)$, while the new
one is at most $f(H)+2\sqrt \ell$, by \eqref{sq:fo:excess}.
\end{proof}

\section{The transport form of the spectral potential}
\label{origin:section}

Lehner's expression puts the free spectral edge in the form of an
optimization over positive matrices. For our potential, the analogous
identity follows from a saddle point between the density and the
fidelity transport. We prove that identity directly, including the case
in which the source occupies only a proper physical subspace. This
justifies the variational interpretation in Section~\ref{sec:motivation}
and provides the equations to be differentiated in the next section.

Let $\mathcal D_D=\{S\succeq0:\Tr S=1\}$ be the real density set,
and let $\mathfrak r$ be a continuous concave regularizer on this set,
with a differentiable extension to positive definite matrices.  Define
\begin{equation}\label{origin:definitions}
 f(Y)=\max_{S\in\mathcal D_D}\{\Tr(YS)+\mathfrak r(S)\},\qquad
 E(H,C)=\max_{S\in\mathcal D_D}
 \{\Tr(HS)+2F(S,\eta_C(S))+\mathfrak r(S)\}.
\end{equation}
The regularizers used in this paper are
$\mathfrak r(S)=2\Tr\sqrt S$ for the square argument and
$\mathfrak r(S)=\theta\Tr S^{1-q}/(1-q)+2\kappa\Tr\sqrt S$
for the rectangular implementation.  In both cases the maximum defining
$E$ has a faithful density, as proved for the corresponding profile.

\begin{theorem}[The regularized transport variational formula]
\label{origin:formula}
Suppose $E(H,C)$ in \eqref{origin:definitions} has a positive definite
maximizing density.  Then
\begin{equation}\label{origin:full}
 \boxed{\displaystyle
 E(H,C)=\inf_{Z\succ0}f\bigl(H+Z^{-1}+\eta_C(Z)\bigr).}
\end{equation}
To state the attained version when the source is singular, let
$K=\operatorname{ran}\eta_C(I_D)$, let $\pi$ denote compression to $K$,
and let $\pi^*$ extend a matrix on $K$ by zero on $K^\perp$.  If
$K\ne\{0\}$, then
\begin{equation}\label{origin:supported}
 E(H,C)=\min_{Z\succ0\text{ on }K}
 f\bigl(H+\pi^*(Z^{-1})+\eta_C(\pi^*Z)\bigr).
\end{equation}
The density in $f$ remains a full $D$-dimensional density.
If $K=\{0\}$, the supported formula is interpreted as $E(H,C)=f(H)$.
\end{theorem}
\begin{proof}
First, the fidelity variational inequality gives, for every full
$Z\succ0$ and density $S$,
\[
 2F(S,\eta_C(S))\le\Tr(SZ^{-1})+\Tr(\eta_C(S)Z).
\]
The source map is self-adjoint for the trace inner product, so maximizing
this inequality over $S$ proves that the right side of
\eqref{origin:full} is at least $E(H,C)$.

We construct a saddle point on the physical source support.  A real symmetric
Kraus factorization of $\eta_C$ shows that each Kraus matrix is supported
on $K$: the kernel of their sum of squares is the intersection of their
kernels.  Thus $\eta_C$ depends only on compression to $K$ and vanishes
on matrices supported on $K^\perp$.  Let $S_*$ be a faithful maximizer,
put $S_0=\pi S_*$ and $M=\eta_C(S_*)|_K$, and first assume $K\ne\{0\}$.
Both $S_0$ and $M$ are positive definite.  The fidelity formula supplies
a unique $Z_*\succ0$ on $K$ satisfying
\begin{equation}\label{origin:transport}
 Z_*MZ_*=S_0.
\end{equation}
Set $B_* =\pi^*(Z_*^{-1})+\eta_C(\pi^*Z_*)$.

The density derivative of $2F(S,\eta_C(S))$ at $S_*$ is $B_*$.
To verify this derivative, differentiate
$\Tr((\pi S)Z^{-1})+\Tr((\eta_C(S)|_K)Z)$ at its minimizing
transport.  Its derivative with respect to $Z$ is
$M-Z_*^{-1}S_0Z_*^{-1}=0$ by \eqref{origin:transport}, so the
transport derivative contributes zero.  The remaining two trace terms,
using self-adjointness of $\eta_C$, give $\Tr(B_*\,\Delta S)$.
The explicit positive square-root formula for the minimizing transport
justifies this differentiation on the fixed support.

Since $S_*$ is interior in the density set, its trace-constrained
stationarity is therefore
\[
 H+B_*+\nabla\mathfrak r(S_*)=\lambda I_D.
\]
The function $S\mapsto\Tr((H+B_*)S)+\mathfrak r(S)$ is concave.
Its tangent derivative vanishes at $S_*$, so the supporting-plane
inequality makes $S_*$ a global maximizer, including over boundary
densities by continuity.  Fidelity equality at $Z_*$ now gives
$f(H+B_*)=E(H,C)$.  The supported fidelity inequality supplies the
opposite inequality for every supported $Z$, proving
\eqref{origin:supported}.

Finally, let $P_{K^\perp}$ be the orthogonal projection onto $K^\perp$
and put $Z_t=\pi^*Z_*+tP_{K^\perp}$ for $t>0$.
It is positive definite on the full space, and
\[
 H+Z_t^{-1}+\eta_C(Z_t)=H+B_*+t^{-1}P_{K^\perp}.
\]
Every density has $0\le\Tr(SP_{K^\perp})\le1$, whence
\[
 E(H,C)\le f(H+Z_t^{-1}+\eta_C(Z_t))\le E(H,C)+t^{-1}.
\]
Let $t\to\infty$ to prove \eqref{origin:full}.
If $K=\{0\}$, then $\eta_C=0$ and the same conclusion follows by
using $Z_t=tI_D$.  If $K$ is the full space, $Z_*$ itself attains
\eqref{origin:full}.
\end{proof}

The free contribution is therefore encoded by the transport optimization.
In particular, reducing covariance changes the same objective whose
center derivative controls movement. At a faithful source and optimizing
pair $(S,Z)$, a coefficient direction $a\in\ran C$ gives
\[
 D_CE[aa^T]=\Tr(S\A(a)Z\A(a))
          =\|Z^{1/2}\A(a)S^{1/2}\|_{\HS}^2.
\]
The envelope calculation behind this formula is included in
Lemma~\ref{sq:an:hessian}. It exhibits the interaction omitted by the
single variance matrix $\eta_C(I)$: covariance is measured through the
optimizing density and transport, in precisely the direction that moves
the center.

For comparison, the unregularized expression has the SDP formulation
\[
 \inf_{Z\succ0}\lambda_{\max}(H+Z^{-1}+\eta_C(Z))
 =\inf_{t,Z}\left\{t:
 \begin{pmatrix}tI_D-H-\eta_C(Z)&I_D\\I_D&Z\end{pmatrix}
 \succeq0\right\},
\]
where $t\in\R$ and $Z$ is real symmetric. The block constraint implies
$Z\succeq0$; if $Zv=0$, positivity forces the off-diagonal block to
annihilate $v$, so $I_Dv=0$. Thus $Z\succ0$, and its Schur complement
is precisely the eigenvalue bound.
This uses standard SDP and Schur-complement machinery
\cite[Section~4.6.2 and Appendix~A.5.5]{boydvandenberghe2004}.
For the square-root and dyadic-power regularizers used here,
Proposition~\ref{solver:program} derives the additional block constraints
that give an exact SDP for the regularized value.  The general variational
identity alone does not assert SDP representability for every regularizer.

\section{The optimizer response and its stability}
\label{sq:an:section}

Moving the center changes the density that tests it. The corresponding
second derivative is governed by the inverse of the density curvature,
including the change transmitted through the optimizing transport.
Our objective in this section is to put that inverse response in a form
that can be estimated on the coefficient directions of the walk.

We retain the physical dimension $D$, the real coefficient space, the
lifted contractions $\A_i$, and the positive regularizer strength $\theta$.
Write $\mathcal H_D$ for the real symmetric $D$-by-$D$ matrices with
inner product $\langle X,Y\rangle=\Tr(XY)$. This is the space of density
variations. Left-right multiplication and balanced frame operators also
act on the full real matrix space, with inner product $\Tr(X^TY)$.
Whenever a density curvature is used on that larger space, it denotes
the positive self-adjoint extension given by its matrix-unit multiplier;
its restriction to $\mathcal H_D$ is the actual density Hessian.
The displayed operators preserve the symmetric subspace whenever a
response on that subspace is asserted. Adjoint notation $*$ refers to
the relevant real Hilbert-space inner product.

\subsection{A smooth optimizer on each covariance support}

Differentiation requires a domain on which the matrix square roots have
fixed support. A \emph{coefficient face} means that the range of $C$ is
fixed, while its positive eigenvalues may vary. Even on such a face,
$\eta_C(S)$ can have a kernel in physical space. The next lemma shows
that this kernel is fixed as well. It therefore permits differentiation on that fixed support, while
the density remains a full $D$-dimensional matrix. The lower bound for $S$ will later make numerical tolerances
uniform over all states of the walk.

\begin{lemma}[Quantitative faithfulness and smoothness on a coefficient face]
\label{sq:an:smooth}
Let \(V\subseteq\mathbb R^k\) be a fixed coefficient subspace, and let
\(C\) be positive definite on \(V\) and zero on \(V^\perp\).
The maximizer \(S=S(H,C)\) defining \(E(H,C)\) is unique, is positive
definite, and depends smoothly on \(H\) and on \(C\) in this relative
open coefficient face. This assertion allows
\(\eta_C(S)\) to be singular on the full physical space.

Writing \(v=\|\eta_C(I)\|_{\op}\), one has the quantitative bound
\begin{equation}
 S(H,C)\succeq
 \frac{\theta^2}
 {(2\|H\|_{\op}+2\sqrt v+\theta\sqrt D)^2}\,I.
 \label{sq:an:density-floor}
\end{equation}
\end{lemma}

\begin{proof}
Existence, uniqueness, faithfulness, and continuity of the optimizer
are given by Lemma~\ref{sq:fo:optimizer}. We supply the source-support,
quantitative, and differential assertions.

Write \(\A(z)=\sum_i z_i\A_i\), and set
\[
 K(V)=\operatorname{span}\{\ran \A(z):z\in V\}.
\]
For a covariance \(C\) as in the statement, a real square-root
factorization gives real symmetric matrices \(B'_a\) spanning the same
linear space as the matrices \(\A(z)\), \(z\in V\), with
\(\eta_C(X)=\sum_a B'_aXB'_a\).
For every \(S>0\),
\[
 \lambda_{\min}(S)\eta_C(I)
 \preceq\eta_C(S)\preceq\|S\|_{\op}\eta_C(I).
\]
The kernel of \(\eta_C(I)=\sum_a(B'_a)^2\) is the common kernel of
the \(B'_a\). Since these matrices are real symmetric, its orthogonal
complement is \(K(V)\). Thus the physical support of \(\eta_C(S)\)
is the fixed space \(K(V)\). Let \(\Pi\) be the projection onto this
space. Every \(B'_a\) is supported there, so
\(\eta_C(S)=\eta_C(\Pi S\Pi)\). The compression identity in
Lemma~\ref{sq:fo:fidelity} yields
\[
 F(S,\eta_C(S))
 =F\bigl(\Pi S\Pi,\eta_C(\Pi S\Pi)\bigr),
\]
where the two arguments on the right are positive definite on \(K(V)\).
Consequently the free term is smooth in \(S>0\) and \(C\) on this
coefficient face. If \(K(V)=\{0\}\), it is identically zero and the
same conclusion holds.

Let
\[
 G=\nabla_S\bigl[2F(S,\eta_C(S))\bigr].
\]
Monotonicity of fidelity in each argument, together with positivity of
\(\eta_C\), implies \(G\succeq0\): its pairing with every positive
matrix direction is nonnegative. The free term is homogeneous of
degree one in \(S\), so differentiation of its scalar scaling gives
\[
 \Tr(SG)=2F(S,\eta_C(S)).
 \tag{*}
\]
The density is an interior maximizer on its trace hyperplane.
Its stationarity equation, with a real trace multiplier $\lambda$, is therefore
\[
 H+G+\theta S^{-1/2}=\lambda I.
 \]
Pairing this equation with \(S\), and using \((*)\), gives
\[
 \lambda
 =\Tr(HS)+2F(S,\eta_C(S))+\theta\Tr\sqrt S.
 \]
Since \(\Tr S=1\), self-adjointness of \(\eta_C\) gives
\(\Tr\eta_C(S)=\Tr(S\eta_C(I))\le v\). The trace bound for fidelity
and the trace square-root bound from the foundational lemmas imply
\[
 \lambda\le\|H\|_{\op}+2\sqrt v+\theta\sqrt D.
 \]
As \(G\succeq0\), stationarity also gives
\[
 \theta S^{-1/2}\preceq\lambda I-H
 \preceq(\lambda+\|H\|_{\op})I.
 \]
This proves \eqref{sq:an:density-floor}.

For smooth dependence, the negative Hessian of the Tsallis term is
strictly positive. Explicitly, in an eigenbasis of \(S\) with
eigenvalues \(s_i>0\), the formula in Lemma~\ref{sq:fo:sqrt} gives
\begin{equation}
 \left\langle X,
       -D_S^2(2\theta\Tr\sqrt S)\,X\right\rangle
 =\theta\sum_{i,j}
       \frac{|X_{ij}|^2}
       {\sqrt{s_i s_j}(\sqrt{s_i}+\sqrt{s_j})}>0
 \quad (X\ne0).
 \label{sq:an:tsallis-positive}
\end{equation}
The free term is concave, so its negative Hessian is positive
semidefinite. The full negative density Hessian is consequently
strictly positive, also after restriction to the trace-zero tangent.
The implicit-function theorem applied to the density stationarity
equations on that tangent proves smooth dependence on \(H,C\).

\end{proof}

\subsection{Linearizing the density and transport equations}

The potential is defined by a nonlinear optimization, but its local
response is determined by a coupled linear system. We write that system
before eliminating either optimizer. Throughout the differentiation,
$\theta$ and the coefficient support are fixed. We first assume that the
physical source is positive definite; the corresponding equations on a
proper physical support are given at the end of the subsection.

Put $M=\eta_C(S)$ and $W=Z^{-1}$. At the maximizing density and its
minimizing fidelity transport, the stationarity equations are
\begin{equation}\label{sq:an:nonlinear-stationarity}
 \boxed{\quad
 \begin{aligned}
  \eta_C(S)&=WSW,\\
  H+W+\eta_C(Z)+\theta S^{-1/2}&=\lambda I,\\
  \Tr S&=1.
 \end{aligned}\quad}
\end{equation}
The first equation is transport stationarity, equivalently $ZMZ=S$.
The second is density stationarity: the multiplier $\lambda$ enforces
its trace constraint. At the positive density and transport, these
stationarity and normalization equations are the KKT conditions of the
regularized saddle-point problem. The next lemma differentiates them
to obtain the linearized KKT system for the response of $(S,Z,\lambda)$
as the center and covariance change.

This is the self-consistent system whose stability matters for the walk.
The method of linearizing such a system has a parallel in matrix Dyson
equation analysis \cite{aek2019}; at a regular spectral edge, stability
also requires isolating the response that matters rather than applying
an undifferentiated inverse bound \cite[Section 4]{aeks2020}.
Here the matrix-valued term $\theta S^{-1/2}$ and the trace constraint
are part of the equations. We differentiate the entire optimizing triple
$(S,Z,\lambda)$ and estimate the quadratic forms generated by our
coefficient matrices. All estimates below are derived for this system.

We now define the derivatives used by the movement covariance.
The center derivative $J$ is defined for all real coefficient directions;
the covariance derivative $\Gamma$ is defined on the current coefficient
support. Put
\begin{equation}
 \begin{aligned}
 J(u,w)&=\tfrac12D_H^2E(H,C)[\A(u),\A(w)],\\
 \Gamma(u,w)&=D_CE(H,C)
       \left[\frac{uw^T+wu^T}{2}\right].
 \end{aligned}
 \label{sq:an:response-definitions}
\end{equation}
Lemma~\ref{sq:an:smooth} makes these derivatives well-defined. The first
formula is used for every $u,w\in\mathbb R^k$ and the second for
$u,w\in\ran C$. For coefficient trace expressions we extend $\Gamma$
by zero on $(\ran C)^\perp$; the extension does not affect a trace
against $C$ or any $0\preceq Q\preceq C$. For physical matrix
directions we write $D_H^2E[B,B]$ explicitly, reserving $J$ for the
coefficient form.

\begin{lemma}[Linearized KKT system and exact response]
\label{sq:an:hessian}
At a positive density \(S\), suppose initially that
\(M=\eta_C(S)>0\) on the full physical space. Let \(Z>0\) solve
\(ZMZ=S\), and write
\[
 \begin{aligned}
 W&=Z^{-1},& \mathcal R(X)&=WXW,\\
 \mathcal K(U)&=WUM+MUW,&\mathcal L&=\mathcal R-\eta_C.
 \end{aligned}
\]
For an affine parameter path $H(t)=H+tB$, $C(t)=C+t\Delta$
within the coefficient face, suppose $S$ is the maximizing density and
write $X=\dot S(0)$ and $U=\dot Z(0)$. Here $\Delta$ is real
symmetric and supported on $\ran C$, and $t$ is small enough that
$C(t)$ remains positive definite on that support.
Let $A_\theta^{\rm original}$ denote the negative Hessian of
$2\theta\Tr\sqrt S$. The derivatives solve
\begin{equation}\label{sq:an:coupled-linearization}
 \begin{aligned}
  \mathcal K U&=\mathcal L X-\eta_\Delta(S),\\
  A_\theta^{\rm original}X+\mathcal L^*U+\dot\lambda I
       &=B+\eta_\Delta(Z),\\
  \Tr X&=0.
 \end{aligned}
\end{equation}
All operators in this system are evaluated at the original state.
In particular, its forcing includes both the center velocity $B$ and
the covariance velocity $\Delta$.

The negative Hessian of the free term
\(S\mapsto2F(S,\eta_C(S))\) is
\begin{equation}
 A_f^{\rm original}=\mathcal L^*\mathcal K^{-1}\mathcal L.
 \label{sq:an:free-original}
\end{equation}
If the second fidelity argument is held fixed at \(M\), its negative
Hessian with respect to the first argument is
\begin{equation}
 A_{\rm ref}=\mathcal R^*\mathcal K^{-1}\mathcal R,\qquad
 A_{\rm ref}^{-1}(B)=ZBS+SBZ.
 \label{sq:an:reference-inverse}
\end{equation}
The inverse negative Hessian of \(2\theta\Tr\sqrt S\) is
\begin{equation}
 \bigl(A_\theta^{\rm original}\bigr)^{-1}(B)
 =\frac1\theta(SB\sqrt S+\sqrt S BS).
 \label{sq:an:tsallis-inverse}
\end{equation}

At the maximizing density, let \(\mathcal V_0\) be the real
space of trace-zero real symmetric matrices, let
\(\iota:\mathcal V_0\longrightarrow\mathcal H_D\) be its inclusion,
and let $B_0=\iota^*B=B-(\Tr B/D)I$. Then
\begin{equation}
 \tfrac12D_H^2E(H,C)[B,B]=\frac12
 \left\langle B_0,
 \{\iota^*(A_f^{\rm original}+A_\theta^{\rm original})\iota\}^{-1}
 B_0\right\rangle.
 \label{sq:an:optimized-center}
\end{equation}
In particular this is at most the corresponding unconstrained
quadratic form
\(\frac12\langle B,
 (A_f^{\rm original}+A_\theta^{\rm original})^{-1}B\rangle\).
The covariance derivative is
\begin{equation}
 \Gamma(u,w)
 =\Tr\bigl(S\A(u)Z\A(w)\bigr).
 \label{sq:an:gamma-original}
\end{equation}
\end{lemma}

\begin{proof}
\emph{Differentiating the transport minimization.}
The fidelity variational formula of Lemma~\ref{sq:fo:fidelity} gives
\[
 2F(S,\eta_C(S))
 =\min_{Z>0}g(S,Z),\qquad
 g(S,Z)=\Tr(SZ^{-1})+\Tr(\eta_C(S)Z).
 \]
The minimizer is unique: its stationarity equation is
\(M=WSW\), equivalently \(ZMZ=S\), and hence
\[
 Z=M^{-1/2}(M^{1/2}SM^{1/2})^{1/2}M^{-1/2}.
 \]
Along the affine path, the derivatives of the source and inverse
transport are
\[
 \dot M=\eta_C(X)+\eta_\Delta(S),\qquad \dot W=-WUW.
\]
Thus differentiating the first equation of
\eqref{sq:an:nonlinear-stationarity} gives
\begin{align*}
 \eta_C(X)+\eta_\Delta(S)
 &=(-WUW)SW+WXW+WS(-WUW)\\
 &=WXW-WUM-MUW.
\end{align*}
This is the first equation of \eqref{sq:an:coupled-linearization}.
For the second equation, the four density-gradient terms differentiate
as follows:
\[
 \dot H=B,\qquad \dot W=-WUW,\qquad
 \frac{d}{dt}\eta_{C(t)}(Z(t))=\eta_C(U)+\eta_\Delta(Z),\qquad
 \frac{d}{dt}(\theta S^{-1/2})=-A_\theta^{\rm original}X.
\]
Move the terms containing $X,U$ to the left and differentiate
$\Tr S=1$. This proves the remaining equations of
\eqref{sq:an:coupled-linearization}. Here $\mathcal L^*=\mathcal L$,
since both $\mathcal R$ and $\eta_C$ are self-adjoint; the adjoint
notation also records the two roles of this operator in elimination.

The transport Hessian is $D_Z^2g=\mathcal K$. It is self-adjoint,
and for nonzero real symmetric $U$,
\[
 \langle U,\mathcal K(U)\rangle
 =2\Tr(UWUM)
 =2\|W^{1/2}UM^{1/2}\|_{\HS}^2>0.
\]
It therefore has a genuine inverse. To compute the density Hessian
of the free term, vary $S$ at fixed $C$. The first linearized equation
then reads $\mathcal K(D_SZ[X])=\mathcal L X$. The free density
gradient is $W+\eta_C(Z)$, whose derivative is
$-\mathcal L^*D_SZ[X]$. Substitution proves
\eqref{sq:an:free-original}. The transport response is exactly what
produces this quadratic form.

For a fixed second argument \(M\), the mixed derivative is
\(-\mathcal R\) instead, which gives the first identity in
\eqref{sq:an:reference-inverse}. As \(\mathcal R^{-1}(B)=ZBZ\),
its inverse is \(\mathcal R^{-1}\mathcal K\mathcal R^{-1}\).
Substituting \(ZBZ\) into \(\mathcal K\) and then conjugating by \(Z\)
gives
\[
 \mathcal R^{-1}\mathcal K\mathcal R^{-1}(B)
 =ZBZ\,MZ+ZMZ\,BZ=ZBS+SBZ.
 \]

\smallskip\noindent
\emph{The regularizer curvature.}
For \eqref{sq:an:tsallis-inverse}, let \(T_S=\sqrt S\).
Differentiating \(T_S^2=S\) gives
\[
 D_ST_S[X]=(L_{T_S}+R_{T_S})^{-1}(X).
 \]
The gradient of \(2\theta\Tr\sqrt S\) is \(\theta T_S^{-1}\).
Its negative derivative is therefore
\[
 A_\theta^{\rm original}(X)
 =\theta T_S^{-1}
       (L_{T_S}+R_{T_S})^{-1}(X)T_S^{-1}.
 \]
In a \(T_S\)-eigenbasis the displayed factors commute as operators
on matrix entries. Inverting them gives
\[
 (A_\theta^{\rm original})^{-1}(B)
 =\theta^{-1}
       (L_{T_S}+R_{T_S})(T_SBT_S),
 \]
which is exactly \eqref{sq:an:tsallis-inverse}.

\smallskip\noindent
\emph{Differentiating the maximizing density.}
Solve the first linearized equation for $U$ and substitute it into
the second. Write $A=A_f^{\rm original}+A_\theta^{\rm original}$
and define the effective density forcing
\begin{equation}\label{sq:an:effective-forcing}
 B_{\rm eff}=B+\eta_\Delta(Z)
                +\mathcal L^*\mathcal K^{-1}\eta_\Delta(S),
 \qquad B_{{\rm eff},0}=\iota^*B_{\rm eff}.
\end{equation}
Every term involving $X$ remains on the left:
\[
 (A_\theta^{\rm original}+\mathcal L^*\mathcal K^{-1}\mathcal L)X
       +\dot\lambda I
 =B+\eta_\Delta(Z)+\mathcal L^*\mathcal K^{-1}\eta_\Delta(S)
 =B_{\rm eff}.
\]
The covariance forcing has two contributions: the explicit change
$\eta_\Delta(Z)$ of the density gradient, and the change transmitted
through the transport equation. Neither is present if $C$ is fixed.
Since $\Tr X=0$, write $X=\iota X_0$ and apply $\iota^*$.
The multiplier disappears because $\iota^*I=0$, yielding
\[
 (\iota^*A\iota)X_0=B_{{\rm eff},0}.
\]
The regularizer makes this compressed operator positive definite.
Consequently the derivatives of both optimizers are
\begin{equation}\label{sq:an:eliminated-linearization}
 \begin{aligned}
  X&=\iota(\iota^*A\iota)^{-1}B_{{\rm eff},0},\\
  U&=\mathcal K^{-1}\{\mathcal L X-\eta_\Delta(S)\}.
 \end{aligned}
\end{equation}
The remaining multiplier is recovered as
$\dot\lambda=D^{-1}\Tr(B_{\rm eff}-AX)$.

For a pure center perturbation $\Delta=0$, one has $B_{\rm eff}=B$.
The envelope identity $D_HE[B]=\Tr(SB)$ then gives
$D_H^2E[B,B]=\Tr(XB)$ and proves \eqref{sq:an:optimized-center}.
To check the envelope identity directly, differentiate the objective
at its optimizing density and transport. The density chain-rule term
is $\lambda\Tr X=0$, and the transport chain-rule term vanishes by
transport stationarity. Only the explicit center derivative remains.
The comparison with the unconstrained inverse follows from the
variational inverse formula
\[
 \langle B,A^{-1}B\rangle
 =\sup_X\{2\langle B,X\rangle-\langle X,AX\rangle\}:
 \]
restricting this supremum to trace-zero \(X\) can only decrease it.

\smallskip\noindent
\emph{The covariance derivative.}
Finally, when \(C\) varies within its coefficient face, the same
chain rule eliminates the density derivative. In the fidelity
minimization the transport derivative has zero coefficient because
$\nabla_Zg=0$ at its minimizer. The only remaining explicit covariance
derivative is therefore
\[
 D_CE[\dot C]=\Tr\bigl(Z\eta_{\dot C}(S)\bigr).
 \]
Substituting
\(\dot C=(uw^T+wu^T)/2\) proves
\eqref{sq:an:gamma-original}. In particular
\[
 \Gamma(u,u)=\Tr(S\A(u)Z\A(u))
 =\|Z^{1/2}\A(u)S^{1/2}\|_{\HS}^2\ge0.
 \]
For the general affine path, the same envelope calculation gives
\[
 \frac{d}{dt}E(H(t),C(t))
 =\Tr(S(t)B)+\Tr\bigl(Z(t)\eta_\Delta(S(t))\bigr).
\]
Differentiating once more at zero and using self-adjointness of
$\eta_\Delta$ gives
\begin{align*}
 E''(0)
 &=\langle B+\eta_\Delta(Z),X\rangle
                         +\langle\eta_\Delta(S),U\rangle\\
 &=\langle B_{\rm eff},X\rangle
            -\langle\eta_\Delta(S),\mathcal K^{-1}\eta_\Delta(S)\rangle.
\end{align*}
Substituting the formula for $X$ gives the full affine-path response:
\begin{equation}\label{sq:an:joint-affine-response}
 \left.\frac{d^2}{dt^2}E(H+tB,C+t\Delta)\right|_{t=0}
 =\langle B_{{\rm eff},0},(\iota^*A\iota)^{-1}B_{{\rm eff},0}\rangle
   -\langle\eta_\Delta(S),\mathcal K^{-1}\eta_\Delta(S)\rangle.
\end{equation}
The positive term is the response of
the maximizing density. The negative term comes from minimizing over
the transport. The square-root regularizer remains in $A$ throughout.
\end{proof}

\paragraph{The equations when the physical source is singular.}
Let $K$ be the fixed physical support supplied by Lemma~\ref{sq:an:smooth},
let $\pi$ denote compression to $K$, and let $\pi^*$ extend by zero.
The density $S$ is still a full trace-one matrix. Set $S_0=\pi S$,
and take $M,Z,W,\mathcal K,\mathcal L$ on $K$, using $S_0$ in
place of $S$. On this support the two transport arguments are positive
definite. The nonlinear density equation on the full space is
\[
 H+\pi^*\{W+\eta_C(Z)\}+\theta S^{-1/2}=\lambda I.
\]
For the same affine parameter path, its exact linearization is
\begin{equation}\label{sq:an:singular-linearization}
 \begin{aligned}
  \mathcal K U&=\mathcal L(\pi X)-\eta_\Delta(S_0),\\
  A_\theta^{\rm original}X+\pi^*\mathcal L^*U+\dot\lambda I
       &=B+\pi^*\eta_\Delta(Z),\\
  \Tr X&=0.
 \end{aligned}
\end{equation}
Indeed the fidelity depends on $S$ only through $\pi S$, whereas
$\theta S^{-1/2}$ is differentiated on the full physical space.
Eliminating $U$ therefore replaces $A$ and $B_{\rm eff}$ above by
\begin{equation}\label{sq:an:singular-effective-forcing}
 \begin{aligned}
  A&=A_\theta^{\rm original}
                   +\pi^*\mathcal L^*\mathcal K^{-1}\mathcal L\pi,\\
  B_{\rm eff}&=B+\pi^*\{\eta_\Delta(Z)
                +\mathcal L^*\mathcal K^{-1}\eta_\Delta(S_0)\}.
 \end{aligned}
\end{equation}
The trace-zero inversion is still on the full density space. In
particular $\Tr(\pi X)$ need not vanish: mass may move between $K$
and its complement. Formula~\eqref{sq:an:joint-affine-response} then
holds with these replacements and with $\eta_\Delta(S_0)$ in the
transport term. If $K=\{0\}$, all source and transport terms vanish,
and only the full regularizer response remains. The later compression
comparison in Proposition~\ref{sq:an:singular-transfer} bounds these exact
responses by a local quadratic form; it does not replace the global
trace constraint during differentiation.

\subsection{Balanced coordinates and the coefficient forces}

The formulas just proved involve the transport $Z$, the density $S$,
and the source $\eta_C(S)$ in asymmetric positions. A congruence by
$Z^{1/2}$ puts the transport equation into the fixed-point form
$\Phi(P)=P$. A second change of coordinates normalizes the reference
curvature. These are changes of coordinates in a quadratic form at a
fixed state; they are not changes to the algorithm's state.

In the next two lemmas, $P$ is the balanced density, $R$ records the
square-root regularizer in the same coordinates, and $D_a$ are the
balanced coefficient matrices. The vectors $F_a$ are these matrices
viewed as forces acting on density perturbations. Their Gram matrix is
exactly the covariance derivative after factoring $C$. This identity
connects the cap tested by the algorithm to the inverse-Hessian
expression that we need to bound.

\begin{lemma}[Exact balanced normal form]
\label{sq:an:balanced}
Factor the real covariance \(C\) and write
\[
 B'_a=\sum_i(C^{1/2})_{ia}\A_i,\qquad
 \eta_C(X)=\sum_aB'_aXB'_a.
 \]
At a density with faithful source, define
\begin{equation}
 \begin{aligned}
 P&=Z^{-1/2}SZ^{-1/2},&R&=Z^{-1/2}\sqrt S Z^{-1/2},\\
 D_a&=Z^{1/2}B'_aZ^{1/2},&\Phi(X)&=\sum_aD_aXD_a.
 \end{aligned}
 \label{sq:an:balanced-definitions}
\end{equation}
Then \(\Phi\) is self-adjoint, positive, and an ordinary
Hilbert--Schmidt contraction, and \(\Phi(P)=P\).
Put \(\mathcal J_P=L_P+R_P\). In the density coordinate
\(Y=Z^{-1/2}(dS)Z^{-1/2}\), the exact negative curvatures are
\begin{equation}
 A_f=(\Id-\Phi)\mathcal J_P^{-1}(\Id-\Phi),\qquad
 A_\theta^{-1}=\theta^{-1}(L_PR_R+L_RR_P).
 \label{sq:an:balanced-hessians}
\end{equation}
For \(Y=\mathcal J_P^{1/2}y\), define
\begin{equation}
 \begin{aligned}
 T&=\mathcal J_P^{-1/2}\Phi\mathcal J_P^{1/2},\\
 H_\theta&=\mathcal J_P^{1/2}A_\theta\mathcal J_P^{1/2},\\
 \mathcal A&=(\Id-T)^*(\Id-T)+H_\theta.
 \end{aligned}
 \label{sq:an:whitened}
\end{equation}
This is the full negative density Hessian in the \(y\) coordinate.
The matrices specifying the linear coordinate changes are frozen at the
current point; the variation of the optimizing transport has already
been included in \eqref{sq:an:free-original}.
\end{lemma}

\begin{proof}
The real symmetric Kraus factors make $\Phi$ completely positive and
self-adjoint on the full real matrix Hilbert space. Here complete positivity
means that every block amplification preserves positive semidefinite
matrices; each amplification is a sum of congruences by the amplified
Kraus matrices, which proves that property. The transport identity gives
\[
 \Phi(P)=Z^{1/2}\eta_C(S)Z^{1/2}=Z^{-1/2}SZ^{-1/2}=P.
\]
First consider a symmetric eigenmatrix $X\ne0$ with eigenvalue $\lambda$.
Faithfulness of $P$ gives a finite $b>0$ such that
$-bP\preceq X\preceq bP$. Positivity and the fixed point imply
$-bP\preceq\lambda^jX\preceq bP$ for every integer $j\ge0$,
which excludes $|\lambda|>1$.

For an arbitrary real eigenmatrix $X$, form the symmetric block matrix
\[
 \widetilde X=\begin{pmatrix}0&X\\X^T&0\end{pmatrix}.
\]
The positive map with Kraus factors $D_a\oplus D_a$ sends
$\widetilde X$ to $\lambda\widetilde X$ and fixes $P\oplus P\succ0$.
The same order argument excludes $|\lambda|>1$ without assuming that
$X$ is symmetric. Since $\Phi$ is self-adjoint on the full real matrix
space, its eigenmatrices form an orthonormal basis there. Thus
$\|\Phi\|_{\HS\to\HS}\le1$ on that full space.

For the curvature transformation, write
\(\mathcal G=\operatorname{Ad}_{Z^{1/2}}\), so \(dS=\mathcal G(Y)\).
The operators of Lemma~\ref{sq:an:hessian} satisfy
\[
 \mathcal K
 =\operatorname{Ad}_{Z^{-1/2}}\mathcal J_P
                    \operatorname{Ad}_{Z^{-1/2}},
 \qquad
 \mathcal L\mathcal G
 =\operatorname{Ad}_{Z^{-1/2}}(\Id-\Phi).
 \]
To verify the first identity, use
\(M=Z^{-1/2}PZ^{-1/2}\) and expand the two summands; they are
\(Z^{-1}UM\) and \(MUZ^{-1}\). The second follows by subtracting
\(\eta_C(Z^{1/2}YZ^{1/2})\) from \(Z^{-1/2}YZ^{-1/2}\).
Substituting these identities into
\(\mathcal G^*\mathcal L^*\mathcal K^{-1}\mathcal L\mathcal G\)
cancels the conjugations and proves the first formula of
\eqref{sq:an:balanced-hessians}.

The balanced Tsallis curvature is
\(\mathcal G^*A_\theta^{\rm original}\mathcal G\). Its inverse is
\(\mathcal G^{-1}(A_\theta^{\rm original})^{-1}\mathcal G^{-1}\).
Substitution of \eqref{sq:an:tsallis-inverse} gives, for a real symmetric \(B\),
\[
 A_\theta^{-1}(B)=\theta^{-1}(PBR+RBP),
 \]
which proves the second formula. Finally congruencing the sum of the
two curvatures by \(\mathcal J_P^{1/2}\) proves
\eqref{sq:an:whitened}. The factors \(\mathcal J_P^{1/2}\) cannot
generally be omitted: \(\Phi\) need not commute with \(\mathcal J_P\).
\end{proof}

The fixed point explains why the reference term alone is insufficient:
$(I-\Phi)P=0$. The regularizer contributes the positive term $H_\theta$
to the full curvature, so the complete system remains invertible. We
will exploit the structure of the forces instead of estimating a
spectral gap for $I-\Phi$.

\begin{table}[ht]
\centering
\begin{tabularx}{\linewidth}{@{}p{.22\linewidth}p{.31\linewidth}>{\raggedright\arraybackslash}X@{}}
\toprule
Coordinate & Change of variables & Useful identity \\
\midrule
Original density & $dS$ & Transport: $Z\eta_C(S)Z=S$. \\
Balanced density & $Y=Z^{-1/2}(dS)Z^{-1/2}$ &
The source becomes $\Phi$ with $\Phi(P)=P$. \\
Normalized curvature & $Y=\mathcal J_P^{1/2}y$ &
The density curvature becomes $(I-T)^*(I-T)+H_\theta$. \\
\bottomrule
\end{tabularx}
\caption{Changes of coordinates at one fixed state. They rewrite the
response quadratic form; the algorithm still stores only its original
coefficient state and covariance.}
\label{tab:response-coordinates}
\end{table}
The next lemma identifies the corresponding forces and shows that their
total squared mass is the fidelity. This is where the covariance reserve
reappears inside the curvature calculation.

\begin{lemma}[Coefficient forces and their budgets]
\label{sq:an:frame}
Assume \(0\preceq C\preceq I_k\), and assume that the
covariance derivative satisfies
\(\Gamma|_{\ran C}\preceq tI\). In the notation of
Lemma~\ref{sq:an:balanced}, put
\begin{equation}
 F_a=2^{-1/2}\mathcal J_P^{1/2}D_a,\qquad
 \Sigma=\sum_a|F_a\rangle\langle F_a|.
 \label{sq:an:frame-definitions}
\end{equation}
Then
\begin{align}
 \langle F_a,F_b\rangle
 &=\Tr(PD_aD_b)=:\Gamma'_{ab},&
 0\preceq\Sigma&\preceq t\Id,
 \label{sq:an:frame-cap}\\
 \Tr_{\HS}\Sigma
 &=\Tr P=F(S,\eta_C(S))\nonumber\\
 &=\Tr(C\Gamma)\le\sqrt k,&\Tr R^2&\le\Tr\eta_C(S)\le k.
 \label{sq:an:frame-budgets}
\end{align}
Here \(\Tr_{\HS}\) denotes the trace of a
Hilbert-space operator.
At the optimizer,
\begin{equation}
 \Tr(CJ)\le\Tr_{\HS}(\Sigma\mathcal A^{-1}).
 \label{sq:an:observed-response}
\end{equation}
For an arbitrary positive \(S\) with \(\Tr S\le1\), the same
fixed-point, Gram, and trace-budget identities hold when \(\Gamma'\)
is defined by the displayed Gram. The local unconstrained response is
the quadratic form \(\Tr_{\HS}(\Sigma\mathcal A^{-1})\).
None of these algebraic assertions requires such a subnormalized \(S\)
to optimize a separate barrier.
\end{lemma}

\begin{proof}
For any real vector \(z\), the retained direction \(u=C^{1/2}z\)
lies in \(\ran C\). Thus the real Gram after factoring the covariance
satisfies
\[
 \Gamma'=C^{1/2}\Gamma C^{1/2}
 \preceq tC\preceq tI.
 \]
This equation uses only the restriction of \(\Gamma\) to \(\ran C\);
it does not presume differentiability in directions which create a new
coefficient support. The covariance derivative
\eqref{sq:an:gamma-original} gives
\[
 \Gamma'_{ab}
 =\Tr(SB'_aZB'_b)=\Tr(PD_aD_b).
 \]
On the other hand,
\[
 \langle F_a,F_b\rangle
 =\tfrac12\Tr\bigl(D_a(PD_b+D_bP)\bigr)
 =\Tr(PD_aD_b).
 \]
The nonzero eigenvalues of the frame operator equal those of its real
Gram matrix, proving \eqref{sq:an:frame-cap}.

Summing the diagonal Gram entries, using self-adjointness of \(\eta_C\),
and then the transport equation gives
\[
 \begin{aligned}
 \sum_a\Gamma'_{aa}&=\Tr(S\eta_C(Z))=\Tr(\eta_C(S)Z)\\
 &=F(S,\eta_C(S))=\Tr(SZ^{-1})=\Tr P.
 \end{aligned}
 \]
It is also \(\Tr(C\Gamma)\). Although the \(B'_a\) need not be
contractions, the original contraction realization implies
\(\eta_C(I)\preceq kI\). Indeed, for a physical vector \(u\), the
coefficient vector with physical entries \((\A_iu)_i\) satisfies
\[
 u^T\eta_C(I)u
 =\sum_{i,j}C_{ij}\langle \A_iu,\A_ju\rangle
 \le\sum_i\|\A_iu\|^2\le k\|u\|^2.
 \]
Consequently
\[
 \begin{aligned}
 \Tr\eta_C(S)&=\Tr(S\eta_C(I))\le k\Tr S,\\
 F(S,\eta_C(S))&\le\sqrt{\Tr S\,\Tr\eta_C(S)}\le\sqrt k\,\Tr S.
 \end{aligned}
 \]
These inequalities prove the claimed budgets when \(\Tr S\le1\).

For the \(R\)-budget, set \(W=Z^{-1}\) and write
\([X,Y]=XY-YX\) for the commutator of two matrices. Direct expansion gives
\[
 \begin{aligned}
 \Tr\eta_C(S)-\Tr R^2&=\Tr(SW^2)-\Tr(\sqrt S\,W\sqrt S\,W)\\
 &=\tfrac12\|[\sqrt S,W]\|_{\HS}^2\ge0.
 \end{aligned}
 \]

Finally the physical force \(B'_a\) acts on a density perturbation as
\[
 \Tr(B'_a\,dS)
 =\langle D_a,Y\rangle
 =\langle\mathcal J_P^{1/2}D_a,y\rangle.
 \]
One half of its unconstrained inverse-Hessian response is therefore
\(\langle F_a,\mathcal A^{-1}F_a\rangle\).
The optimized center response has the additional trace constraint, which
only decreases the inverse variational supremum by
Lemma~\ref{sq:an:hessian}. Summing the factored coefficient directions
proves \eqref{sq:an:observed-response}. All earlier identities in this
proof concern a current positive matrix and its source, so they remain
valid for the stated subnormalized matrices.
\end{proof}

\begin{lemma}[Inverse domination retaining all density coupling]
\label{sq:an:inverse}
For an operator \(T\) on a finite-dimensional real Hilbert space and
a strictly positive self-adjoint \(H_\theta\), set
\[
 \mathcal A=(\Id-T)^*(\Id-T)+H_\theta,\qquad
 K_0=TH_\theta^{-1}T^*.
 \]
Then
\begin{equation}
 \mathcal A^{-1}\preceq\Id+K_0.
 \label{sq:an:inverse-domination}
\end{equation}
\end{lemma}

\begin{proof}
For every force \(f\), split
\(\langle f,x\rangle=
\langle f,(\Id-T)x\rangle+\langle T^*f,x\rangle\) in the
inverse variational formula. Completing squares, separately in the
two resulting expressions, gives
\begin{align*}
 \langle f,\mathcal A^{-1}f\rangle
 &=\sup_x\{2\langle f,(\Id-T)x\rangle-\|(\Id-T)x\|^2
       \\
 &\hspace{9mm}+2\langle T^*f,x\rangle-\langle x,H_\theta x\rangle\}\\
 &\le\|f\|^2+\langle T^*f,H_\theta^{-1}T^*f\rangle
 =\langle f,(\Id+K_0)f\rangle.
\end{align*}
The same \(x\) occurs in both terms before the inequality. Each separate
supremum bounds its term from above; their sum therefore bounds the
original coupled optimization.
\end{proof}

\subsection{Singular physical sources}

If the covariance-weighted matrices $B'_a$ vanish on a physical subspace,
the fidelity term depends only on the complementary compression of
$S$. The regularizer still depends on the full $S$, including off-block
perturbations. Simply restricting the whole Hessian would therefore
require justification. We obtain the needed comparison by first
minimizing the full regularizer's quadratic energy over every
perturbation with a fixed compression. This allows the omitted
coordinates to respond optimally, and so preserves the direction of
the desired upper bound for the inverse response.

\begin{proposition}[Compression comparison for retained responses]
\label{sq:an:singular-transfer}
Let \(S=S(H,C)>0\) be the full optimizing density, and allow
both \(C\) and its physical source to be singular. Factor \(C\) as
above and put
\[
 K=\operatorname{span}_a\ran B'_a.
 \]
If \(K=\{0\}\), all covariance-weighted forces \(B'_a\) vanish and
\(\Tr(CJ)=0\). Otherwise let
\(\pi\) be compression onto \(K\), and set \(S_0=\pi(S)\).
The compressed source \(M_0=\eta_C(S_0)\) is faithful on \(K\).
Let \(A_{f,0}\) and \(A_{\theta,0}\) be the full, unconstrained
negative density Hessians of the compressed free and Tsallis terms at
\(S_0\). For a real symmetric force $B$ supported on $K$, write
$B_0=\pi(B)$. Then
\begin{equation}
 \tfrac12D_H^2E(H,C)[B,B]\le
 \frac12\langle B_0,(A_{f,0}+A_{\theta,0})^{-1}B_0\rangle.
 \label{sq:an:singular-response}
\end{equation}
The covariance derivative on \(\ran C\) is exactly the
compressed covariance derivative evaluated at \(S_0\).
All balanced identities and budgets in
Lemmas~\ref{sq:an:balanced}--\ref{sq:an:frame} apply to this compressed
problem with the original, unchanged \(\theta\) and
\(\Tr S_0\le1\). In particular, any upper bound for the sum of the
local unconstrained observed responses with these cap and budget
assumptions also bounds \(\Tr(CJ)\) for the original full problem.
\end{proposition}

\begin{proof}
Let \(\Pi\) be the physical projection onto \(K\).
Symmetry implies \(B'_a=\Pi B'_a\Pi\). Thus
\(\eta_C(S)=\eta_C(\Pi S\Pi)\), supported on \(K\).
The compression \(S_0\) is positive definite there because \(S>0\).
For \(u\in K\),
\[
 u^TM_0u=\sum_a(B'_au)^*S_0(B'_au).
 \]
If this vanishes, every \(B'_au\) vanishes. Such a vector lies in
the orthogonal complement of the span of their ranges, hence is zero.
This proves \(M_0>0\) on \(K\).

The two possible orders of the fidelity square root have the same
nonzero eigenvalues. Since its second argument is supported on \(K\),
Lemma~\ref{sq:fo:fidelity} gives
\begin{equation}
 F(S,\eta_C(S))=F(S_0,\eta_C(S_0)).
 \label{sq:an:singular-fidelity}
\end{equation}
Consequently the full free density curvature is
\(\pi^*A_{f,0}\pi\). Every coefficient direction in \(\ran C\)
is a linear combination of the \(B'_a\), so its physical force is
supported on \(K\). The derivative of
\eqref{sq:an:singular-fidelity} with respect to such support-preserving
covariance directions proves the claimed equality of covariance
derivatives.

It is essential to retain the full Tsallis response, including
off-block density perturbations. Write
\(T_0=\pi(\sqrt S)\). By square-root concavity under compression,
Lemma~\ref{sq:fo:sqrt} gives \(T_0\preceq\sqrt{S_0}\).
The explicit inverse Hessian \eqref{sq:an:tsallis-inverse} gives
\begin{equation}
 \pi(A_\theta^{\rm original})^{-1}\pi^*
 =\theta^{-1}(L_{S_0}R_{T_0}+L_{T_0}R_{S_0})
 \preceq A_{\theta,0}^{-1}.
 \label{sq:an:compressed-tsallis}
\end{equation}
This last inequality is an inequality of Hilbert-space operators.
For a real symmetric \(X\) on \(K\), its difference quadratic is exactly
\[
 \frac2\theta
 \Tr\bigl(XS_0X(\sqrt{S_0}-T_0)\bigr)\ge0,
 \]
because both \(XS_0X\) and \(\sqrt{S_0}-T_0\) are positive.

We spell out the compression rule for quadratic minimization.
If \(A>0\) on the full real symmetric space, then
\begin{equation}
 \inf_{\pi X=Y}\langle X,AX\rangle
 =\left\langle Y,(\pi A^{-1}\pi^*)^{-1}Y\right\rangle.
 \label{sq:an:shorting-quadratic}
\end{equation}
Indeed a Lagrange multiplier \(\Lambda\) gives
\(AX=\pi^*\Lambda\), hence
\(Y=\pi A^{-1}\pi^*\Lambda\).
The compression is surjective and \(A^{-1}>0\), so this last operator
is positive definite. Substitution yields
\eqref{sq:an:shorting-quadratic}.

Apply this rule to the full Tsallis curvature. Inequality
\eqref{sq:an:compressed-tsallis} reverses under inversion, so the
minimal Tsallis energy at a fixed compressed perturbation \(Y\) is
at least \(\langle Y,A_{\theta,0}Y\rangle\).
Adding the free energy, which depends only on \(Y\), shows
\[
 \inf_{\pi X=Y}
 \langle X,(\pi^*A_{f,0}\pi+A_\theta^{\rm original})X\rangle
 \ge\langle Y,(A_{f,0}+A_{\theta,0})Y\rangle.
 \]
Apply the inverse variational formula to a force supported on \(K\).
Its linear term depends only on \(\pi X\), and the preceding lower
bound on energy gives
\[
 \left\langle B,
   (\pi^*A_{f,0}\pi+A_\theta^{\rm original})^{-1}B\right\rangle
 \le\langle B_0,(A_{f,0}+A_{\theta,0})^{-1}B_0\rangle.
 \]
The actual density trace constraint can only decrease the left-hand
response. This proves \eqref{sq:an:singular-response}.

Finally \(\Tr S_0\le\Tr S=1\), and the original contraction
realization implies \(\eta_C(I_K)\preceq kI_K\).
Thus
\[
 \Tr M_0\le k\Tr S_0\le k,\qquad
 F(S_0,M_0)\le
       \sqrt{\Tr S_0\,\Tr M_0}\le\sqrt k.
 \]
The transport on \(K\) is positive and invertible, so the balanced
formulas apply there. In particular
\(\Tr P=F(S_0,M_0)\) and
\(\Tr R^2\le\Tr M_0\), with the same commutator identity as before.
The covariance cap after factoring \(C\) is unchanged.
These are algebraic identities at $S_0$. The comparison therefore uses
the original center and regularizer strength throughout.
\end{proof}

\begin{remark}\label{sq:an:generality}
The response comparisons are valid for every retained label count $k$.
They use $C\preceq I_k$ and the source budget $\eta_C(I)\preceq kI$,
so they continue to apply after the live count becomes smaller than the
physical dimension. Compression retains the full density response
through quadratic minimization; it does not restart the density problem.
\end{remark}

\section{Bounding the response of the coefficient forces}
\label{sq:te:section}

We now prove the estimate promised in the overview:
\[
 \Gamma|_{\ran C}\preceq Lk^{-1/2}I
 \quad\Longrightarrow\quad
 \Tr(CJ)\le(2+12\sqrt L/\theta)\sqrt k.
\]
An inverse Hessian can be large in directions where the fidelity
curvature is small. The coefficient forces are special, however: they
are built from the same Kraus matrices as the source. Their Gram cap
and their total mass $\sum_a\|F_a\|_{\HS}^2=\Tr P\le\sqrt k$
allow us to bound their combined response without bounding the inverse
uniformly on every matrix direction.

The argument has four steps. We convert the real coefficient Gram cap
into two frame estimates on real matrix space; we bound forces that meet the small
spectrum of $P$ using only the regularizer; we estimate the remaining
forces through two explicit tensor maps; and we insert both estimates
into the full inverse quadratic form. The tensor calculation is needed
only in the third step. The last subsection assembles all constants.

We prove the response estimate using the Kraus frame.
Throughout this section, \(P,R\) are positive definite physical matrices,
\(D_1,\ldots,D_r\) are real symmetric matrices, and
\[
 \Phi(X)=\sum_{a=1}^rD_aXD_a,\qquad \Phi(P)=P.
\]
The matrices \(D_a\) are the balanced Kraus matrices of the preceding
analytic calculation; they need not individually be contractions.
We work on the full real matrix Hilbert space with inner product
$\langle X,Y\rangle=\Tr(X^TY)$. The symmetric density tangent is a
subspace; the tensor maps need the full space because their intermediate
matrices need not be symmetric. We distinguish physical trace $\Tr$
from matrix-Hilbert-space operator trace $\Tr_{\HS}$.
For physical matrices \(U,V\), let
\[
 L_U(X)=UX,\qquad R_V(X)=XV,\qquad
 \operatorname{Ad}_U(X)=UXU^*.
\]
In particular, \(\mathcal J_P=L_P+R_P\).

The analytic identities needed here are
\begin{align}
 F_a&=2^{-1/2}\mathcal J_P^{1/2}D_a,
 &\Sigma&=\sum_a|F_a\rangle\langle F_a|, \label{sq:te:frame-def}\\
 T&=\mathcal J_P^{-1/2}\Phi \mathcal J_P^{1/2},\nonumber\\
 H_\theta^{-1}
   &=\theta^{-1}\mathcal J_P^{-1/2}(L_PR_R+L_RR_P)\mathcal J_P^{-1/2}, \label{sq:te:operators}\\
 \mathcal A&=(\Id-T)^*(\Id-T)+H_\theta,
 &K_0&=TH_\theta^{-1}T^*. \label{sq:te:full-operator}
\end{align}
For the actual covariance, the preceding response calculation gives
\begin{equation}
 \Tr(CJ)\le \Tr_{\HS}(\Sigma\mathcal A^{-1}).
 \label{sq:te:actual-response}
\end{equation}
The global density trace constraint has only been dropped in the
favorable direction in this inequality.  In the application there are
\(k\) retained original contraction labels, and
\begin{equation}
 \Tr P\le\sqrt{k},\qquad \Tr R^2\le k.
 \label{sq:te:physical-budgets}
\end{equation}
Our argument below also applies to a subnormalized faithful density:
only the inequalities in \eqref{sq:te:physical-budgets}, rather than
normalization as an equality, are used.

\subsection{Two frame estimates on real matrix space}

The covariance cap is the Gram bound for the balanced matrices.
In the real setting, the same Gram controls both left- and right-weighted
synthesis maps directly. A second estimate bounds a Hilbert--Schmidt norm
by rearranging four matrix indices; no operator-norm invariance under that
rearrangement is claimed. We use vectorization
$\operatorname{vec}(X)_{(i,j)}=X_{ij}$, with the pair $(i,j)$ in
lexicographic order; this identifies the full real matrix space
isometrically with Euclidean space.

\begin{lemma}[Two bounds on the Kraus frame]
\label{sq:te:two-frame-lemma}
Suppose the real Gram matrix
\[
 G_{ab}=\Tr(PD_aD_b)
\]
satisfies \(G\preceq tI_r\), where \(t>0\).
Define
\[
 \mathcal C_\Phi=\sum_a|D_a\rangle\langle D_a|.
\]
Then, as operators on the full real matrix Hilbert space,
\begin{equation}
 \Sigma\preceq t\Id,\qquad
 \Tr_{\HS}\Sigma=\Tr P,\qquad
 \mathcal C_\Phi\preceq 2t\mathcal J_P^{-1}.
 \label{sq:te:choi-cap}
\end{equation}
The synthesis maps with respective columns
\(\operatorname{vec}(P^{1/2}D_a)\) and
\(\operatorname{vec}(D_aP^{1/2})\) have squared operator norm at most
\(2t\).  Finally,
\begin{equation}
 \|\Phi\operatorname{Ad}_{\sqrt P}\|_{\HS}^{\,2}
 \le 2t\Tr P.
 \label{sq:te:weighted-realignment}
\end{equation}
Here the norm in \eqref{sq:te:weighted-realignment} is the
Hilbert--Schmidt norm of a superoperator.
\end{lemma}

\begin{proof}
Transposition and cyclicity of trace give
$\Tr(PD_aD_b)=\Tr(PD_bD_a)$, so the displayed Gram is real symmetric.
Directly,
\[
 \langle F_a,F_b\rangle
 =\tfrac12\Tr\{D_a(PD_b+D_bP)\}
 =\Tr(PD_aD_b)=G_{ab}.
\]
The synthesis map with columns $F_a$ has domain $\R^r$ and codomain
all real matrices. Its Gram is $G$, so the frame--Gram spectral
correspondence gives $\Sigma\preceq t\Id$ on that full space.
Moreover
\[
 \Tr_{\HS}\Sigma=\sum_a\Tr(PD_a^2)
 =\Tr\!\left(\sum_aD_aPD_a\right)=\Tr P.
\]
The exact congruence
\[
 \Sigma=\tfrac12\mathcal J_P^{1/2}\mathcal C_\Phi\mathcal J_P^{1/2}
\]
then gives $\mathcal C_\Phi\preceq2t\mathcal J_P^{-1}$.
Every equality here holds on full real matrix space: the frames vanish
on skew-symmetric inputs, and $\mathcal J_P$ preserves symmetric and
skew-symmetric subspaces.

For $K_a=D_aP^{1/2}$, its ordinary real Gram is
\[
 \Tr(K_a^TK_b)=\Tr(PD_aD_b)=G_{ab}.
\]
The Gram of $P^{1/2}D_a$ is $G^T=G$. Both synthesis maps therefore
have squared norm at most $t$, which in particular proves the stated
conservative bound $2t$. Right compression by an orthogonal physical
projection preserves this bound because it contracts the
Hilbert--Schmidt norm.

For the weighted estimate \eqref{sq:te:weighted-realignment}, put
\[
 \mathcal C_K=\sum_a|\operatorname{vec}K_a\rangle
                       \langle\operatorname{vec}K_a|.
\]
The preceding Gram estimate gives
\[
 \|\mathcal C_K\|_{\op}\le2t,\qquad
 \Tr_{\HS}\mathcal C_K
   =\sum_a\Tr(K_a^*K_a)
   =\Tr P.
\]
Let \(\Psi(X)=\sum_aK_aXK_a^*
                 =\Phi(P^{1/2}XP^{1/2})\).
Its matrix in the matrix-unit basis has entries
\[
 [\Psi]_{(i,k),(j,l)}
   =\sum_a(K_a)_{ij}(K_a)_{kl}
   =[\mathcal C_K]_{(i,j),(k,l)}.
\]
This is a permutation of four indices.  In particular,
\[
 \|\Psi\|_{\HS}^{\,2}
 =\sum_{i,j,k,l}
       \left|\sum_a(K_a)_{ij}(K_a)_{kl}\right|^2
 =\|\mathcal C_K\|_{\HS}^{\,2}.
\]
Since \(\mathcal C_K\succeq0\),
\[
 \|\mathcal C_K\|_{\HS}^{\,2}
 \le \|\mathcal C_K\|_{\op}\Tr_{\HS}\mathcal C_K
 \le2t\Tr P.
\]
This is exactly \eqref{sq:te:weighted-realignment}.
\end{proof}

\subsection{Forces meeting the small spectrum of the balanced density}

The inverse reference curvature becomes large near small eigenvalues
of $P$. On this part of the spectrum we use the square-root regularizer
directly. Its inverse quadratic form contains a factor of $P$, and
that factor makes the total contribution small. This is a bound on the
forces, so the full Hessian still includes all density directions.

Lemma~\ref{sq:an:balanced} proves contraction on the full real matrix
space, including nonsymmetric inputs, by a real symmetric block
amplification. In particular,
\begin{equation}
 \|\Phi(X)\|_{\HS}\le\|X\|_{\HS}
 \quad\hbox{for every matrix }X.
 \label{sq:te:phi-contraction}
\end{equation}

Fix \(\epsilon>0\), and use the physical spectral projections
\[
 \Pi_h=1_{\{P\ge\epsilon\}},\qquad
 \Pi_l=I-\Pi_h,\qquad P_l=\Pi_lP.
\]
The notation \(\Pi_h\) is unrelated to the discrepancy center \(H\).
Split the real symmetric forces as
\begin{equation}
 D_a=D_a^h+D_a^l,\qquad
 D_a^h=\Pi_hD_a\Pi_h,\qquad
 D_a^l=\Pi_lD_a+\Pi_hD_a\Pi_l.
 \label{sq:te:force-split}
\end{equation}
Although the last two summands need not be real symmetric separately,
their sum \(D_a^l=D_a-\Pi_hD_a\Pi_h\) is real symmetric.

\begin{lemma}[Low-spectrum forcing bound]
\label{sq:te:low-forcing}
For a real symmetric matrix \(D\), its Tsallis-only half-response is
\[
 q_\theta(D)=\theta^{-1}
      \|P^{1/2}DR^{1/2}\|_{\HS}^{\,2}.
\]
If \eqref{sq:te:physical-budgets} holds and \(\epsilon=c/\sqrt{k}\), then
\begin{equation}
 \sum_a q_\theta(D_a^l)\le\frac{4\sqrt{ck}}{\theta}.
 \label{sq:te:low-bound}
\end{equation}
\end{lemma}

\begin{proof}
Writing \(F(D)=2^{-1/2}\mathcal J_P^{1/2}D\), equation
\eqref{sq:te:operators} gives, for real symmetric \(D\),
\begin{align*}
 \langle F(D),H_\theta^{-1}F(D)\rangle
 &=\frac1{2\theta}\langle D,(L_PR_R+L_RR_P)D\rangle\\
 &=\frac1\theta\Tr(PDRD)
 =\frac1\theta\|P^{1/2}DR^{1/2}\|_{\HS}^{\,2}.
\end{align*}
In the last identity one uses \(D^*=D\).
Apply the squared triangle inequality to the last decomposition in
\eqref{sq:te:force-split}.  The two sums of squared norms are
\begin{align*}
 \sum_a\|P^{1/2}\Pi_lD_aR^{1/2}\|_{\HS}^{\,2}
 &=\sum_a\Tr(RD_aP_lD_a)
   =\Tr(R\Phi(P_l)),\\
 \sum_a\|P^{1/2}\Pi_hD_a\Pi_lR^{1/2}\|_{\HS}^{\,2}
 &=\sum_a\Tr(\Pi_lR\Pi_lD_a(P-P_l)D_a)\\
 &=\Tr(\Pi_lR\Pi_l\Phi(P-P_l)).
\end{align*}
Positivity and the fixed-point equation give
\[
 0\preceq\Phi(P-P_l)\preceq P.
\]
Since \(\Pi_lR\Pi_l\succeq0\), the second sum is at most
\(\Tr(\Pi_lR\Pi_lP)=\Tr(RP_l)\).  Therefore
\begin{equation}
 \sum_aq_\theta(D_a^l)
 \le\frac2\theta\{\Tr(R\Phi(P_l))+\Tr(RP_l)\}.
 \label{sq:te:low-intermediate}
\end{equation}
The threshold choice and \eqref{sq:te:physical-budgets} imply
\[
 \|P_l\|_{\HS}^2=\Tr P_l^2
 \le\epsilon\Tr P_l
 \le\epsilon\Tr P\le c.
\]
By \eqref{sq:te:phi-contraction},
\(\|\Phi(P_l)\|_{\HS}\le\sqrt c\), while
\(\|R\|_{\HS}\le\sqrt k\).
Hilbert--Schmidt Cauchy--Schwarz bounds each trace in
\eqref{sq:te:low-intermediate} by \(\sqrt{ck}\), proving
\eqref{sq:te:low-bound}.
\end{proof}

\subsection{A tensor estimate on the complementary spectrum}

After compressing to eigenvalues of $P$ at least $\epsilon$, the
weighted frame cap becomes an ordinary frame cap with constant
$t/\epsilon$. One of the tensor maps below uses this cap; the other
uses the weighted synthesis cap from the preceding lemma. Their
composition gives a physical matrix order inequality. An independent
trace estimate is then combined with it to control the Hilbert--Schmidt
norm. Both estimates are needed: the order inequality alone would
leave an unwanted dependence on the largest eigenvalue of $P$.

The next lemma does not assume the scalar budgets
\eqref{sq:te:physical-budgets}.  It uses only the fixed point and the
coefficient Gram cap, and hence separates the tensor estimate from its
discrepancy application.

\begin{lemma}[High-spectrum Kraus tensor bound]
\label{sq:te:tensor-lemma}
Under the assumptions of Lemma~\ref{sq:te:two-frame-lemma}, set
\[
 \beta=\frac{t}{\epsilon},\qquad
 W_h=\sum_a\Phi(D_a^h)P\Phi(D_a^h).
\]
Then
\begin{align}
 W_h&\preceq 2t\beta\,\Phi(\Pi_h)
       \preceq2\beta^2P, \label{sq:te:high-order}\\
 \Tr(PW_h)&\le2t\beta\Tr P, \label{sq:te:high-trace}\\
 \Tr W_h^2&\le4t\beta^3\Tr P. \label{sq:te:high-hs}
\end{align}
\end{lemma}

\begin{proof}
Let \(\mathcal E(X)=\Pi_hX\Pi_h\).  This orthogonal projection on
matrix space commutes with \(\mathcal J_P\).  On its range,
\(\mathcal J_P\succeq2\epsilon\Id\).  Compression of
\eqref{sq:te:choi-cap} therefore gives
\begin{equation}
 \mathcal C_h:=\sum_a|D_a^h\rangle\langle D_a^h|
   =\mathcal E\mathcal C_\Phi\mathcal E
   \preceq\beta\mathcal E.
 \label{sq:te:high-choi}
\end{equation}

We prove the order estimate by writing both flattenings explicitly.
Choose an eigenbasis of \(P\), so that \(\Pi_h\) is diagonal.
For a fixed real physical vector \(v\), put
\[
 M_b=P^{1/2}D_b\Pi_h,\qquad
 w_{bj}=(\Pi_hD_bv)_j.
\]
The linear map from the coefficient index \(b\) to the pair of
physical indices \((i,l)\) with entries
\[
 \mathsf F_{(i,l),b}=(M_b)_{il}
\]
has squared norm at most \(2t\), by the synthesis assertion of
Lemma~\ref{sq:te:two-frame-lemma} and right compression.
The second map has entries
\[
 \mathsf G_{a,(l,j)}=(D_a^h)_{lj}.
\]
It is the transpose of the real synthesis matrix with columns
$\operatorname{vec}(D_a^h)$. Hence \eqref{sq:te:high-choi} gives
$\|\mathsf G\|_{\op}^2\le\beta$.

View \(w\) as a vector with indices \((b,j)\).  Its squared norm is
\[
 \|w\|^2=\sum_b\|\Pi_hD_bv\|^2
        =v^T\Phi(\Pi_h)v.
\]
Apply \(\mathsf F\) to the \(b\) index, leaving \(j\) unchanged, and
then apply \(\mathsf G\) to the \((l,j)\) indices, leaving \(i\)
unchanged.  Equivalently these maps are
\(\mathsf F\otimes I\) followed by \(I\otimes\mathsf G\).
Their final entries are exactly
\begin{align*}
 s_{ia}
 &=\sum_{b,l,j}
       (P^{1/2}D_b\Pi_h)_{il}(D_a^h)_{lj}(\Pi_hD_bv)_j\\
 &=\left[P^{1/2}\sum_bD_bD_a^hD_bv\right]_i
   =[P^{1/2}\Phi(D_a^h)v]_i.
\end{align*}
Here \(\Pi_hD_a^h\Pi_h=D_a^h\) was used in the second line.
The squared norm of the composed map is at most \(2t\beta\).
Since every \(\Phi(D_a^h)\) is real symmetric, this proves
\[
 v^TW_hv=\sum_a\|P^{1/2}\Phi(D_a^h)v\|^2
 \le2t\beta\,v^T\Phi(\Pi_h)v.
\]
The inequality holds for every physical vector \(v\).
Also \(\epsilon\Pi_h\preceq P\), so positivity gives
\[
 \Phi(\Pi_h)\preceq\epsilon^{-1}\Phi(P)
                  =\epsilon^{-1}P.
\]
Both assertions in \eqref{sq:te:high-order} follow.

The trace bound uses the independent weighted frame estimate.
The exact identity is
\begin{align}
 \Tr(PW_h)
 &=\sum_a\Tr\{P\Phi(D_a^h)P\Phi(D_a^h)\}\notag\\
 &=\sum_a\langle D_a^h,
                  \Phi\operatorname{Ad}_P\Phi(D_a^h)\rangle\notag\\
 &=\Tr_{\HS}
          (\mathcal C_h\Phi\operatorname{Ad}_P\Phi).
 \label{sq:te:trace-identity}
\end{align}
We used self-adjointness of \(\Phi\) in the second line.
The operator \(\Phi\operatorname{Ad}_P\Phi\) is positive
semidefinite, even if \(\Phi\) has some negative eigenvalues as a
Hilbert-space operator, because
\[
 \langle X,\Phi\operatorname{Ad}_P\Phi(X)\rangle
     =\|P^{1/2}\Phi(X)P^{1/2}\|_{\HS}^{\,2}.
\]
It follows from \eqref{sq:te:high-choi} that
\begin{align*}
 \Tr(PW_h)
 &\le\beta\Tr_{\HS}
                 (\mathcal E\Phi\operatorname{Ad}_P\Phi)\\
 &\le\beta\Tr_{\HS}(\Phi\operatorname{Ad}_P\Phi)\\
 &=\beta\|\Phi\operatorname{Ad}_{\sqrt P}\|_{\HS}^{\,2}
 \le2t\beta\Tr P.
\end{align*}
The middle equality uses
\[
 (\Phi\operatorname{Ad}_{\sqrt P})
 (\Phi\operatorname{Ad}_{\sqrt P})^*
       =\Phi\operatorname{Ad}_P\Phi.
\]
This proves \eqref{sq:te:high-trace}.
Finally both \(W_h\) and \(2\beta^2P-W_h\) are positive semidefinite;
their product has nonnegative trace.  Hence
\[
 \Tr W_h^2\le2\beta^2\Tr(PW_h)
            \le4t\beta^3\Tr P,
\]
which proves \eqref{sq:te:high-hs} without assuming that \(P\) and \(W_h\)
commute.
\end{proof}

\subsection{The contribution from the complementary spectrum}

The preceding tensor estimate is useful because the remaining
inverse-Hessian contribution has an exact physical-trace expression.
We now identify it as $\theta^{-1}\Tr(RW_h)$. Once this identity is
written down, the two available Hilbert--Schmidt budgets complete the
estimate by Cauchy--Schwarz.

Let
\[
 F_a^h=2^{-1/2}\mathcal J_P^{1/2}D_a^h,\qquad
 \Sigma_h=\sum_a|F_a^h\rangle\langle F_a^h|.
\]
Let $\mathcal E(X)=\Pi_hX\Pi_h$ on full matrix space.
Because \(\mathcal E\) commutes with \(\mathcal J_P\),
\begin{equation}
 \Sigma_h=\mathcal E\Sigma\mathcal E
          =\tfrac12\mathcal J_P^{1/2}\mathcal C_h\mathcal J_P^{1/2},
 \qquad
 \Tr_{\HS}\Sigma_h\le\Tr P.
 \label{sq:te:high-frame}
\end{equation}
The full operator \(K_0\) in \eqref{sq:te:full-operator} is retained:
only the forcing frame is compressed.

\begin{lemma}[Exact high observed identity]
\label{sq:te:observed-identity}
One has
\begin{equation}
 \Tr_{\HS}(\Sigma_hK_0)=\theta^{-1}\Tr(RW_h).
 \label{sq:te:observed-exact}
\end{equation}
If \eqref{sq:te:physical-budgets} holds, \(t=L/\sqrt k\), and
\(\epsilon=c/\sqrt k\), then
\begin{equation}
 \Tr_{\HS}(\Sigma_hK_0)
      \le\frac{2L^2}{c^{3/2}\theta}\sqrt k.
 \label{sq:te:observed-bound}
\end{equation}
\end{lemma}

\begin{proof}
Set \(\mathcal B=L_PR_R+L_RR_P\).  Cancellation of the neighboring
Sylvester factors in \eqref{sq:te:operators} gives
\[
 K_0=\theta^{-1}\mathcal J_P^{-1/2}\Phi\mathcal B\Phi \mathcal J_P^{-1/2}.
\]
Using \eqref{sq:te:high-frame}, cyclicity of the Hilbert-space trace, and
self-adjointness of \(\Phi\), we obtain
\begin{align*}
 \Tr_{\HS}(\Sigma_hK_0)
 &=\frac1{2\theta}
         \Tr_{\HS}(\mathcal C_h\Phi\mathcal B\Phi)\\
 &=\frac1{2\theta}
         \sum_a\langle\Phi(D_a^h),\mathcal B\Phi(D_a^h)\rangle.
\end{align*}
For the real symmetric matrix \(Y_a=\Phi(D_a^h)\), the inner product in
the last line equals
\[
 \Tr(Y_aPY_aR)+\Tr(Y_aRY_aP)=2\Tr(RY_aPY_a).
\]
This proves \eqref{sq:te:observed-exact}, including its factor.
Lemma~\ref{sq:te:tensor-lemma} and the scalar budgets imply
\[
 \|W_h\|_{\HS}^{\,2}
 \le4t(t/\epsilon)^3\Tr P
 \le\frac{4L^4}{c^3}.
\]
Hilbert--Schmidt Cauchy--Schwarz, \(\|R\|_{\HS}\le\sqrt k\), and
\eqref{sq:te:observed-exact} now give \eqref{sq:te:observed-bound}.
\end{proof}

\subsection{The cap-to-response estimate}

We can now state the analytic theorem without balanced-coordinate
notation. Its hypotheses concern the covariance derivative at the
current optimizer, and its conclusion concerns the full center
response at that same optimizer. The cutoff parameter in the proof
is a device for estimating the response; the algorithm never computes
it or changes its covariance using that cutoff.

\begin{theorem}[Observed cap-to-response estimate]
\label{sq:te:cap-response-theorem}
Suppose \(k\ge1\) original contraction labels are retained and the
coefficient covariance satisfies \(0\preceq C\preceq I_k\).
Assume its support-preserving covariance derivative satisfies
\[
 \Gamma\big|_{\ran C}\preceq\frac{L}{\sqrt{k}}I
 \qquad(L>0).
\]
Its full optimized center response obeys
\begin{equation}
 \boxed{\displaystyle
 \Tr(CJ)\le
       \left(2+\frac{12\sqrt L}{\theta}\right)\sqrt{k}.}
 \label{sq:br:cap-response}
\end{equation}
More generally, before selecting the cutoff parameter, every \(c>0\)
gives
\begin{equation}
 \Tr(CJ)\le
 \left(2+\frac{8\sqrt c+4L^2c^{-3/2}}{\theta}\right)\sqrt{k}.
 \label{sq:te:bridge-parameter}
\end{equation}
\end{theorem}

\begin{proof}
First suppose the physical source is faithful.
After factoring the real coefficient covariance, the balanced force Gram is
the congruence \(C^{1/2}\Gamma C^{1/2}\).  The assumed cap only needs
to hold on \(\ran C\), since \(C^{1/2}z\in\ran C\).  Because \(C\preceq I\),
\[
 C^{1/2}\Gamma C^{1/2}
 \preceq\frac{L}{\sqrt k}C
 \preceq\frac{L}{\sqrt k}I.
\]
Thus the preceding lemmas apply with \(t=L/\sqrt k\) to the
factored frame.  The scalar budgets
\eqref{sq:te:physical-budgets} also refer to this same covariance.

Lemma~\ref{sq:an:inverse} gives
\begin{equation}\label{sq:te:inverse-comparison}
 \mathcal A^{-1}\preceq\Id+K_0.
\end{equation}
Also \(\mathcal A\succeq H_\theta\), and hence
\(\mathcal A^{-1}\preceq H_\theta^{-1}\).

Split each force as
\[
 F_a=F_a^h+F_a^l,\qquad
 F_a^l=2^{-1/2}\mathcal J_P^{1/2}D_a^l.
\]
The squared triangle inequality in the norm defined by
\(\mathcal A^{-1}\) gives
\begin{align*}
 \Tr(CJ)
 &\le\sum_a\langle F_a,\mathcal A^{-1}F_a\rangle\\
 &\le2\sum_a\langle F_a^h,\mathcal A^{-1}F_a^h\rangle
       +2\sum_a\langle F_a^l,\mathcal A^{-1}F_a^l\rangle\\
 &\le2\Tr_{\HS}\Sigma_h
       +2\Tr_{\HS}(\Sigma_hK_0)
       +2\sum_aq_\theta(D_a^l).
\end{align*}
The last line uses \eqref{sq:te:inverse-comparison} only for the high
forces, and the Tsallis-only comparison only for the low forces.
The force split leaves the inverse quadratic form on the full density
space. Its cross term has been bounded by the squared triangle inequality.
Choose \(\epsilon=c/\sqrt k\).  Equations
\eqref{sq:te:high-frame}, \eqref{sq:te:observed-bound}, and
\eqref{sq:te:low-bound} yield
\[
 \Tr(CJ)
 \le2\sqrt k+\frac{4L^2}{c^{3/2}\theta}\sqrt k
                 +\frac{8\sqrt c}{\theta}\sqrt k,
\]
which is \eqref{sq:te:bridge-parameter}.
Taking \(c=L\) proves \eqref{sq:br:cap-response}.

For a singular physical source, apply the compression comparison
of Proposition~\ref{sq:an:singular-transfer}.
It bounds the actual response on the covariance-weighted forces by the local
unconstrained response on the source support, at a faithful density
\(S_0\) with \(\Tr S_0\le1\).
The local covariance derivative is the same actual retained derivative,
so the Gram cap is unchanged.  That lemma also preserves
\(\Tr P\le\sqrt k\) and \(\Tr R^2\le k\), using the same retained
label count \(k\) and the same \(\theta\).
All inequalities above therefore apply to this local response.
The local estimate applies at this subnormalized $S_0$ by its stated
hypotheses.
If the source support is zero, all covariance-weighted forces \(B'_a\)
vanish and \(\Tr(CJ)=0\). This proves the theorem for every covariance
\(0\preceq C\preceq I_k\).
\end{proof}

With $L=4096$ and $\theta=1$, the coefficient in the bound is $770$.
The spectral split above is an estimating device for the response; the
algorithm enforces only the covariance derivative cap. We now use that
cap to construct a walk whose progress can be charged to covariance
withdrawal.

\section{From local movement to a full signing}
\label{sq:ms:short-section}

We have proved the response estimate needed for a prepared covariance.
We now turn it into a signing. The argument separates pathwise radial
progress from expected spectral cost: the former limits the number of
accepted epochs, while the latter makes acceptance sufficiently likely.
We first choose a legal movement covariance, then specify the finite
algorithm and account for preparation, rounding, acceptance, and resets.

\subsection{A large subspace on which the covariance is at least one half}
For a subspace $U$, write $P_U$ for its orthogonal projection. At a
prepared state, retain the spectral subspace
\[
 U=\ran\mathbf 1_{[1/2,1]}(C).
\]
On $U$, the stored covariance dominates one half of the identity. We
intersect $U$ with the legal movement subspace $V$, and use one half of
the projection onto that intersection. This makes every retained
movement direction equally likely. The trace ledger will show that
$C$ remains sufficiently close to its initial identity for this
intersection to have large dimension.

\begin{lemma}[The scaled projection covariance]\label{sq:fi:short}
Let $0\preceq C\preceq I_\ell$, let $V\subseteq\R^\ell$ be a
subspace, and put
\[
 U=\ran\mathbf 1_{[1/2,1]}(C),\qquad W=U\cap V,\qquad Q=\tfrac12P_W.
\]
Then $0\preceq Q\preceq C$, $\ran Q\subseteq V$, and
\begin{equation}\label{sq:fi:short-trace}
 \Tr Q=\tfrac12\dim W
 \ge\tfrac12\bigl(2\Tr C-\rank C-\operatorname{codim}V\bigr)
 \ge\tfrac12\bigl(2\Tr C-\ell-\operatorname{codim}V\bigr).
\end{equation}
\end{lemma}
\begin{proof}
The spectral theorem gives $C\succeq\tfrac12P_U$. Since $W\subseteq U$,
$P_W\preceq P_U$, proving the order bound. For each positive eigenvalue
$\lambda\in(0,1]$,
\[
 \mathbf 1_{\{\lambda\ge1/2\}}\ge2\lambda-1.
\]
Summing over the positive spectrum gives
$\dim U\ge2\Tr C-\rank C$. Intersecting with $V$ removes at most
$\operatorname{codim}V$ dimensions. Taking the trace of
$Q=\tfrac12P_W$ proves the remaining assertions.
\end{proof}

The same estimate can be used when every original coordinate is kept in
the coefficient space throughout the proof. The relevant dimension is
the number of labels live at the start of the epoch.

\begin{lemma}[The bound relative to the live labels]\label{walk:live-flat-trace}
Suppose $0\preceq C\preceq I_n$, $\rank C\le\ell$, and
$Ce_i=0$ for every $i$ in the initially frozen set $F_0$. Let $F$ be
the current frozen set and define
\[
 V=\{z\in\R^n:z_i=0\ (i\in F),\ z^Tx=0\},\qquad
 Q=\tfrac12P_{\ran\mathbf 1_{[1/2,1]}(C)\cap V}.
\]
Then
\begin{equation}\label{walk:live-flat-bound}
 \Tr Q\ge\tfrac12\bigl(2\Tr C-\ell-|F\setminus F_0|-1\bigr).
\end{equation}
\end{lemma}
\begin{proof}
The high spectral subspace of $C$ is perpendicular to every $e_i$ with
$i\in F_0$, because those vectors belong to $\ker C$. Within that
subspace, only the equations for $F\setminus F_0$ and the single
radial equation can remove dimensions. Their total rank is at most
$|F\setminus F_0|+1$. The preceding eigenvalue count, with
$\rank C\le\ell$, gives \eqref{walk:live-flat-bound}.
\end{proof}

\begin{lemma}[Computing the legal projection]\label{sq:lem:scalarshort}
Suppose $V=\{z:v_j^Tz=0,\ 1\le j\le b\}$. Start with
$P_0=\mathbf 1_{[1/2,1]}(C)$. For each $j$, put
$a_j=v_j^TP_{j-1}v_j$. If $a_j=0$, leave the matrix unchanged;
otherwise set
\[
 P_j=P_{j-1}-\frac{(P_{j-1}v_j)(P_{j-1}v_j)^T}{a_j}.
\]
Then $P_b=P_{U\cap V}$ and the movement covariance is $Q=P_b/2$.
After forming $P_0$ by EVD, the constraint scan uses
$O(b(\ell+1)^2)$ scalar operations and divides only by positive numbers.
\end{lemma}
\begin{proof}
Inductively $P_{j-1}$ is an orthogonal projection. Thus
$a_j=\|P_{j-1}v_j\|_2^2$. If it vanishes, every vector in the current
range already satisfies the next equation. Otherwise
$w=P_{j-1}v_j$ is a nonzero vector in that range, and
$P_{j-1}-ww^T/\|w\|_2^2$ projects onto the part perpendicular to $w$.
For a vector $z$ in the current range, $w^Tz=v_j^Tz$.
The update therefore intersects that range with $v_j^\perp$.
This proves the formula by induction, including repeated or dependent
constraints. Each update requires one matrix--vector product, one
quadratic form, and one rank-one subtraction.
\end{proof}

\subsection{Analytic parameters and routine contracts}\label{walk:contracts}
The analysis below uses the finite algorithm specified in
Section~\ref{intro:algorithm}. We first state the parameter relationships
and routine guarantees it needs. This separates the progress and
probability argument from the precision chosen to implement it.

For an input with $n,m\ge1$, put $s=n+D+2$, $L=4096$,
$\delta=2^{-13}$, $a_0=(16384n)^{-1}$, and $\epsilon_A=10^{-4}$.
The rounding margin is $a_0$. Let $B\ge2$ be a response constant, meaning
that a covariance derivative cap $L/\sqrt\ell$ implies
$\Tr(CJ)\le B\sqrt\ell$ for each epoch size $\ell\le n$.
For the square profile, $B=770$; the rectangular value is proved in
Section~\ref{rect:parameters}. Define
\[
 \tau=(B+1)^{-1},\quad K=\lceil64(B+1)+129\rceil,
 \quad M=K(n+1),\quad r=M+b+1,
\]
where $b\ge1$ is the requested confidence parameter.
The trial clock has horizon $\tau$, a phase uses at most $K$ accepted
epochs, and every epoch request has at most $r$ trials.

Let $M_C$ bound the second covariance derivative on legal preparation
segments and $M_4$ bound the fourth derivative on movement curves, as
in Theorem~\ref{reg:main}. Choose
\[
 h=\min\{a_0/(2\sqrt n),\ 1/2,\ \sqrt{\tau/2},
                        \sqrt{2^{-40}/M_4}\},
 \qquad
 t_\ell=L/\sqrt\ell,\qquad
 \alpha_\ell=\min\{\delta/2,t_\ell/(2M_C)\}.
\]
An epoch retains its $\ell\ge32$ initial live labels. Its clock $T$ starts
at zero and increases by $h^2$ after each movement. The trial stops when
$T+h^2>\tau$, at least $\ell/64$ labels have frozen, or the total
preparation loss $d$ exceeds $\ell/64$. After each movement it rounds
coordinates within $a_0$ of a face. The potential and supporting-plane
reports used for its terminal test are defined in
\eqref{walk:reported-excess}; the thresholds are $16.5\sqrt\ell$ and
$4.5\sqrt\ell$, together with $d\le\ell/64$.

The following are the contracts later verified for the numerical routines.
\begin{itemize}[leftmargin=*,itemsep=3pt]
\item Preparation holds $H$ fixed and decreases $C$. Its paid decrements
have trace $\alpha_\ell$ and lower the potential by at least
$t_\ell\alpha_\ell/2$. The total trace of discarded small eigenvalues
is at most $\ell/1024$. At a movement, the prepared covariance has its
positive eigenvalues at least $2\delta$ and satisfies the cap.
\item For the prepared covariance, retain its eigenspace $U$ with
eigenvalues at least $1/2$ and intersect it with the frozen and radial
constraints to obtain $W$. The movement covariance is $Q=P_W/2$.
Writing $k=\dim W$, the routine $\proc{Sample}(Q)$ chooses
$\pm\sqrt{k/2}\,u_j$ uniformly from an orthonormal EVD basis of $W$
and the two signs. Its covariance is $Q$ and its squared length is
$\Tr Q=k/2$.
\item Source-free and endpoint value reports have error at most
$\epsilon_A/3$. The reported source-free density at a saved center is
within $\epsilon_A/(3s^2)$ of its maximizing density in Hilbert--Schmidt
norm. These reports define the acceptance test in Algorithm~\ref{alg:epoch}.
\end{itemize}
The proof below shows what these guarantees accomplish. Sections
\ref{query:section}--\ref{impl:main} prove that the stated SDP algorithm
provides them with polynomial work, so they are intermediate interfaces
in the proof rather than assumptions of the final theorem.
\space

\subsection{One small step and its natural time scale}\label{walk:small-step}

The algorithm has two separate obligations. It must make enough progress
toward the cube's vertices, and the states it accepts must have affordable
spectral cost. Orthogonality handles the first pathwise; conditional
moment bounds and the acceptance test handle the second. We establish
these facts once for the two profiles. In this section $E$ denotes
either profile in \eqref{intro:potential}, and $B\ge2$ is its response
constant from the scalar recipe. For the square profile the preceding
matrix estimates prove
\begin{equation}\label{walk:response-contract}
 \Gamma|_{\ran C}\preceq t_\ell I
 \quad\Longrightarrow\quad 0\preceq J,\qquad
 \Tr(CJ)\le B\sqrt\ell.
\end{equation}
For the rectangular profile this implication will be supplied by
Proposition~\ref{rect:ridge-response}. The other shared facts are
\begin{equation}\label{walk:potential-contract}
 \begin{gathered}
 f\text{ is convex},\quad E(H,C)\ge f(H)\ge\lambda_{\max}(H),\\
 0\le E(H,C)-f(H)\le2\sqrt\ell,\qquad
 |E(H,C)-E(H',C)|\le\|H-H'\|_{\op}.
 \end{gathered}
\end{equation}
Covariance monotonicity implies $\Gamma\succeq0$ on its supported
domain. All centers include the contributions of frozen coordinates.

\begin{lemma}[Symmetric movement and matched withdrawal]\label{walk:step}
Let $Q=\tfrac12P_W\preceq C\preceq I_\ell$, write
$k=\dim W$ and $w=\Tr Q=k/2>0$, and suppose $W$ vanishes on frozen
labels and is perpendicular to the current retained vector $x_{I_0}$.
The uniform signed sample $\xi=\proc{Sample}(Q)$
satisfies
\[
 \E\xi=0,\quad \E\xi\xi^T=Q,\quad \|\xi\|_2^2=w,
 \quad \xi\in\ran Q.
\]
For sufficiently small $h$, the step before rounding is
\[
 x'_{I_0}=x_{I_0}+h\xi,\qquad C'=C-h^2Q.
\]
It preserves the covariance range and gives, for every sample,
\begin{equation}\label{walk:radial}
 \|x'\|_2^2-\|x\|_2^2=h^2\Tr Q=\Tr(C-C').
\end{equation}
If $|e_\xi^{(4)}(u)|\le M_4$ for $|u|\le h$, where
$e_\xi(u)=E(H+u\A(\xi),C-u^2Q)$, then
\begin{equation}\label{walk:drift}
 \left|\E E(H+h\A(\xi),C-h^2Q)-E(H,C)
       -h^2\Tr\bigl(Q(J-\Gamma)\bigr)\right|
 \le M_4h^4/24.
\end{equation}
\end{lemma}
\begin{proof}
Let $u_1,\ldots,u_k$ be the orthonormal basis of $W$ used by the
sampler. The two signs cancel the mean, and
\[
 \E\xi\xi^T=\frac1k\sum_{j=1}^k\frac k2u_ju_j^T
 =\tfrac12P_W=Q.
\]
Every sampled vector has length $\sqrt w$ and belongs to $W$.
Orthogonality to
$x_{I_0}$ gives \eqref{walk:radial}. Also
$C\succeq C-h^2Q\succeq(1-h^2)C$, so $h<1$ preserves its
kernel and range. A sufficiently small $h$ keeps both signs of every
sample inside the cube.

At zero the covariance velocity is zero and its acceleration is $-2Q$.
The chain rule and the envelope identities therefore give
\[
 e_\xi'(0)=\Tr(S\A(\xi)),\qquad
 e_\xi''(0)=2J[\xi,\xi]-2\Tr(\Gamma Q),
\]
where $S$ is the current optimizing density. The second derivative
includes the response of that optimizer. Since
$e_{-\xi}(u)=e_\xi(-u)$, Taylor expansion through degree three and
symmetry cancel the linear and cubic terms. Averaging
$J[\xi,\xi]$ gives $\Tr(QJ)$. The fourth-order remainder for each
sample is bounded by $M_4h^4/24$, proving the claim.
Fixed-face smoothness justifies the derivatives even when the physical
source is singular; Appendix~\ref{reg:section} gives the uniform bound.
\end{proof}

The natural clock advances by $h^2$, not $h$. Dividing
\eqref{walk:radial} by this clock increment gives radial rate $\Tr Q$;
dividing \eqref{walk:drift} gives potential drift
$\Tr(Q(J-\Gamma))+O(M_4h^2)$. This identifies the small-step time scale: radial progress is pathwise,
while spectral control comes from averaging the two signs. The rest of
the argument turns those distinct forms of control into a useful
accepted epoch.

\subsection{Covariance loss and pathwise progress}\label{walk:invariant-section}

Preparation consumes covariance without increasing $\|x\|_2^2$.
To show that enough covariance survives for movement, we distinguish
that consumption from the matching withdrawal attached to each sample.
The following preparation contract supplies the distinction.
Section~\ref{impl:main} verifies them for the routines just specified.
Cleanup removes eigenvalues below $2\delta$, and their cumulative trace
is at most $\ell/1024$. Each other removal has the form
$C\leftarrow C-\alpha_\ell vv^T$ for a supported unit vector and
decreases $E$ by at least $t_\ell\alpha_\ell/2$. On return, either
the epoch has excessive covariance loss or the cap in
\eqref{walk:response-contract} holds. Every positive covariance
eigenvalue at a movement is at least $2\delta$. The latter statement
holds because cleanup is also performed on the final pass through
\proc{Prepare}.

For the proof only, separate the recorded loss $d$ into paid loss $a$
and dust $u$. Let $v$ be total movement withdrawal, and let $r_c$ be
the sum of the absolute coordinate changes caused by rounding. These
four nonnegative scalars start at zero. The algorithm need store only
their sum $d=a+u$. Also write $y$ for the sum of centered coefficient
movements, extended by zero outside $I_0$, so the stored matrix $Y$
is $\A(y_{I_0})$.

\begin{lemma}[Invariant and terminal alternatives]\label{walk:invariant}
Every state of Algorithm~\ref{alg:epoch} is in the cube and preserves
frozen signs. It satisfies
\begin{align}
 \Tr C+a+u+v&=\ell,& \ell T/16\le v&\le\ell T,\label{walk:ledger}\\
 u&\le\ell/1024,& r_c&\le a_0\ell,\label{walk:rounding}\\
 \|x\|_2^2&\ge\|x_*\|_2^2+v,&
 \|x-x_*-y\|_1&\le r_c.\label{walk:progress}
\end{align}
Before each movement, $\Tr Q\ge\ell/16>0$. A trial has at most
$\lfloor\tau/h^2\rfloor$ movements. If it terminates with
$d\le\ell/64$, then either at least $\ell/64$ new labels have frozen,
or its elapsed time satisfies $T>\tau-h^2\ge\tau/2$.
\end{lemma}
\begin{proof}
Preparation holds the point fixed and subtracts exactly the covariance
trace charged to $a+u$. A movement subtracts $h^2\Tr Q$, charged to
$v$, so the trace identity is preserved. Since $Q\preceq I_\ell$,
$v\le\ell T$. Before a movement the stopping tests give
$a+u\le\ell/64$ and fewer than $\ell/64$ new frozen coordinates.
The subspace $V$ has codimension at most that number plus one.
The full trace identity, together with $v\le\ell T\le\ell\tau$, gives
\[
 \Tr C=\ell-a-u-v\ge\ell(1-1/64-\tau).
\]
Consequently Lemma~\ref{sq:fi:short} gives
\begin{equation}\label{walk:flat-positive}
 \Tr Q\ge\tfrac12\bigl(2\Tr C-\ell-\ell/64-1\bigr)
 \ge\tfrac12\bigl(\ell(1-2\tau-3/64)-1\bigr)
 \ge\ell/16>0,
\end{equation}
because $\tau=1/(B+1)\le1/3$ and $\ell\ge32$. This establishes
normalization of the next draw and, by summation, $v\ge\ell T/16$.
The estimate uses the covariance remaining near its initial trace;
a bound as weak as $\Tr C\ge\ell/8$ would not justify this cutoff.
If the proof retains all original labels in its matrix coordinates,
Lemma~\ref{walk:live-flat-trace} gives the identical estimate: initially
frozen coordinates already lie in the covariance kernel, so only new
frozen labels and the radial equation are charged.

Initially every live coordinate is farther than $a_0$ from a face.
This remains true after each rounding operation. As
$h\|\xi\|_2\le h\sqrt n\le a_0/2$, a movement cannot leave the
cube. Frozen coordinates do not move. Rounding increases absolute
values, costs at most $a_0$ per newly frozen coordinate, and happens
at most once to each label. It therefore proves the rounding bound
and, with \eqref{walk:radial}, both progress statements.
The covariance floor is preserved between preparations since
$C-h^2Q\succeq(1-h^2)C$ and $h^2\le1/2$.

Each movement adds exactly $h^2$ to $T$, whose stopping rule keeps
$T\le\tau$. Preparation is finite by its trace budget, as proved in
Section~\ref{impl:runtime}. The trial thus terminates. If neither
loss nor freezes triggered stopping, $T+h^2>\tau$; the recipe has
$h^2\le\tau/2$, which proves the alternative.
\end{proof}

\subsection{The supporting-plane excess}\label{walk:acceptance-section}

Fix a saved epoch state $x_*$ and put $H_*=H(x_*)$. Let $S_*$ be
the exact maximizing density of $f(H_*)$. This is a mathematical
reference density; the algorithm uses its feasible approximation.
For a trial state define
\begin{equation}\label{walk:excess}
 \Psi=E(H,C)-f(H_*)-\Tr(S_*(H-H_*)),\qquad
 M_* =\Tr(S_*Y).
\end{equation}
The subscript in $M_*$ distinguishes this martingale from the global
integer trial-budget parameter $M$. Convexity of $f$ gives $\Psi\ge0$:
the affine function $f(H_*)+\Tr(S_*(H-H_*))$ is a supporting plane
of $f$, and $E\ge f$. This subtraction is useful even when the
saved center is large, because initially $\Psi\le2\sqrt\ell$.

Let $\widehat f_*$ and $\widehat S_*$ be the reported source-free value
and density at $H_*$, and let $\widehat E$ be the reported endpoint
value. The algorithm computes
\begin{equation}\label{walk:reported-excess}
 \widehat\Psi=\widehat E-\widehat f_*
       -\Tr\bigl(\widehat S_*(H-H_*)\bigr),\qquad
 \widehat M=\Tr(\widehat S_*Y).
\end{equation}
It accepts when $d\le\ell/64$, $\widehat\Psi\le16.5\sqrt\ell$, and
$\widehat M\le4.5\sqrt\ell$. We first prove a stronger good event for
the exact quantities and then check that report errors fit within the
acceptance margins.

\begin{lemma}[Moment bounds for the stopped trial]\label{walk:moments}
At the terminal state of a trial,
\begin{equation}\label{walk:moments-eq}
 \E\left[\Psi+\frac{2048}{\sqrt\ell}a\right]
 \le\frac{151}{50}\sqrt\ell,
 \qquad \E M_*^2\le\ell.
\end{equation}
With probability at least $499/800$, the terminal state satisfies
\begin{equation}\label{walk:good-event}
 d\le\ell/64,\qquad \Psi\le16\sqrt\ell,
 \qquad M_*\le4\sqrt\ell.
\end{equation}
\end{lemma}
\begin{proof}
A paid cut decreases $\Psi$ by at least
$t_\ell\alpha_\ell/2=2048\alpha_\ell/\sqrt\ell$; cleanup cannot
increase it. After preparation, \eqref{walk:response-contract},
$0\preceq Q\preceq C$, and $\Gamma\succeq0$ imply
\[
 \E[\Delta E\mid\text{state before sampling}]
 \le h^2B\sqrt\ell+M_4h^4/24
 \le h^2(B\sqrt\ell+2^{-40}).
\]
The fixed supporting plane has zero conditional mean change on a
centered movement. Rounding of total coordinate size $b_s$ changes
both $E$ and the plane by at most $b_s$, since the atoms are
contractions and $S_*$ is a density. Summing over the finite movement
horizon, with stopped states made absorbing, yields
\[
 \E\left[\Psi+\frac{2048}{\sqrt\ell}a\right]
 \le2\sqrt\ell+(B\sqrt\ell+2^{-40})\tau+2a_0\ell
 \le\frac{151}{50}\sqrt\ell.
\]
Here $B\tau<1$, $2a_0\ell\le1/8192$, and $\ell\ge32$.

For the other moment, set $g_i=\Tr(S_*\A_i)$ on the retained
labels, so $|g_i|\le1$. The increment of $M_*$ is $hg^T\xi$.
It has conditional mean zero and variance
$h^2g^TQg\le h^2\ell$. The cross term in its squared running sum
has zero conditional expectation. Summing these equalities over the
same finite horizon gives $\E M_*^2\le\ell\E T\le\ell$.

If $d>\ell/64$, the dust bound implies $a>15\ell/1024$, hence
$(2048/\sqrt\ell)a>30\sqrt\ell$. Markov's inequality bounds
this failure probability by $151/1500<1/8$. The two other bad events
have probabilities at most $151/800$ and $1/16$, respectively,
using nonnegativity of $\Psi$ and the second-moment bound for $M_*$.
Their union has probability at most
$1/8+151/800+1/16=301/800$, proving \eqref{walk:good-event}.
\end{proof}

The preceding good event is stated using the exact density, but the
algorithm sees only feasible approximate densities. The acceptance
thresholds have room for their error. Conversely, passing those thresholds
gives a deterministic spectral bound for the returned endpoint. This is
why a favorable expectation suffices despite individual movements being
allowed to increase the potential.

\begin{lemma}[The observable test]\label{walk:acceptance}
The algorithm's reports satisfy
\[
 |\widehat\Psi-\Psi|\le\sqrt\ell/100,
 \qquad |\widehat M-M_*|\le\sqrt\ell/100.
\]
Each trial is accepted with probability at least $499/800$.
An accepted trial increases the potential by at most $25\sqrt\ell$.
Resetting its terminal covariance to identity on the remaining live
labels adds at most $2\sqrt\ell$. Thus the accepted trial and reset
together increase the potential by at most $27\sqrt\ell$ relative
to the saved epoch state.
\end{lemma}
\begin{proof}
The report contract in Section~\ref{walk:contracts} gives source-free
and endpoint value errors at most $\epsilon_A/3$, and source-free density
error at most $\epsilon_A/(3s^2)$. These bounds are obtained from primal
objective gaps in Lemma~\ref{query:gap-density} and implemented in
Lemma~\ref{impl:acceptance}.

The invariant bounds $\|H-H_*\|_{\op}\le2\ell$ and
$\|Y\|_{\op}\le2\ell+a_0\ell$. Multiplication by $\sqrt D$
bounds their Hilbert--Schmidt norms by $3s^2$. Pairing either
matrix with the density error therefore costs at most $\epsilon_A$.
Thus the error in $\widehat\Psi$ is at most $5\epsilon_A/3$ and
that in $\widehat M$ at most $\epsilon_A$, both smaller than
$\sqrt\ell/100$. The event \eqref{walk:good-event} passes the
thresholds $16.5\sqrt\ell$ and $4.5\sqrt\ell$.

Conversely an accepted trial has
$\Psi\le16.51\sqrt\ell$ and $M_*\le4.51\sqrt\ell$.
Rounding contributes at most $a_0\ell$ to the supporting plane.
Subtracting the initial excess, which is nonnegative, shows
\[
 E(H,C)-E(H_*,I_\ell)
 \le\Psi+M_*+a_0\ell<25\sqrt\ell.
\]
At the new center, replacing any covariance by identity on
$\ell'\le\ell$ remaining labels costs at most $2\sqrt{\ell'}$:
both versions are at least $f(H)$, and the new one is at most
$f(H)+2\sqrt{\ell'}$. This source-free comparison also accounts for the change of retained
labels between the two covariance matrices.
\end{proof}

\subsection{Phases and the square discrepancy bound}\label{walk:phases}

A single accepted epoch need not freeze a constant fraction of its
labels. If it does not, it must have run for a definite amount of its
natural clock, and then its radial progress is substantial. The squared
norm cannot increase indefinitely inside a cube. Grouping accepted epochs
into phases therefore gives a geometric decrease in the live count.

\begin{lemma}[A phase halves its live count]\label{walk:phase}
If no requested epoch exhausts its retries, a phase starting with $k$
live labels halves that count within $K$ accepted epochs, or rounds
the last fewer than $32$ labels and finishes. Its potential increase
is at most $(27K+64)\sqrt k$.
\end{lemma}
\begin{proof}
Before the phase completes, every epoch starts with $\ell>k/2$.
Classify its accepted endpoint as freeze-ended if at least $\ell/64$
new labels froze; classify all other accepted endpoints as time-ended.
Each freeze-ended epoch fixes more than $k/128$ labels. Fewer than
$128$ such epochs can occur before completion. By
Lemma~\ref{walk:invariant}, a time-ended epoch has $T\ge\tau/2$
and increases squared norm by at least
$\ell\tau/32>k\tau/64$. Restrict to the $k$ labels that were live
at the start of this phase. All other coordinates are fixed throughout
the phase, so this restriction has exactly the same norm increments.
Its squared norm is nondecreasing and at most $k$, including its labels
that have since frozen.
Fewer than $64/\tau=64(B+1)$ time-ended epochs can therefore
occur in an unfinished phase. The chosen $K$ exceeds the sum of
these bounds.

Each accepted epoch with covariance replacement costs at most
$27\sqrt\ell\le27\sqrt k$. Rounding the last fewer than $32$
coordinates changes the center by norm at most $64$; eliminating
their covariance cannot increase the potential. For $k\ge1$ this
adds at most $64\sqrt k$, proving the claim.
\end{proof}

\begin{proposition}[The common signing estimate]\label{walk:assembly}
Assume \eqref{walk:response-contract} and
\eqref{walk:potential-contract}, and the quantitative routine contracts
used above. Every returned signing satisfies
\begin{equation}\label{walk:global-bound}
 \left\|\sum_i\sigma_i A_i\right\|_{\op}
 \le E(0,I_n)+4(27K+64)\sqrt n.
\end{equation}
The probability of returning failure is at most $2^{-b}$.
\end{proposition}
\begin{proof}
Every completed phase halves an integer live count, so $n+1$ phases
suffice. At boundaries the potential is evaluated at the full center
and identity on precisely the remaining labels. The change of label
set is accounted for in Lemma~\ref{walk:acceptance}; frozen terms
stay in the center. The phase costs sum to at most
\[
 (27K+64)\sum_{j\ge0}\sqrt{n2^{-j}}
 <4(27K+64)\sqrt n.
\]
At termination the covariance is zero and all coordinates are signs.
The two-sided encoding identity and \eqref{walk:potential-contract} prove
\eqref{walk:global-bound} pathwise.

There are at most $M=K(n+1)$ adaptive epoch requests. At each request,
conditional on the saved state and all preceding history, a trial has
rejection probability at most $301/800<1/2$. Fresh draws on each
retry give failure probability at most $2^{-r}$. The union bound
over the at most $M$ requests applies to these conditional bounds.
With $r=M+b+1$ and $M\le2^M$, total failure is at most
$M2^{-M-b-1}\le2^{-b}$. A positive success probability also proves
the existence statement.
\end{proof}

\begin{corollary}[The square part of Theorem~\ref{main:combined}]
\label{walk:square-conclusion}
For $m\le n$, Algorithm~\ref{alg:full} has the discrepancy guarantee
$10^8\sqrt n$ and failure probability at most $2^{-b}$.
\end{corollary}
\begin{proof}
The square-root response estimate gives $B=770$ and $K=49{,}473$.
The initial potential is at most
$2\sqrt n+2\sqrt D\le5\sqrt n$, since $D=2m\le2n$.
Substitution in \eqref{walk:global-bound} gives
$[5+4(27\cdot49{,}473+64)]\sqrt n<10^8\sqrt n$.
Sections~\ref{query:section}--\ref{impl:main} establish the numerical
contracts used here and bound the work on every execution.
\end{proof}

\section{The rectangular extension: changing the regularizer}\label{rect:extension}

The square-root potential pays $2\sqrt D$ at zero center, which is
appropriate when $D$ is comparable to $n$. For large $D/n$, a suitable
generalized Tsallis power regularizer replaces this by $\theta D^q/(1-q)$.
The cost is an increased
response bound, proportional to $k^{1-q}/(\theta q)$ on $k$ labels.
The aim of this part is to prove that tradeoff and choose $q$ and
$\theta$ to balance its two sides.

The square proof has separated the walk from its analytic input. The walk
needs a bounded covariance derivative to control the sum of center
curvatures over the remaining covariance. Its geometric constraints,
matched covariance withdrawal, and progress identity do not depend on
the square-root exponent. We now change that exponent and prove the
corresponding response estimate. The new ingredient is a trace
interpolation that keeps the transport and density coupled.

Throughout this section $n$ is the number of original matrices, $m\ge n$
is their dimension, and $D=2m$. The lifted atoms $\A_i$, center $H$, and
covariance map $\eta_C$ retain their earlier definitions. A local count
$k$ means the number of original labels in an epoch, including labels
that freeze during that epoch; it is not the rank of $C$. A local covariance
satisfies $0\preceq C\preceq I_k$. All matrices
and coefficient covariances in the real calculation have real entries.
The full space $\R^{D\times D}$ has inner product
$\langle X,Y\rangle=\Tr(X^TY)$; the symmetric matrices form a subspace.
Left-right multiplication estimates are made on the full space, and
inverse responses are restricted to symmetric perturbations when needed.

Fix $0<q\le1/2$ and $\theta>0$. Define
\[
 r_{q,\theta}(S)=\frac{\theta}{1-q}\Tr S^{1-q},\qquad
 E_{q,\theta}(H,C)=\max_{S\succeq0,\ \Tr S=1}
 \{\Tr(HS)+2F(S,\eta_C(S))+r_{q,\theta}(S)\},
\]
where $F$ is the unsquared fidelity defined in
Lemma~\ref{sq:fo:fidelity}. Write $E=E_{q,\theta}$ and $f(H)=E(H,0)$
within the base-power analysis below. At $q=1/2$ and $\theta=1$ this
is exactly the square potential. The small square-root term needed
for the rectangular implementation will be added after proving the
pointwise response estimate.

For a coefficient vector $z$, put $\A(z)=\sum_i z_i\A_i$. At a fixed state,
let $J$ and $\Gamma$ denote the real symmetric coefficient forms
\[
 J[z,z]=\tfrac12D_H^2E(H,C)[\A(z),\A(z)],\qquad
 D_CE(H,C)[V]=\Tr(\Gamma V).
\]
Covariance perturbations $V$ are restricted to the current range of $C$.
The physical source support can be smaller than the full physical space;
we will keep the density on the full space throughout.

\subsection{The power potential and its optimizer}

We first justify the matrix-power derivatives used in this part. The
calculation below differentiates a resolvent integral, so it also covers
repeated eigenvalues without choosing differentiable eigenvectors.

\begin{lemma}[Differentiating powers on the positive cone]
\label{rect:power-differentiation}
For $-1<a<1$, $a\ne0$, the map $S\mapsto S^a$ is smooth on the
positive definite cone. In an eigenbasis $S=\operatorname{diag}(s_i)$,
its derivative in a symmetric direction $X$ is
\[
 (D(S^a)[X])_{ij}=d_a(s_i,s_j)X_{ij},\qquad
 d_a(u,v)=
 \begin{cases}(u^a-v^a)/(u-v),&u\ne v,\\
 au^{a-1},&u=v.
 \end{cases}
\]
In particular, $D(\Tr S^a)[X]=a\Tr(S^{a-1}X)$.
\end{lemma}
\begin{proof}
For $0<b<1$, define the positive scalar
\[
 c_b=\left(\int_0^\infty\frac{u^{b-1}}{1+u}\,du\right)^{-1}.
\]
The integral is finite at both endpoints. Substitution $t=su$ for a
scalar $s>0$, followed by spectral calculus, gives
\[
 S^b=c_b\int_0^\infty t^{b-1}S(S+tI)^{-1}\,dt,
 \qquad
 S^{-b}=c_{1-b}\int_0^\infty t^{-b}(S+tI)^{-1}\,dt.
\]
On a neighborhood with a common positive spectral lower bound, these
integrals and all their derivatives converge uniformly. For example,
the derivative of the matrix factor $S(S+tI)^{-1}$ in the first integrand is
$t(S+tI)^{-1}X(S+tI)^{-1}$, while the derivative of the second
resolvent is $-(S+tI)^{-1}X(S+tI)^{-1}$. Higher derivatives insert
additional resolvents, which give integrable bounds at infinity and
remain uniformly bounded near zero. Differentiation under the integral
therefore proves smoothness. In an eigenbasis the resolvent products
multiply $X_{ij}$ by $1/((s_i+t)(s_j+t))$. Subtracting the two scalar
integral formulas and dividing by $s_i-s_j$ gives the claimed divided
difference; if the eigenvalues agree, scalar differentiation gives its
continuous value. Taking the trace proves the final assertion.
\end{proof}

\begin{lemma}[Potential facts and fixed-face regularity]\label{rect:basic}
The optimizing density is unique and positive definite. The functions
$E$ and $f$ are convex and $1$-Lipschitz in $H$, and $E$ is monotone
in $C$. Moreover,
\begin{gather*}
 0\le E(H,C)-f(H)\le2\sqrt k,\qquad E(H,C)\ge\lambda_{\max}(H),\\
 E(0,I_n)\le2\sqrt n+\frac\theta{1-q}D^q.
\end{gather*}
On a fixed coefficient support, the potential and its optimizer are
smooth while the coefficient covariance is positive definite on that support.
If $v=\|\eta_C(I)\|$, the density satisfies
\begin{equation}\label{rect:density-floor}
 S\succeq
 \left(\frac{\theta}{2\|H\|+2\sqrt v+\theta D^q}\right)^{1/q}I.
\end{equation}

\end{lemma}
\begin{proof}
For $a\in(0,1)$, write $s_i$ for the eigenvalues of a faithful $S$.
Lemma~\ref{rect:power-differentiation} shows that
$-D^2(\Tr S^a)$ is positive definite on the faithful cone: its diagonal
matrix-space coefficients are $a(1-a)s_i^{a-2}$, and its off-diagonal
coefficients are the positive divided differences of $-a s^{a-1}$.
This also proves strict concavity on every nonconstant segment after
restriction to the sum of its endpoint supports.
The objective in the definition of $E_{q,\theta}$ is thus strictly concave.
At a singular density, mixing with $I/D$ gains order $u^{1-q}$ in the
regularizer. The linear loss is $O(u)$ and, by fidelity concavity,
so is any free-term loss. Thus a maximizer cannot be singular.
Existence follows from compactness, and the unique maximizer is continuous
by taking subsequences in this same compact density set.

Convexity and the Lipschitz bound follow by maximizing affine functions
whose linear coefficients are densities. The source order and fidelity
monotonicity give monotonicity in $C$. The trace bound in
Lemma~\ref{sq:fo:fidelity} gives the $2\sqrt k$ excess, while testing a
pure top eigenstate gives $E\ge\lambda_{\max}$. Jensen's scalar
inequality gives $\Tr S^{1-q}\le D^q$, proving the initial bound.

On a fixed coefficient support $V$, the physical source support is
$K(V)=\operatorname{span}_{z\in V}\ran(\sum_i z_i\A_i)$.
Symmetry of the factors and the faithfulness of $S$ show that this
support is independent of the particular positive coefficient covariance on $V$.
The fidelity compression identity makes the free term smooth on this
fixed physical support. Write its density gradient as $G$. Monotonicity
gives $G\succeq0$, and degree-one homogeneity gives
$\Tr(SG)=2F(S,\eta_C(S))$. Stationarity, with scalar trace multiplier $\lambda$, is
\[
 H+G+\theta S^{-q}=\lambda I.
\]
Pairing with $S$ shows
$\lambda\le\|H\|+2\sqrt v+\theta D^q$. Dropping $G\succeq0$
then proves \eqref{rect:density-floor}.
The strictly positive regularizer curvature makes the density stationarity
equations invertible on the trace-zero tangent. The implicit-function
theorem gives smoothness. The quantitative estimates used by the algorithm will instead follow
from the additional ridge in the implementation section.
\end{proof}

\subsection{What changes in the response equations}\label{rect:response-equations}

The transport equation is unchanged. At a faithful source
$M=\eta_C(S)$, let $Z\succ0$ solve $ZMZ=S$ and set $W=Z^{-1}$.
The density equation and trace constraint are now
\[
 H+W+\eta_C(Z)+\theta S^{-q}=\lambda I,\qquad \Tr S=1.
\]
To see exactly where the new power enters, take velocities
$B=\dot H$, $V=\dot C$, $X=\dot S$, and $U=\dot Z$. Set
\[
 A_r=-D(\theta S^{-q}),\qquad
 \mathcal L(X)=WXW-\eta_C(X),\qquad
 \mathcal K(U)=WUM+MUW.
\]
Here $A_r,\mathcal L,\mathcal K$ are operators on symmetric matrix
space. Differentiating $ZMZ=S$, the density equation, and the trace
constraint gives
\begin{equation}\label{rect:coupled}
 \mathcal K U=\mathcal L X-\eta_V(S),\qquad
 A_rX+\mathcal L^*U+\dot\lambda I=B+\eta_V(Z),\qquad \Tr X=0.
\end{equation}
For example $D(Z^{-1})[U]=-WUW$, so the direct covariance velocity
and the transport response produce both terms involving $\eta_V$.
Eliminating $U$ gives the same inverse-response formula as in the
square calculation, with $A_r$ in place of the square-root curvature.
If $\iota$ includes the trace-zero symmetric matrices into all symmetric
matrices, put
\[
 A_0=\iota^*(A_r+\mathcal L^*\mathcal K^{-1}\mathcal L)\iota,
 \qquad B_0=\iota^*B.
\]
Then a pure center velocity satisfies
\begin{equation}\label{rect:optimized-response}
 \dot S=\iota A_0^{-1}B_0,\qquad
 D_H^2E[B,B]=\langle B_0,A_0^{-1}B_0\rangle,
 \qquad D_CE[V]=\Tr(Z\eta_V(S)).
\end{equation}
Positive regularizer curvature makes $A_0$ invertible. If the physical
source is singular, the transport and $\mathcal K$ are compressed to
its fixed support, exactly as in the square linearization; the full
regularizer, density, and trace constraint are retained. The response
comparison for that case is proved in Section~\ref{rect:compression-section}.

The changes of coordinates in Lemmas~\ref{sq:an:balanced}
and~\ref{sq:an:frame} involve the transport equation and the free
curvature only. Thus they apply with any positive definite regularizer
curvature. We record the resulting interface, to identify the operators
that will be estimated below.

\begin{lemma}[The shared balanced response with a new regularizer]
\label{rect:free-response}
Factor $\eta_C(X)=\sum_a B'_aXB'_a$, where
$B'_a=\sum_i(C^{1/2})_{ia}\A_i$, and let $S\succ0$, $\Tr S\le1$,
have faithful source. Define
\[
 P=W^{1/2}SW^{1/2},\qquad D_a=Z^{1/2}B'_aZ^{1/2},\qquad
 \Phi(X)=\sum_aD_aXD_a,\qquad \mathcal J_P=L_P+R_P.
\]
In the coordinate $Y=W^{1/2}(dS)W^{1/2}$ the free negative density
Hessian is $(I-\Phi)\mathcal J_P^{-1}(I-\Phi)$.
The map $\Phi$ is self-adjoint, positive, and a Hilbert--Schmidt
contraction, and $\Phi(P)=P$. If $A_{r,\mathrm{bal}}$ is the negative
regularizer Hessian in this coordinate, define
\[
 T=\mathcal J_P^{-1/2}\Phi\mathcal J_P^{1/2},\qquad
 H_r=\mathcal J_P^{1/2}A_{r,\mathrm{bal}}\mathcal J_P^{1/2},\qquad
 \mathcal A=(I-T)^*(I-T)+H_r,
\]
\[
 F_a=2^{-1/2}\mathcal J_P^{1/2}D_a,\qquad
 \Sigma=\sum_a|F_a\rangle\langle F_a|.
\]
At the optimizing density,
\begin{equation}\label{rect:gram-response}
 (C^{1/2}\Gamma C^{1/2})_{ab}=\Tr(PD_aD_b),\qquad
 \Tr(CJ)\le\Tr_{\HS}(\Sigma\mathcal A^{-1}).
\end{equation}
The right side also defines the local unconstrained response at any
of the stated subnormalized densities.
\end{lemma}
\begin{proof}
The coordinate identities are precisely the transport identities of
Lemma~\ref{sq:an:balanced}; its derivation does not use the regularizer.
The contraction assertion, including nonsymmetric real matrix directions,
is the fixed-point argument in Lemma~\ref{sq:an:balanced}.
The covariance envelope identity in \eqref{rect:optimized-response}
gives the Gram formula. Its entries are real, and symmetric because
transposition and cyclicity of trace give
$\Tr(PD_aD_b)=\Tr(PD_bD_a)$ for real symmetric $P,D_a,D_b$.
For the response formula, the physical force $B'_a$ pairs with
$dS$ as $\langle D_a,Y\rangle$. The further substitution
$Y=\mathcal J_P^{1/2}y$ turns its half-response into
$\langle F_a,\mathcal A^{-1}F_a\rangle$ after dropping the trace
constraint. Dropping a constraint only enlarges the inverse variational
supremum. Summation proves the claim as in Lemma~\ref{sq:an:frame}.
The unconstrained formula is an identity at every stated density $S$.
\end{proof}

\section{General Tsallis curvature and correlated interpolation}
\label{rect:general-tsallis}

We now return to the regularizer. In the square-root case the inverse
curvature has an exact two-term expression in left and right multiplication.
For the rectangular rate we need a variable exponent. A scalar divided
difference comparison supplies the replacement, and a weighted H\"older
inequality interpolates the resulting physical trace budgets. The
interpolation is performed before any information about the common
optimizing density is discarded.

Fix \(0<q\le 1/2\) and \(\theta>0\). We use the concave regularizer
\[
 r_{q,\theta}(S)=\frac{\theta}{1-q}\Tr S^{1-q}.
\]
At \(q=1/2\), its normalization is exactly
\(r_{1/2,\theta}(S)=2\theta\Tr\sqrt S\).

\begin{lemma}[Inverse power curvature]
\label{rect:power-inverse}
At \(S\succ0\), let \(A_{q,\theta}^{\mathrm{orig}}\) be the negative
Hessian of \(r_{q,\theta}\). In an eigenbasis of \(S\), its inverse has
entrywise multiplier \(\theta^{-1}h_q(s_i,s_j)\), where
\[
 h_q(a,b)=\frac{a-b}{b^{-q}-a^{-q}},\qquad
 h_q(a,a)=\frac{a^{1+q}}q.
\]
As positive operators on matrix Hilbert space,
\begin{equation}
 (A_{q,\theta}^{\mathrm{orig}})^{-1}
 \preceq\frac1{2\theta q}
       (L_S R_{S^q}+L_{S^q}R_S).
 \label{rect:power-inverse-bound}
\end{equation}
Equality holds in \eqref{rect:power-inverse-bound} when \(q=1/2\).
\end{lemma}

\begin{proof}
The gradient of \(r_{q,\theta}\) is \(\theta S^{-q}\). Its negative
derivative acts by the positive divided difference
\[
 \theta\frac{s_j^{-q}-s_i^{-q}}{s_i-s_j},
\]
with diagonal value \(\theta q s_i^{-1-q}\). This proves the inverse
formula. It remains to show, for \(a,b>0\),
\begin{equation}
 h_q(a,b)\le\frac{a b^q+a^q b}{2q}.
 \label{rect:scalar-kernel-bound}
\end{equation}
The case \(a=b\) is equality. By symmetry and homogeneity, the other
cases reduce to \(t=a/b>1\). Inequality
\eqref{rect:scalar-kernel-bound} is equivalent to
\[
 2q(t-1)\le(t^q-1)(t^{1-q}+1),
\]
or to \(t^{1-q}-t^q\le(1-2q)(t-1)\). Write
\(r=1-2q\in[0,1)\) and \(t=e^{2x}\), \(x>0\). Cancellation reduces
the last inequality to \(\sinh(rx)\le r\sinh x\), which follows from
convexity of \(\sinh\) on \([0,\infty)\) and \(\sinh0=0\).
The operators in \eqref{rect:power-inverse-bound} are diagonal in the
same matrix-unit basis, so the scalar inequalities imply their stated
positive-operator comparison.
\end{proof}

\subsection{Balanced quantities and an interpolation inequality}

Let \(C\preceq I_k\) be a coefficient covariance on \(k\) original real symmetric
contractions. Factor its real covariance and write
\[
 \eta_C(X)=\sum_a B'_aXB'_a.
\]
Suppose \(S\succ0\), \(\Tr S\le1\), and its physical source
\(M=\eta_C(S)\) is faithful. Let \(Z\succ0\) solve \(ZMZ=S\), and set
\begin{equation}
 \begin{gathered}
 W=Z^{-1},\qquad P=W^{1/2}SW^{1/2},\\
 D_a=Z^{1/2}B'_aZ^{1/2},\qquad
 \Phi(X)=\sum_aD_aXD_a,\\
 R_q=W^{1/2}S^qW^{1/2},\qquad
 R=R_{1/2}=W^{1/2}\sqrt S W^{1/2}.
 \end{gathered}
 \label{rect:balanced-definitions}
\end{equation}
The balanced free-curvature identities give \(\Phi(P)=P\);
\(\Phi\) is self-adjoint and a Hilbert--Schmidt contraction. The
contraction realization gives \(\eta_C(I)\preceq kI\), and hence
\begin{equation}
 \Tr P\le\sqrt k,\qquad
 \Tr(WP)=\Tr(SW^2)=\Tr M\le k,\qquad
 \Tr R^2\le k.
 \label{rect:endpoint-budgets}
\end{equation}
Indeed, the first inequality is the fidelity trace bound. For the last,
\[
 \Tr M-\Tr R^2
 =\Tr(SW^2)-\Tr(\sqrt S W\sqrt S W)
 =\tfrac12\|[\sqrt S,W]\|_{\mathrm{HS}}^2\ge0.
\]

Let \(\mathcal G=\operatorname{Ad}_{Z^{1/2}}\), so that
\(dS=\mathcal G(Y)\) is the balanced density coordinate. Write
\(A_{q,\theta}=\mathcal G^*A_{q,\theta}^{\mathrm{orig}}\mathcal G\).
Congruencing \eqref{rect:power-inverse-bound} by \(\mathcal G^{-1}\)
gives
\begin{equation}
 A_{q,\theta}^{-1}
 \preceq\frac1{2\theta q}
       (L_P R_{R_q}+L_{R_q}R_P).
 \label{rect:balanced-inverse-bound}
\end{equation}

\begin{lemma}[Correlated trace interpolation]
\label{rect:trace-interpolation}
For every physical matrix \(X\succeq0\),
\begin{equation}
 \Tr(R_qX)
 \le [\Tr(WX)]^{1-2q}[\Tr(RX)]^{2q}.
 \label{rect:trace-interpolation-bound}
\end{equation}
At \(q=1/2\), the right side is interpreted as \(\Tr(RX)\).
\end{lemma}

\begin{proof}
Set \(Y=W^{1/2}XW^{1/2}\succeq0\). In an eigenbasis of \(S\), let
\(y_i=Y_{ii}\ge0\). Weighted scalar H\"older gives
\[
 \Tr(R_qX)=\sum_i y_i(s_i^{1/2})^{2q}
 \le\left(\sum_i y_i\right)^{1-2q}
      \left(\sum_i y_i s_i^{1/2}\right)^{2q}.
\]
These sums are \(\Tr(WX)\) and \(\Tr(RX)\). If the first sum is
zero, then \(Y=0\) and the assertion is immediate. 
\end{proof}

\subsection{The low and high observed forces}

All parts of the response bound are now available. We split each force
into the corner where both eigenvalues of $P$ exceed a cutoff and its
complement. The complement is inexpensive because one of its eigenvalues
is small. The high corner is controlled by the tensor estimate. The
inverse response is a positive quadratic form, so the cross term is
bounded by the two diagonal terms. This is the same force split as in
the square proof, with a new inverse-curvature weight.

Assume the real balanced Gram satisfies
\[
 [\Tr(PD_aD_b)]_{ab}\preceq tI,\qquad t=L/\sqrt k.
\]
Fix \(c>0\), and set
\[
 \varepsilon=c/\sqrt k,\quad
 \Pi_h=1_{\{P\ge\varepsilon\}},\quad
 \Pi_l=I-\Pi_h,\quad P_l=\Pi_lP,
\]
\[
 D_a^h=\Pi_hD_a\Pi_h,\qquad
 D_a^l=D_a-D_a^h=\Pi_lD_a+\Pi_hD_a\Pi_l.
\]
Only the forcing vectors are split; all curvature inverses below act
on their full spaces.

\begin{lemma}[Low forcing for general powers]
\label{rect:low-forcing}
For real symmetric \(D\), define the unconstrained regularizer-only
half-response
\[
 Q_{q,\theta}(D)=\tfrac12\langle D,A_{q,\theta}^{-1}D\rangle.
\]
Then
\begin{equation}
 \sum_aQ_{q,\theta}(D_a^l)
 \le\frac{2c^q}{\theta q}k^{1-q}.
 \label{rect:low-forcing-bound}
\end{equation}
\end{lemma}

\begin{proof}
By \eqref{rect:balanced-inverse-bound},
\[
 Q_{q,\theta}(D)\le\frac1{2\theta q}\Tr(R_qDPD).
\]
The squared triangle inequality applied to
\(P^{1/2}D_a^lR_q^{1/2}\) gives
\begin{align*}
 \sum_a\Tr(R_qD_a^lPD_a^l)
 &\le2\Tr(R_q\Phi(P_l))\\
 &\quad+2\Tr\{\Pi_lR_q\Pi_l\Phi(P-P_l)\}\\
 &\le2\{\Tr(R_q\Phi(P_l))+\Tr(R_qP_l)\}.
\end{align*}
The last step follows from \(0\preceq\Phi(P-P_l)\preceq P\);
it does not require \(R_q\) to commute with \(P\).

The two positive matrices \(X=\Phi(P_l)\) and \(X=P_l\) satisfy
\[
 \Tr(WX)\le\Tr(WP)\le k.
\]
Moreover,
\[
 \|P_l\|_{\mathrm{HS}}^2\le\varepsilon\Tr P\le c,
 \qquad \|\Phi(P_l)\|_{\mathrm{HS}}\le\sqrt c.
\]
Consequently \(\Tr(RX)\le\sqrt{ck}\) for either choice of \(X\), by
\eqref{rect:endpoint-budgets} and Hilbert--Schmidt Cauchy--Schwarz.
Lemma~\ref{rect:trace-interpolation} bounds each resulting trace by
\[
 \Tr(R_qX)\le k^{1-2q}(\sqrt{ck})^{2q}=c^qk^{1-q}.
\]
Combining the factors proves \eqref{rect:low-forcing-bound}.
\end{proof}

The high-force estimates from the square proof depend only on the
source geometry, so they can be reused. Put
\[
 \mathcal J_P=L_P+R_P,\quad
 F_a^h=2^{-1/2}\mathcal J_P^{1/2}D_a^h,\quad
 \Sigma_h=\sum_a|F_a^h\rangle\langle F_a^h|,
\]
\[
 \beta=t/\varepsilon=L/c,\qquad
 W_h=\sum_a\Phi(D_a^h)P\Phi(D_a^h).
\]
Lemma~\ref{sq:te:tensor-lemma} and identity~\eqref{sq:te:high-frame} give
\begin{equation}
 \begin{gathered}
 \Tr_{\mathrm{HS}}\Sigma_h\le\Tr P,\qquad
 0\preceq W_h\preceq2\beta^2P,\\
 \Tr(PW_h)\le2t\beta\Tr P,\qquad
 \Tr W_h^2\le4t\beta^3\Tr P.
 \end{gathered}
 \label{rect:old-tensor-bounds}
\end{equation}
These statements use the actual real Gram cap, the self-adjoint Kraus
representation, and
\(\Phi(P)=P\). In particular they do not involve the choice of density
regularizer. By \eqref{rect:endpoint-budgets}, they imply
\begin{equation}
 \Tr_{\mathrm{HS}}\Sigma_h\le\sqrt k,\qquad
 \|W_h\|_{\mathrm{HS}}\le2L^2c^{-3/2}.
 \label{rect:high-numeric-bounds}
\end{equation}

Use the full whitened curvatures
\[
 T=\mathcal J_P^{-1/2}\Phi\mathcal J_P^{1/2},\quad
 H_{q,\theta}=\mathcal J_P^{1/2}A_{q,\theta}\mathcal J_P^{1/2},
\]
\[
 \mathcal A=(\mathrm{Id}-T)^*(\mathrm{Id}-T)+H_{q,\theta},\qquad
 K_q=T H_{q,\theta}^{-1}T^*.
\]

\begin{lemma}[High observed response for general powers]
\label{rect:high-response}
One has
\begin{equation}
 \Tr_{\mathrm{HS}}(\Sigma_hK_q)
 \le\frac{L^2c^{q-2}}{\theta q}k^{1-q}.
 \label{rect:high-response-bound}
\end{equation}
\end{lemma}

\begin{proof}
The order bound in \eqref{rect:old-tensor-bounds} and the trace budget
give
\[
 \Tr(WW_h)\le2L^2c^{-2}k.
\]
The Hilbert--Schmidt bound in \eqref{rect:high-numeric-bounds} gives
\[
 \Tr(RW_h)\le2L^2c^{-3/2}\sqrt k.
\]
Since \(W_h\succeq0\), Lemma~\ref{rect:trace-interpolation} implies
\begin{align*}
 \Tr(R_qW_h)
 &\le(2L^2c^{-2}k)^{1-2q}
      (2L^2c^{-3/2}\sqrt k)^{2q}\\
 &=2L^2c^{q-2}k^{1-q}.
\end{align*}
Cancellation of the neighboring Sylvester factors is exact for the
true inverse curvature:
\[
 K_q=\mathcal J_P^{-1/2}\Phi A_{q,\theta}^{-1}
                         \Phi\mathcal J_P^{-1/2}.
\]
Thus, using self-adjointness of \(\Phi\),
\begin{align*}
 \Tr_{\mathrm{HS}}(\Sigma_hK_q)
 &=\tfrac12\sum_a
   \langle\Phi(D_a^h),A_{q,\theta}^{-1}\Phi(D_a^h)\rangle\\
 &\le\frac1{2\theta q}\Tr(R_qW_h).
\end{align*}
The last inequality is \eqref{rect:balanced-inverse-bound}; combining
it with the preceding trace estimate proves the result.
\end{proof}

\subsection{The faithful-source cap-to-response bound}

\begin{theorem}[General Tsallis cap-to-response estimate]
\label{rect:faithful-cap-response}
For the potential
\[
 E_{q,\theta}(H,C)=
 \max_{S\succeq0,\,\Tr S=1}
 \left\{\Tr(HS)+2F(S,\eta_C(S))
                    +\frac{\theta}{1-q}\Tr S^{1-q}\right\},
\]
suppose \(k\ge1\), \(0\preceq C\preceq I_k\), and the actual positive
optimizing density has a faithful physical source. Let \(J\) be one half of the actual optimized center Hessian
and let \(\Gamma\) be the actual covariance derivative. If
\[
 \Gamma|_{\operatorname{ran}C}\preceq\frac L{\sqrt k}I,
\]
then every \(c>0\) gives
\begin{equation}
 \Tr(CJ)\le2\sqrt k+
       \frac{4c^q+2L^2c^{q-2}}{\theta q}k^{1-q}.
 \label{rect:cap-response-parameter}
\end{equation}
In particular, choosing \(c=L\) yields
\begin{equation}
 \boxed{\displaystyle
 \Tr(CJ)\le2\sqrt k+\frac{6L^q}{\theta q}k^{1-q}.}
 \label{rect:cap-response}
\end{equation}
The same bound holds for the local unconstrained response at any
faithful subnormalized density with the same actual contraction
realization and Gram cap.
\end{theorem}

\begin{proof}
The covariance derivative is unchanged in form by the regularizer:
\[
 \Gamma(u,u)=\Tr(S\A(u)Z\A(u)).
\]
After factoring the coefficient covariance,
its balanced force Gram is \(C^{1/2}\Gamma C^{1/2}\). The assumed cap
and \(C\preceq I\) therefore imply
\[
 C^{1/2}\Gamma C^{1/2}\preceq(L/\sqrt k)C\preceq(L/\sqrt k)I.
\]
All preceding force estimates consequently apply.

Lemma~\ref{sq:an:inverse}, applied with the power regularizer, and
$\mathcal A\succeq H_{q,\theta}$ give
\[
 \mathcal A^{-1}\preceq I+K_q,\qquad
 \mathcal A^{-1}\preceq H_{q,\theta}^{-1}.
\]
The actual density response is initially restricted to the trace-zero
tangent. Dropping that constraint only increases the inverse
variational supremum, giving
\[
 \Tr(CJ)\le\sum_a\langle F_a,\mathcal A^{-1}F_a\rangle,
 \qquad F_a=2^{-1/2}\mathcal J_P^{1/2}D_a.
\]
Split \(F_a=F_a^h+F_a^l\), corresponding to the displayed force split.
The squared triangle inequality in the norm defined by the full
\(\mathcal A^{-1}\), followed by the two inverse bounds, yields
\begin{align*}
 \Tr(CJ)
 &\le2\Tr_{\mathrm{HS}}\Sigma_h
       +2\Tr_{\mathrm{HS}}(\Sigma_hK_q)
       +2\sum_aQ_{q,\theta}(D_a^l)\\
 &\le2\sqrt k+
       \frac{2L^2c^{q-2}+4c^q}{\theta q}k^{1-q}.
\end{align*}
Here \(\langle F_a^l,H_{q,\theta}^{-1}F_a^l\rangle
=Q_{q,\theta}(D_a^l)\) by the definitions of the whitened coordinate
and half-response. This proves \eqref{rect:cap-response-parameter};
\(c=L\) proves \eqref{rect:cap-response}.

Every algebraic and trace estimate also holds at a faithful
subnormalized density. Stationarity was used only to identify the
actual trace-constrained optimized Hessian. Without that step the
same argument bounds the local unconstrained response directly.
\end{proof}

\section{Reciprocal powers and singular physical sources}\label{rect:compression-section}

A faithful density need not give an invertible source on the whole
physical space: the input matrices may share a kernel. The source can be
compressed to its support, but the deterministic center can still couple
that support to its complement. Thus replacing the whole optimization by
a smaller one would lose part of the density response.
We instead minimize the full regularizer quadratic form over all
off-support perturbations with a specified compression. The next inverse
curvature identity makes that minimization comparable to the compressed
regularizer. This is the reason for using reciprocal integer powers.

The faithful response estimate applies to all $0<q\le1/2$. We now
restrict to $q=1/p$, $p\ge2$ an integer, to transfer it to every coefficient covariance.

\begin{lemma}[Exact inverse curvature and compression]\label{rect:compression}
Let $q=1/p$, $p\ge2$, and
$\mathcal B_S=-D^2[\theta\Tr S^{1-q}/(1-q)]$ on the full matrix
Hilbert space. Then
\begin{equation}\label{rect:finite-sum}
 \mathcal B_S^{-1}
 =\frac1\theta\sum_{j=1}^p
       L_{S^{j/p}}R_{S^{(p+1-j)/p}}.
\end{equation}
For compression $\pi(X)=V^TXV$ by any isometry, and $S_0=\pi(S)$,
\begin{equation}\label{rect:compression-order}
 \pi\mathcal B_S^{-1}\pi^*\preceq\mathcal B_{S_0}^{-1}.
\end{equation}
\end{lemma}
\begin{proof}
In an eigenbasis of $S$, the inverse multiplier is
$\theta^{-1}(a-b)/(b^{-q}-a^{-q})$.
Writing $a=x^p,b=y^p$ factors this as
$\theta^{-1}\sum_{j=1}^p x^j y^{p+1-j}$, proving
\eqref{rect:finite-sum}, including the continuous value at $a=b$.
The sum preserves real symmetric matrices by pairing indices $j$ and $p+1-j$.

For $0<a<1$, the resolvent representation proved in
Lemma~\ref{rect:power-differentiation} gives
\[
 S^a=c_a
       \int_0^\infty t^{a-1}S(S+tI)^{-1}\,dt.
\]
Here $c_a>0$ is the normalizing scalar defined in that lemma's proof.
Block elimination yields
$\pi((S+tI)^{-1})\succeq(S_0+tI)^{-1}$. Subtracting $t$ times this
inequality from the identity and integrating proves
$\pi(S^a)\preceq S_0^a$; for $a=1$ equality holds.
Each exponent in \eqref{rect:finite-sum} belongs to $[1/p,1]$.
Moreover
$\pi L_{S^a}R_{S^b}\pi^*=L_{\pi(S^a)}R_{\pi(S^b)}$.
If $0\preceq U\preceq U'$ and $0\preceq V\preceq V'$, then
\[
 L_{U'}R_{V'}-L_UR_V
 =L_{U'-U}R_{V'}+L_UR_{V'-V}\succeq0;
\]
positivity follows from
$\langle X,L_UR_VX\rangle=\|U^{1/2}XV^{1/2}\|_{\HS}^2$.
Apply this order comparison termwise to obtain
\eqref{rect:compression-order}. These are inequalities on the full
real matrix Hilbert space and hence also on its real symmetric subspace
when applied to the sums. The argument uses order under compression term by term.
\end{proof}

\begin{proposition}[Response transfer to singular sources]\label{rect:all-covariances}
For $q=1/p$, the faithful cap-to-response estimate proved above remains
valid for every coefficient covariance $0\preceq C\preceq I_k$, including singular physical sources,
with unchanged parameters and constants.
\end{proposition}
\begin{proof}
Let $S$ be the full optimizing density, factor the coefficient covariance into
real symmetric $B'_a$, and set $K=\operatorname{span}_a\ran B'_a$.
If $K=0$, all covariance-weighted forces $B'_a$ vanish, so
$\Tr(CJ)=0$ and the assertion is immediate.
Otherwise let $\Pi_K$ be the orthogonal projection onto $K$, let $\pi$
compress matrices to $K$, and let $\pi^*$ extend them by zero.
Symmetry gives $B'_a=\Pi_KB'_a\Pi_K$.
For $S_0=\pi(S)$ the compressed source $M_0=\eta_C(S_0)$ is faithful:
if $u^TM_0u=0$, then every $B'_au=0$, and a vector in $K$ with this
property is zero. Fidelity compression gives
\[
 F(S,\eta_C(S))=F(S_0,\eta_C(S_0)).
\]
Define $A_{f,0}$ to be the negative density Hessian of
$2F(S_0,\eta_C(S_0))$ on the compressed symmetric matrix space.
The full free density curvature is then $\pi^*A_{f,0}\pi$.
The identity also preserves the actual covariance derivatives in
directions supported on the coefficient covariance's coefficient range.

For a positive definite quadratic form $A$, minimize
$\langle X,AX\rangle$ under $\pi X=Y$. A Lagrange multiplier $\Lambda$
gives $AX=\pi^*\Lambda$. Since $\pi$ is onto,
$\pi A^{-1}\pi^*$ is positive definite, and therefore
\[
 \Lambda=(\pi A^{-1}\pi^*)^{-1}Y,\qquad
 X=A^{-1}\pi^*(\pi A^{-1}\pi^*)^{-1}Y.
\]
Substitution gives
\[
 \inf_{\pi X=Y}\langle X,AX\rangle
 =\langle Y,(\pi A^{-1}\pi^*)^{-1}Y\rangle.
\]
By \eqref{rect:compression-order}, the minimal full regularizer curvature
at fixed compressed perturbation $Y$ is at least
$\langle Y,\mathcal B_{S_0}Y\rangle$. Adding the free curvature and
using the inverse variational formula proves, for every compressed
coefficient force $B$ on $K$,
\[
 \tfrac12 D_H^2E(H,C)[\pi^*B,\pi^*B]\le\tfrac12
  \langle B,(A_{f,0}+\mathcal B_{S_0})^{-1}B\rangle.
\]
The actual density trace restriction can only decrease the left-hand
response. In particular off-support density perturbations have been
accounted for through quadratic minimization, rather than discarded.
Every covariance-weighted force $B'_a$ is supported on $K$; equivalently,
$\A(u)$ is supported there for $u\in\ran C$. The covariance derivative
on this range and its cap agree with the compressed ones. Finally $\Tr S_0\le1$ and
$\eta_C(I_K)\preceq kI_K$, so all budgets in
\eqref{rect:endpoint-budgets} apply on $K$. The faithful estimate was proved for
arbitrary subnormalized faithful densities; the stated subnormalized hypotheses hold at $S_0$. Summing the
local estimates proves the assertion.
\end{proof}

\section{A small square-root term and the aspect-ratio choice}\label{rect:ridge-section}

The power potential has the desired analytic response, but its density
floor contains the exponent $1/q$. When that exponent grows with the
aspect ratio, the floor alone is inadequate for a uniform polynomial
accuracy bound. We therefore add a square-root term whose value is at
most two. Its purpose is numerical regularity; the next proof checks
that it preserves the response estimate at the new optimizing density.

For $\kappa>0$ define
\begin{equation}\label{rect:ridge-potential}
 \begin{aligned}
 E_\kappa(H,C)=\max_{S\succeq0,\ \Tr S=1}
 \biggl\{&\Tr(HS)+2F(S,\eta_C(S))\\
 &+\frac{\theta}{1-q}\Tr S^{1-q}+2\kappa\Tr\sqrt S\biggr\},
 \end{aligned}
\end{equation}
and set $f_\kappa(H)=E_\kappa(H,0)$.
The optimizer is unique and faithful by the same strict-concavity and
boundary-mixing arguments as in Lemma~\ref{rect:basic}.
The potential remains convex and $1$-Lipschitz in $H$, and monotone in $C$.
In particular
\begin{equation}\label{rect:ridge-basic}
 \begin{gathered}
 E_\kappa(H,C)\ge\lambda_{\max}(H),\qquad
 0\le E_\kappa(H,C)-f_\kappa(H)\le2\sqrt k,\\
 E_\kappa(0,I_n)\le2\sqrt n+\frac{\theta}{1-q}D^q+2\kappa\sqrt D.
 \end{gathered}
\end{equation}
The extra square-root term is independent of $C$, so the formula for
$\Gamma$ in \eqref{rect:optimized-response} is unchanged in form. All
its factors are now evaluated at the new optimizer.

\begin{proposition}[The same response bound at the new optimizer]\label{rect:ridge-response}
Assume $q=1/p$ with integer $p\ge2$. Let
$J_\kappa=\frac12D_H^2E_\kappa$ in coefficient coordinates and
$\Gamma_\kappa=D_CE_\kappa$. At every coefficient covariance
$0\preceq C\preceq I_k$, $k\ge1$, including singular ones,
\begin{equation}\label{rect:ridge-cap}
 \Gamma_\kappa|_{\ran C}\preceq \frac L{\sqrt{k}}I
 \quad\Longrightarrow\quad
 \Tr(CJ_\kappa)\le2\sqrt{k}+
 \frac{6L^q}{\theta q}k^{1-q}=:B_k\sqrt{k}.
\end{equation}
\end{proposition}
\begin{proof}
Fix the new optimizing density $S$. At this same density let $A_0$ be
the negative density Hessian of the fidelity plus the original power
regularizer, restricted to the trace-zero tangent. Let
$A_{1/2}=-D_S^2(2\Tr\sqrt S)$ on that tangent. Then
\[
 A_\kappa=A_0+\kappa A_{1/2}\succeq A_0>0.
\]
For any force functional $b$, the inverse variational formula gives
\[
 \langle b,A_\kappa^{-1}b\rangle
 =\sup_X\{2\langle b,X\rangle-\langle X,A_\kappa X\rangle\}
 \le\sup_X\{2\langle b,X\rangle-\langle X,A_0X\rangle\}.
\]
The optimized center Hessian is the left-hand inverse response.
The pointwise response estimate in \eqref{rect:cap-response} bounds the right-hand
response at every faithful subnormalized density, not only at an
optimizer of the old objective. Its covariance Gram is evaluated at
this same $S$ and is precisely $\Gamma_\kappa$, since the new term is
independent of $C$. Summing the force inequalities proves
\eqref{rect:ridge-cap} for faithful physical sources.

For a singular physical source, fix its support $K$. Proposition~\ref{rect:all-covariances} minimizes the full negative density
quadratic form over perturbations with a prescribed compression to $K$.
Adding the nonnegative form $\kappa A_{1/2}$ can only increase this
minimum. The original compressed comparison therefore remains valid;
its reduced density is faithful and has trace at most one. Applying the
same pointwise estimate on $K$ proves the claim. This argument does not
compare Hessians at two different optimizing densities.
\end{proof}

\subsection{Choosing the exponent and balancing the two costs}\label{rect:parameters}

Fix $L=4096$ and $\kappa=D^{-1}$. Choose $p\ge2$ to be a power of two
such that, with $r_{\mathrm{asp}}=D/n$,
\begin{equation}\label{rect:p-choice}
 r_{\mathrm{asp}}^{1/p}\le2,\qquad p\le4+2\log_2 r_{\mathrm{asp}}\le4(1+\log r_{\mathrm{asp}}),
 \qquad q=1/p.
\end{equation}
Here $m\ge n\ge1$ gives $r_{\mathrm{asp}}\ge2$, and $\log$ is the natural logarithm.
A finite way to make this choice starts with $(p,b_{\mathrm{pow}})=(2,4)$ and, while
$D>n b_{\mathrm{pow}}$, replaces $(p,b_{\mathrm{pow}})$ by $(2p,b_{\mathrm{pow}}^2)$. Throughout $b_{\mathrm{pow}}=2^p$.
At termination the first inequality holds. If a doubling occurred, the
previous power failed and $p<2\log_2r_{\mathrm{asp}}$; otherwise $p=2$. This proves
the displayed upper bound. The number of doublings is at most $6D$.
The value $b_{\mathrm{pow}}$ is used only in an exact-arithmetic comparison; no
logarithm or non-dyadic power is required as a primitive.

Choose $\theta>0$ by
\begin{equation}\label{rect:theta-choice}
 \theta^2=\frac{(1-q)L^q n^{1-q}}{qD^q}.
\end{equation}
All powers here are computed using repeated square roots, since
$p$ is a power of two. The same dyadic choice permits the finite SDP
representation derived in Section~\ref{solver:main}. The parameters are chosen from the original
$n,m$ and stay fixed through all epochs and phases.

Define
\[
 C_{\mathrm{init}}=\frac{\theta D^q}{1-q},\qquad
 C_{\mathrm{resp}}=\frac{L^qn^{1-q}}{\theta q},\qquad
 B_n=2+\frac{6L^q}{\theta q}n^{1/2-q}.
\]
The scalar $C_{\mathrm{init}}$ is the initial power-regularization cost. The scalar $C_{\mathrm{resp}}$
measures the accumulated response cost: the common walk bounds its
phase charge by an absolute constant times $(B_n+1)\sqrt n$.
Equation~\eqref{rect:theta-choice} makes $C_{\mathrm{init}}=C_{\mathrm{resp}}$, and hence
\begin{equation}\label{rect:balanced-costs}
 C_{\mathrm{init}}^2=C_{\mathrm{resp}}^2=\frac{L^q n(D/n)^q}{q(1-q)}.
\end{equation}
Since $L^q\le64$, $(D/n)^q\le2$, and
$1/[q(1-q)]\le2p\le8(1+\log(D/n))$, writing
$\mathcal W=\sqrt{n(1+\log(D/n))}$ gives
\begin{equation}\label{rect:cost-bounds}
 C_{\mathrm{init}}=C_{\mathrm{resp}}\le\sqrt{1024}\,\mathcal W=32\mathcal W,
 \qquad (B_n+1)\sqrt n=3\sqrt n+6C_{\mathrm{resp}}\le195\mathcal W.
\end{equation}
Also $2\kappa\sqrt D\le2$, so the initial value in
\eqref{rect:ridge-basic} is at most $36\mathcal W$.
For every epoch size $k\le n$, the response bound
\eqref{rect:ridge-cap} is at most $B_n\sqrt k$, because
$1/2-q\ge0$. These are exactly the two scale estimates required by
the common epoch argument. They yield
$O(\sqrt{n(1+\log(2m/n))})$ discrepancy.
Since $m\ge n$ implies $\log(2m/n)\ge\log2$, this is equivalently
$O(\sqrt{n\log(2m/n)})$.

At $q=1/2$, $\theta=1$, and $\kappa=0$, the same response formula
reduces to $(2+12\sqrt L)\sqrt k=770\sqrt k$, the square estimate.
The rectangular theorem therefore uses the same geometric walk and
probability argument. The work specific to its aspect ratio lies in
the power-curvature comparison, the interpolation, and the uniform
regularity provided by the small square-root term.

\begin{corollary}[The rectangular part of Theorem~\ref{main:combined}]
\label{rect:conclusion}
For $m>n\ge1$, Algorithm~\ref{alg:full} returns a signing of discrepancy
at most $10^8\sqrt{n(1+\log(2m/n))}$, except with probability at most
$2^{-b}$ of returning failure.
\end{corollary}
\begin{proof}
The implemented profile is $E_\kappa$ with the parameters above.
Proposition~\ref{rect:ridge-response} supplies the common response
contract with $B=B_n$, uniformly for all epoch sizes at most $n$.
The remaining potential properties are \eqref{rect:ridge-basic}.
The legal movement uses the same scaled projection covariance as in the
square case. Its trace bound is local to the current epoch's live count
$\ell$, even when all $n$ original labels remain in the coefficient
matrices: covariance annihilates the labels frozen before the epoch,
and its rank is at most $\ell$. Lemma~\ref{walk:live-flat-trace} charges
only newly frozen labels and the radial constraint. With the common
ledger and $\tau=(B+1)^{-1}\le1/3$, inequality~\eqref{walk:flat-positive}
therefore yields $\Tr Q\ge\ell/16$. The exponent choice changes the
response budget and the time scale, while the uniform signed projection
step and its progress calculation stay the same.

The common signing estimate, Proposition~\ref{walk:assembly}, therefore
applies once the numerical contracts are verified in the implementation
part. Since $B_n\ge2$, the scalar recipe gives
$K\le64(B_n+1)+130\le152(B_n+1)$. Using
\eqref{rect:cost-bounds} and the initial bound $36\mathcal W$, we obtain
\[
 E_\kappa(0,I_n)+4(27K+64)\sqrt n
 \le[36+4\cdot27\cdot152\cdot195+256]\mathcal W
 <10^8\mathcal W.
\]
The failure bound is the common retry argument. The case $m=n$ may
use the already proved square profile; its guarantee is stronger than
the displayed rectangular estimate.
\end{proof}

\subsection{Why the extra square root gives polynomial conditioning}

\begin{lemma}[A density floor for the implemented potential]
\label{rect:ridge-floor}
Let $S$ maximize $E_\kappa(H,C)$, with $0\preceq C\preceq I_k$ and
$\kappa>0$. Then
\[
 S\succeq
 \left(\frac{\kappa}{2\|H\|+2\sqrt k+\theta D^q+\kappa\sqrt D}\right)^2 I.
\]
For the choices in Section~\ref{rect:parameters}, put $s=n+D+2$.
If $k\le n$ and $\|H\|\le n+1$, then in particular
\begin{equation}\label{rect:polynomial-floor}
 S\succeq 2^{-10}s^{-6}I.
\end{equation}
\end{lemma}
\begin{proof}
Let $G$ be the density gradient of $2F(S,\eta_C(S))$, with the
transport compressed to its physical source support if necessary.
Monotonicity gives $G\succeq0$, and homogeneity gives
$\Tr(SG)=2F(S,\eta_C(S))\le2\sqrt k$. Stationarity reads
\[
 H+G+\theta S^{-q}+\kappa S^{-1/2}=\lambda I.
\]
Pair with $S$ and use $\Tr S^{1-q}\le D^q$ and
$\Tr\sqrt S\le\sqrt D$ to obtain
$\lambda\le\|H\|+2\sqrt k+\theta D^q+\kappa\sqrt D$.
Dropping the positive terms $G$ and $\theta S^{-q}$ shows
$\kappa S^{-1/2}\preceq(\lambda+\|H\|)I$, which proves the
first inequality by scalar spectral calculus.

The parameter bounds imply $p\le6D$ and
\[
 \theta^2\le64pn\le384s^2,
 \qquad \theta\le20s,\qquad D^q\le\sqrt D,\qquad
 \kappa=D^{-1}\ge s^{-1}.
\]
The denominator in the first bound is at most
$2s+2\sqrt s+20s\sqrt s+1\le32s^{3/2}$.
Thus $S\succeq2^{-10}s^{-5}I$, which implies the stated more
conservative floor. The estimate depends on the density regularizer
and the total source budget.
\end{proof}

\section{Computing the first derivatives from a primal SDP solution}
\label{query:section}

The proofs so far use a covariance derivative to prepare the state and
a supporting density to decide whether a trial is useful. A feasible
approximate primal solution supplies both reports. There are three steps:
strong concavity converts objective error to density error; the transport
formula converts density error to covariance-derivative error; and a
positive shift makes that derivative stable to evaluate. This gives the
reports used by the algorithm directly from its SDP solution.

For the duration of this section, write the density objective in the common
form
\begin{equation}\label{query:objective}
 \Phi_{H,C}(S)=\operatorname{Tr}(HS)+2F(S,\eta_C(S))
                  +2c\operatorname{Tr}\sqrt S+\mathcal R_0(S),
 \qquad
 E(H,C)=\max_{S\succeq0,\ \operatorname{Tr}S=1}\Phi_{H,C}(S).
\end{equation}
Here $F$ is the unsquared fidelity, $c>0$, and $\mathcal R_0$ is concave.
For the square potential, $c=1$ and $\mathcal R_0=0$; for the rectangular
potential, $c=\kappa$ and $\mathcal R_0$ is its generalized Tsallis power
regularizer. All matrices
in this section are real. The physical matrices $\A_1,\ldots,\A_\ell$ are
symmetric contractions of order $D$. The retained coefficient space is a
subspace $U\subseteq\mathbb R^\ell$, with orthogonal projection $P_U$, and
the current covariance satisfies
\begin{equation}\label{query:coefficient-floor}
 aP_U\preceq C\preceq P_U,\qquad 0<a\le1.
\end{equation}
Thus $\ell$ counts the original labels of the epoch, rather than
$\dim U$. We use the notation
\[
 \A(v)=\sum_{i=1}^\ell v_i\A_i,\qquad
 \eta_C(S)=\sum_{i,j=1}^\ell C_{ij}\A_iS\A_j.
\]
The operator norm on $U$ means the norm of the compression to $U$.

\subsection{What a certified primal solution provides}

Let $S_*$ be the unique maximizing density in
\eqref{query:objective}. A \emph{certified primal gap} $\nu$ means that the
solver returns a feasible point of the lifted SDP whose objective value is
at least $E(H,C)-\nu$. Its density component is denoted by $\widehat S$.
Feasibility makes its lifted objective a lower bound for the density
objective at $\widehat S$. That lower bound is what allows the objective
gap to control the distance to the optimizing density.

\begin{lemma}[Objective accuracy gives density accuracy]
\label{query:gap-density}
A certified primal gap $\nu\ge0$ gives
\begin{equation}\label{query:density-error}
 \|\widehat S-S_*\|_{\mathrm{op}}
 \le\|\widehat S-S_*\|_{\mathrm{HS}}
 \le2\sqrt{\nu/c}.
\end{equation}
\end{lemma}
\begin{proof}
At a positive trace-one density $S$, diagonalize $S$ with eigenvalues
$s_1,\ldots,s_D\in(0,1]$. The square-root derivative formula gives, for
every real symmetric $X$,
\[
 -D^2\bigl(2c\operatorname{Tr}\sqrt S\bigr)[X,X]
 =c\sum_{i,j}
 \frac{|X_{ij}|^2}{\sqrt{s_is_j}(\sqrt{s_i}+\sqrt{s_j})}
 \ge\frac c2\|X\|_{\mathrm{HS}}^2.
\]
The remaining nonlinear terms in \eqref{query:objective} are concave,
so $\Phi_{H,C}$ has the same strong concavity bound on the density set.
The optimizer is faithful, and its derivative vanishes on trace-zero
directions. Integrating the second-derivative bound along the segment
from $S_*$ to any positive density $S$ gives
\[
 E(H,C)-\Phi_{H,C}(S)\ge\frac c4\|S-S_*\|_{\mathrm{HS}}^2.
\]
For a singular density, apply this inequality at interior points of the
segment and pass to its endpoint by continuity. Finally, the objective
of a feasible lifted SDP point is no larger than
$\Phi_{H,C}(\widehat S)$: maximizing its auxiliary variables at fixed
$\widehat S$ gives exactly that quantity. The certified gap therefore
implies \eqref{query:density-error}.
\end{proof}

\subsection{Stability of the covariance derivative}

For any faithful density $S$, let $\Gamma(S)$ denote the covariance
derivative of $2F(S,\eta_C(S))$, keeping $S$ fixed and allowing only
perturbations supported on $U$. At $S=S_*$, the envelope identity makes
this the derivative of the optimized potential:
\[
 D_CE(H,C)[V]=\operatorname{Tr}(\Gamma(S_*)V)
 \quad\text{for symmetric }V=P_UVP_U.
\]
This notation separates the analytic derivative at the true optimizer
from the same explicit formula evaluated at $\widehat S$.

\begin{lemma}[Relative density stability]
\label{query:relative-stability}
Suppose $S_*\succeq\mu I_D$ and
$\|\widehat S-S_*\|_{\mathrm{op}}\le e\le\mu/2$, where
$\widehat S$ is a trace-one density. Under
\eqref{query:coefficient-floor},
\begin{equation}\label{query:gamma-density-error}
 \|\Gamma(\widehat S)-\Gamma(S_*)\|_{\mathrm{op},U}
 \le\frac{4e\sqrt\ell}{\mu a}.
\end{equation}
No lower bound on a positive eigenvalue of the physical source is needed.
\end{lemma}
\begin{proof}
Set $r=e/\mu\le1/2$. The density floor gives
\[
 (1-r)S_*\preceq\widehat S\preceq(1+r)S_*.
\]
Positivity of $\eta_C$ transfers this comparison to the two sources.
Their common physical support is
\[
 K=\operatorname{span}\{\operatorname{ran}\A(v):v\in U\}.
\]
Indeed a positive factorization of $C$ writes $\eta_C(I)$ as a sum of
squares of matrices spanning $\{\A(v):v\in U\}$, and every faithful
density gives a source with the same kernel as $\eta_C(I)$.

On $K$, write $S_0$ for the compression of $S_*$ and
$M=\eta_C(S_*)$. Let $Z$ be the unique positive solution of $ZMZ=S_0$;
use hats for the corresponding quantities at $\widehat S$. If $K$ is
zero, every retained force vanishes and there is nothing to prove.
Otherwise the two equivalent transport formulas are
\begin{align*}
 Z(S_0,M)
 &=M^{-1/2}(M^{1/2}S_0M^{1/2})^{1/2}M^{-1/2}\\
 &=S_0^{1/2}(S_0^{1/2}MS_0^{1/2})^{-1/2}S_0^{1/2}.
\end{align*}
Square-root monotonicity in the first formula shows monotonicity in
$S_0$. Inverse order followed by square-root monotonicity in the second
shows antitonicity in $M$. Scalar homogeneity then yields
\[
 \sqrt{\frac{1-r}{1+r}}\,Z
 \preceq\widehat Z\preceq
 \sqrt{\frac{1+r}{1-r}}\,Z.
\]
The derivative formula, with transports extended by zero outside $K$, is
\[
 \Gamma(S_*)[v,v]=\operatorname{Tr}(S_*\A(v)Z\A(v)),\qquad v\in U.
\]
This is a nonnegative quadratic form and is separately monotone in the
positive matrices $S_*$ and $Z$. Hence
\[
 (1-r)\sqrt{\frac{1-r}{1+r}}\,\Gamma(S_*)
 \preceq\Gamma(\widehat S)\preceq
 (1+r)\sqrt{\frac{1+r}{1-r}}\,\Gamma(S_*).
\]
For $0\le r\le1/2$, the two scalar factors lie between $1-4r$ and
$1+4r$. Finally,
\[
 \operatorname{Tr}(C\Gamma(S_*))
 =F(S_*,\eta_C(S_*))\le\sqrt\ell.
\]
The last inequality follows from
$\operatorname{Tr}S_*=1$ and
$\operatorname{Tr}\eta_C(S_*)\le\ell$.
Since $C\succeq aP_U$ and $\Gamma(S_*)\succeq0$ on $U$, its operator
norm on $U$ is at most $\sqrt\ell/a$. Combining these estimates proves
\eqref{query:gamma-density-error}.
\end{proof}

\subsection{A stable formula for the covariance derivative}

The transport formula in the preceding proof is useful analytically,
but direct numerical inversion of its source would introduce an
unnecessary conditioning problem. A small positive shift after a change
of coordinates avoids this problem. The shift is used only to evaluate
the derivative; it does not modify the potential.

\begin{lemma}[A stable direct formula]
\label{query:stable-formula}
For a faithful trace-one density $S$, define
\begin{equation}\label{query:whitened-matrices}
 R=S^{1/2},\qquad L_i=R\A_iR,\qquad
 N=R\eta_C(S)R=\sum_{i,j}C_{ij}L_iL_j.
\end{equation}
For $\sigma>0$, put
\begin{equation}\label{query:shifted-estimator}
 P_\sigma=(N+\sigma I_D)^{-1/2},\qquad
 (\Gamma_\sigma(S))_{ij}=\operatorname{Tr}(L_iP_\sigma L_j).
\end{equation}
The displayed coefficient matrix is real symmetric. On the retained
coefficient space,
\begin{equation}\label{query:shift-error}
 0\preceq\Gamma(S)-\Gamma_\sigma(S),\qquad
 \|\Gamma(S)-\Gamma_\sigma(S)\|_{\mathrm{op},U}
 \le\frac{D\sqrt\sigma}{a}.
\end{equation}
\end{lemma}
\begin{proof}
Factor $C$ into positive rank-one terms. The resulting expression for
$N$ is a sum of squares of symmetric matrices spanning
$\{L(v)=R\A(v)R:v\in U\}$. Consequently every such $L(v)$ vanishes on
$\ker N$. On the fixed support of $N$, differentiation of
$2\operatorname{Tr}\sqrt N$ therefore gives
\[
 \Gamma(S)[v,v]
 =\operatorname{Tr}(L(v)N^{\dagger/2}L(v)).
\]
Here $N^{\dagger/2}$ denotes the inverse square root on
$\operatorname{ran}N$, extended by zero on its kernel. This formula
also proves agreement with the transport expression used above.
On every positive eigenspace of $N$, the difference between the inverse
square root and $P_\sigma$ is nonnegative. The forces vanish on the
kernel, so the resulting coefficient difference is positive semidefinite
on $U$. Its weighted trace is
\begin{align*}
 \operatorname{Tr}\bigl(C(\Gamma(S)-\Gamma_\sigma(S))\bigr)
 &=\operatorname{Tr}\sqrt N
       -\operatorname{Tr}\bigl(N(N+\sigma I_D)^{-1/2}\bigr)\\
 &=\sum_{\lambda>0}
       \left(\sqrt\lambda-\frac{\lambda}{\sqrt{\lambda+\sigma}}\right)
 \le D\sqrt\sigma,
\end{align*}
where $\lambda$ runs over the positive eigenvalues with multiplicities.
Each summand is at most
$\sqrt{\lambda+\sigma}-\sqrt\lambda\le\sqrt\sigma$.
Dividing the weighted trace bound by $a$ bounds the ordinary trace of
the positive coefficient difference, and hence its operator norm.
\end{proof}

\begin{corollary}[An inverse-polynomial accuracy contract]
\label{query:direct-contract}
Let $\eta>0$ be the desired operator error for the covariance derivative,
and suppose the current density floor is $\mu$. Choose
\begin{equation}\label{query:accuracy-choice}
 e=\min\left\{\frac\mu2,
                 \frac{\eta\mu a}{12\sqrt\ell}\right\},\qquad
 \nu=\frac{ce^2}{4},\qquad
 \sigma=\min\left\{1,
                 \left(\frac{\eta a}{3D}\right)^2\right\}.
\end{equation}
Obtain a feasible primal SDP solution with gap at most $\nu$, and evaluate
$\Gamma_\sigma(\widehat S)$ with operator error at most $\eta/3$ on $U$.
The resulting report $\widehat\Gamma$ satisfies
\[
 \|\widehat\Gamma-D_CE(H,C)\|_{\mathrm{op},U}\le\eta.
\]
Every requested reciprocal accuracy is polynomial when
$D,\ell,c^{-1},\mu^{-1},a^{-1},\eta^{-1}$ have polynomial bounds.
\end{corollary}
\begin{proof}
Lemma~\ref{query:gap-density} bounds the density error by $e$.
Lemmas~\ref{query:relative-stability} and \ref{query:stable-formula}
each contribute at most $\eta/3$. The stated numerical evaluation
budget supplies the final third.
\end{proof}

For completeness, the following estimate makes precise why ordinary
approximate spectral computations suffice for this last evaluation.
It also distinguishes approximation of a first derivative from the
finite-movement Taylor remainders needed elsewhere in the walk.

\begin{lemma}[Error in the spectral evaluation]
\label{query:spectral-error}
Assume $0<\sigma\le1$ and $\ell\ge1$. Let $\widetilde R$ be real
symmetric with
$\|\widetilde R-S^{1/2}\|_{\mathrm{op}}\le d\le1$, and form
\[
 \widetilde L_i=\widetilde R\A_i\widetilde R,
 \qquad \widetilde N=\sum_{i,j}C_{ij}\widetilde L_i\widetilde L_j.
\]
Let $\widetilde P$ be real symmetric and satisfy
\[
 \|\widetilde P-(\widetilde N+\sigma I_D)^{-1/2}\|_{\mathrm{op}}
 \le b.
\]
If
$\widetilde\Gamma_{ij}=\operatorname{Tr}(\widetilde L_i
 \widetilde P\widetilde L_j)$, then
\begin{equation}\label{query:evaluation-error}
 \|\widetilde\Gamma-\Gamma_\sigma(S)\|_{\mathrm{op},U}
 \le135D\ell^2\sigma^{-3/2}d+16D\ell b.
\end{equation}
The bound assumes exact scalar operations in forming the displayed
products; it is the spectral approximation error in the stated
real-arithmetic model.
\end{lemma}
\begin{proof}
For unit $v\in U$, the contraction hypothesis gives
$\|\A(v)\|_{\mathrm{op}}\le\sqrt\ell$. Since
$\|S^{1/2}\|_{\mathrm{op}}\le1$ and
$\|\widetilde R\|_{\mathrm{op}}\le2$,
\[
 \|\widetilde L(v)-L(v)\|_{\mathrm{op}}\le3\sqrt\ell d,
 \qquad \|\widetilde L(v)\|_{\mathrm{op}}\le4\sqrt\ell.
\]
Stacking $L_1,\ldots,L_\ell$ as a map from physical space to
$\ell$ copies of it writes
$N=\mathcal L^T(C\otimes I_D)\mathcal L$.
The stack has norm at most $\sqrt\ell$, its approximating stack has
norm at most $4\sqrt\ell$, and their difference has norm at most
$3\sqrt\ell d$. Thus
\[
 \|\widetilde N-N\|_{\mathrm{op}}\le15\ell d.
\]
Both matrices are positive semidefinite, because $C\succeq0$.
For $X,Y\succeq\sigma I_D$, the resolvent identity in
\[
 X^{-1/2}=\frac1\pi\int_0^\infty
            t^{-1/2}(X+tI_D)^{-1}\,dt
\]
gives
$\|X^{-1/2}-Y^{-1/2}\|_{\mathrm{op}}
 \le\|X-Y\|_{\mathrm{op}}/(2\sigma^{3/2})$.
Expanding the trace quadratic form at a unit $v$ now bounds the error
from changing its two outer factors by
$15D\ell\sigma^{-1/2}d$, the error from replacing $N$ by
$\widetilde N$ in its middle factor by
$120D\ell^2\sigma^{-3/2}d$, and the final middle-factor evaluation
error by $16D\ell b$. Since $\sigma\le1$ and $\ell\ge1$, their sum
implies \eqref{query:evaluation-error}.
\end{proof}

For example, the evaluation budget $\eta/3$ in
Corollary~\ref{query:direct-contract} is ensured by
\[
 d\le\min\left\{1,
           \frac{\eta\sigma^{3/2}}{810D\ell^2}\right\},
 \qquad b\le\frac{\eta}{96D\ell}.
\]
These tolerances depend on the chosen positive shift $\sigma$, not on
the smallest positive eigenvalue of the unshifted physical source.
Exact EVD would simply set $d=b=0$. With approximate EVD one needs
certified matrix-function accuracy, rather than accurate individual
eigenvectors or a lower bound on an eigenvalue gap.

\subsection{The supporting plane for epoch acceptance}

Set $f(H)=E(H,0)$, and let $S_f$ be its unique maximizing density.
Its supporting-plane inequality follows directly by retaining this
density at a new center:
\begin{equation}\label{query:supporting-plane}
 f(H+\Delta H)\ge f(H)+\operatorname{Tr}(S_f\Delta H).
\end{equation}
Thus no differentiation routine is needed to obtain this plane.

\begin{lemma}[A direct supporting-plane report]
\label{query:plane-report}
If a feasible source-free primal SDP solution has gap at most $\nu$
and density component $\widehat S_f$, then the scalar report
$\widehat g=\operatorname{Tr}(\widehat S_f\Delta H)$ satisfies
\begin{equation}\label{query:plane-error}
 \left|\widehat g-\operatorname{Tr}(S_f\Delta H)\right|
 \le2\sqrt{\nu/c}\,\|\Delta H\|_{\mathrm{HS}}.
\end{equation}
For $\Delta H=\A(z)$, one may replace the last norm by
$\sqrt{D\ell}\,\|z\|_2$.
\end{lemma}
\begin{proof}
Hilbert--Schmidt Cauchy--Schwarz and
Lemma~\ref{query:gap-density} give \eqref{query:plane-error}.
For the last assertion,
$\|\A(z)\|_{\mathrm{HS}}\le\sqrt D\|\A(z)\|_{\mathrm{op}}
 \le\sqrt D\sum_i|z_i|
 \le\sqrt{D\ell}\|z\|_2$.
\end{proof}

The cap-test slack and the epoch acceptance margins can therefore be
allocated directly to the certified optimizer error. This replaces the
first-derivative finite differences in both uses. The step size and the
fourth-derivative bound for a finite movement remain separate parts of
the proof.

\section{Computing a feasible optimizing density}
\label{solver:main}

We construct the feasible primal solution required by the preceding
section. The main issue is the geometry of the feasible set: a direct
fidelity block may be singular for every density, even though its value
is well behaved. Factoring the source moves its coefficient matrices
into the objective and yields a larger lift with an explicit interior
ball. That formulation can be solved by a standard ellipsoid algorithm
using the permitted exact eigendecompositions for separation.

The source may therefore have arbitrarily small positive eigenvalues;
the interior radius of the new lift will depend only on its dimensions.

\begin{definition}[Arithmetic model]\label{model:main}
We count scalar real arithmetic, comparisons, nonnegative square roots,
exact symmetric eigendecompositions, and uniform draws in $[0,1)$.
A fixed eigendecomposition routine returns an ordered orthonormal
eigenbasis and its eigenvalues; ties are resolved by the returned order.
Each spectral call counts as one operation, while forming its input and
processing its output are counted separately. The density optimization
is implemented by the semidefinite program and ellipsoid algorithm in
this section.
\end{definition}
The runtime bounds count these real-arithmetic operations;
they do not assert a binary encoding bound.

\subsection{An equivalent SDP with a uniform interior point}

Fix real symmetric contractions $\A_1,\ldots,\A_\ell$ of order $D$, a
real symmetric center $H$ of the same order, and a real covariance
$0\preceq C\preceq I_\ell$.  Here $\ell$ is the number of original labels
in the query; zero columns will be retained in the construction below.
Let $F_{\mathrm{coef}}=C^{1/2}$, computed by exact eigendecomposition, and define
\begin{equation}\label{solver:kraus}
 T_a=\sum_{i=1}^{\ell}(F_{\mathrm{coef}})_{ia}\A_i\quad(1\le a\le\ell),
 \qquad \mathsf T=\begin{pmatrix}T_1&\cdots&T_\ell\end{pmatrix}.
\end{equation}
Thus $\mathsf T$ is a $D$ by $\ell D$ real matrix, and
\begin{equation}\label{solver:sourcefactor}
 \eta_C(S)=\mathsf T(I_\ell\otimes S)\mathsf T^{\mathsf T},
 \qquad
 \mathsf T\mathsf T^{\mathsf T}=\eta_C(I)\preceq\ell I_D.
\end{equation}
For real matrices of the same size, the Hilbert--Schmidt inner product
is $\langle U,V\rangle_{\mathrm{HS}}=\operatorname{Tr}(U^{\mathsf T}V)$.

\begin{lemma}[A uniformly feasible fidelity lift]\label{solver:fidelity}
For every real positive semidefinite $S$,
\begin{equation}\label{solver:fidelitylift}
 F(S,\eta_C(S))=
 \max_{X\in\mathbb R^{D\times\ell D}}
 \left\{\langle\mathsf T,X\rangle_{\mathrm{HS}}:
 \begin{pmatrix}S&X\\X^{\mathsf T}&I_\ell\otimes S\end{pmatrix}
 \succeq0\right\}.
\end{equation}
The maximum is attained, including when $S$ or $C$ is singular.
\end{lemma}
\begin{proof}
For positive semidefinite matrices $P,Q$, positivity of
$\left(\begin{smallmatrix}P&X\\X^{\mathsf T}&Q\end{smallmatrix}\right)$
is equivalent to a factorization $X=P^{1/2}KQ^{1/2}$ with
$\|K\|_{\mathrm{op}}\le1$.  For invertible $P,Q$, this follows by
congruence and the Schur complement.  In the singular case, positivity
first forces the appropriate rows and columns of $X$ to vanish on the
kernels; the same argument on the supports proves the assertion, with
$K$ extended by zero.

Apply the factorization with $P=S$ and $Q=I_\ell\otimes S$.  Set
$N=S^{1/2}\mathsf T(I_\ell\otimes S^{1/2})$.  The maximum in
\eqref{solver:fidelitylift} becomes
$\max_{\|K\|_{\mathrm{op}}\le1}\langle N,K\rangle_{\mathrm{HS}}
=\operatorname{Tr}\sqrt{NN^{\mathsf T}}$.
To check this last identity, take a singular-value decomposition of $N$.
Each diagonal entry of a contraction in the singular-vector bases has
absolute value at most one, giving the upper bound by the sum of the
singular values; the partial isometry matching the two singular-vector
bases attains it.  Finally,
$NN^{\mathsf T}=S^{1/2}\eta_C(S)S^{1/2}$, which proves the lemma.
\end{proof}

We give one format for both regularizers.  Let $a\ge1$ be an integer
and let $b_1,\ldots,b_a\ge0$ be fixed scalar coefficients.  Define
\begin{equation}\label{solver:generalpotential}
 \mathcal E(H,C)=\max_{S\succeq0,\ \operatorname{Tr}S=1}
 \left\{\operatorname{Tr}(HS)+2F(S,\eta_C(S))
       +\sum_{j=1}^{a}b_j\operatorname{Tr}S^{1-2^{-j}}\right\}.
\end{equation}
For the square potential, take $a=1$ and $b_1=2\theta$.
For the rectangular potential, take $p=2^a$, $q=1/p$, add
$\theta/(1-q)$ to $b_a$, add $2\kappa$ to $b_1$, and set the other
coefficients to zero.  If $a=1$, the two contributions are added to
the same coefficient.

\begin{proposition}[A real SDP for both potentials]\label{solver:program}
The value in \eqref{solver:generalpotential} is the attained maximum of
\begin{equation}\label{solver:objective}
 \operatorname{Tr}(HS)+2\langle\mathsf T,X\rangle_{\mathrm{HS}}
                +\sum_{j=1}^{a}b_j\operatorname{Tr}Y_j
\end{equation}
over a real symmetric density $S$, real symmetric matrices
$Y_1,\ldots,Y_a\succeq0$, and a real matrix
$X\in\mathbb R^{D\times\ell D}$, subject to
\begin{equation}\label{solver:constraints}
 \begin{pmatrix}S&X\\X^{\mathsf T}&I_\ell\otimes S\end{pmatrix}\succeq0,
 \qquad
 \begin{pmatrix}S&Y_j\\Y_j&Y_{j-1}\end{pmatrix}\succeq0
 \quad(1\le j\le a),\qquad Y_0=I_D.
\end{equation}
The number of scalar variables after imposing $\operatorname{Tr}S=1$
is
\begin{equation}\label{solver:dimension}
 d=\frac{(a+1)D(D+1)}2-1+\ell D^2.
\end{equation}
There is one block of order $(\ell+1)D$, $a$ blocks of order $2D$,
and the individual positivity constraints.
\end{proposition}
\begin{proof}
Lemma~\ref{solver:fidelity} accounts for $X$.  For the power variables,
define the matrix geometric mean, for $S\succ0$ and $B\succeq0$, by
$S\mathbin{\#}B=S^{1/2}(S^{-1/2}BS^{-1/2})^{1/2}S^{1/2}$.
The positive-block constraint gives $Y_j\preceq S\mathbin{\#}Y_{j-1}$.
Monotonicity of the second argument and induction therefore give
$Y_j\preceq S^{1-2^{-j}}$.
For completeness, when $S$ is invertible the Schur complement gives
\[
 (S^{-1/2}Y_jS^{-1/2})^2
 \preceq S^{-1/2}Y_{j-1}S^{-1/2};
\]
apply square-root monotonicity and conjugate by $S^{1/2}$.
The square-root formula for the geometric mean proves its monotonicity
in the second argument.  If $S$ is singular, positivity forces every
$Y_j$ to be supported on $\operatorname{ran}S$; the same proof on that
support applies.  The commuting choices $Y_j=S^{1-2^{-j}}$ attain all
these inequalities simultaneously.  Since each $b_j$ is nonnegative,
the optimum over the power variables is exactly the required trace
regularizer.  The density set is compact, and its continuous objective
has a maximizer.  The preceding formulas provide attaining auxiliary
variables.  Counting the independent real symmetric coordinates and
the entries of $X$ proves \eqref{solver:dimension}.
\end{proof}

\begin{lemma}[Explicit geometry of the feasible set]\label{solver:geometry}
In the product Hilbert--Schmidt metric on the trace-one affine space,
the feasible set of Proposition~\ref{solver:program} is a compact convex
body.  It contains the ball of radius
\begin{equation}\label{solver:radii}
 r_{\mathrm{in}}=\frac1{64D}
\end{equation}
about the explicitly known point
\begin{equation}\label{solver:interiorpoint}
 S_0=I_D/D,\qquad X_0=0,\qquad Y_{j,0}=I_D/(4D)
 \quad(1\le j\le a).
\end{equation}
It is contained in the ball of radius
$R_{\mathrm{out}}=2\sqrt{(a+2)D}$ about that point.  The absolute value
of its linear objective is at most
\begin{equation}\label{solver:objectivebound}
 B_{\mathrm{obj}}=1+\|H\|_{\mathrm{op}}+2\sqrt\ell
                  +D\sum_{j=1}^{a}b_j.
\end{equation}
These bounds hold uniformly at singular coefficient and physical
covariances.
\end{lemma}
\begin{proof}
Use an orthonormal Hilbert--Schmidt basis for the trace-zero perturbations
of $S$, orthonormal symmetric coordinates for each $Y_j$, and the usual
entry coordinates of $X$.  Thus the stated metric is exactly Euclidean
in dimension $d$, without a coordinate-chart distortion.

At \eqref{solver:interiorpoint}, the fidelity block is
$I_{(\ell+1)D}/D$.  Each power block dominates
\[
 \frac1{4D}\begin{pmatrix}4I_D&I_D\\I_D&I_D\end{pmatrix},
\]
whose smallest eigenvalue is greater than $1/(8D)$.  For the first
power block, its second diagonal block is $I_D$, which is even larger.
The individual $Y_j$ positivity margins equal $1/(4D)$, and the density
margin equals $1/D$.

A perturbation of product Hilbert--Schmidt norm at most $r_{\mathrm{in}}$
changes a power block by operator norm at most $4r_{\mathrm{in}}$.
It changes the fidelity block by at most $2r_{\mathrm{in}}$, because
$\|I_\ell\otimes\Delta S\|_{\mathrm{op}}=
\|\Delta S\|_{\mathrm{op}}$ and the off-diagonal perturbation has norm
$\|\Delta X\|_{\mathrm{op}}$.
Each individual positivity block changes by at most
$r_{\mathrm{in}}$.  All blocks therefore remain positive definite,
proving the inner-ball claim.

For any feasible point, $0\preceq S\preceq I_D$ and
$0\preceq Y_j\preceq S^{1-2^{-j}}\preceq I_D$.
The factorization in Lemma~\ref{solver:fidelity} then gives
$\|X\|_{\mathrm{op}}\le1$, and hence
$\|X\|_{\mathrm{HS}}\le\sqrt D$, since $X$ has $D$ rows.
Also $\|S-I_D/D\|_{\mathrm{HS}}\le1$ and
$\|Y_j-I_D/(4D)\|_{\mathrm{HS}}\le\sqrt D$.
The squared distance to the stated center is at most $1+(a+1)D$,
which is less than $R_{\mathrm{out}}^2$.
The feasible set is closed and convex, so these bounds prove compactness
and the convex-body assertion.

Finally $|\operatorname{Tr}(HS)|\le\|H\|_{\mathrm{op}}$.
The block factorization gives
$|\langle\mathsf T,X\rangle_{\mathrm{HS}}|
\le F(S,\eta_C(S))\le\sqrt\ell$; the last inequality follows from
$\operatorname{Tr}S=1$ and
$\operatorname{Tr}\eta_C(S)=\operatorname{Tr}(S\eta_C(I))\le\ell$.
Each $\operatorname{Tr}Y_j$ is at most $D$, proving
\eqref{solver:objectivebound}.
\end{proof}

\subsection{A standard ellipsoid theorem and the returned certificate}

We state the version of the ellipsoid theorem that we use, including
its output requirement.  An exact separation procedure either certifies
membership in a convex set or supplies a separating affine half-space.

\begin{theorem}[Ellipsoid optimization; Bubeck, Theorem 2.4]
\label{solver:ellipsoid}
Let $\mathcal K\subset\mathbb R^d$ be a compact convex body containing
a ball of radius $r>0$ and contained in a known ball of radius $R$.
Let $f:\mathcal K\to[-B,B]$ be continuous and convex.  Given exact
separation for $\mathcal K$ and exact values and subgradients of $f$,
the ellipsoid algorithm returns a feasible visited point $z_t$ with
\[
 f(z_t)-\min_{z\in\mathcal K}f(z)
 \le\frac{2BR}{r}\exp\left(-\frac{t}{2d^2}\right)
\]
whenever $t\ge2d^2\log(R/r)$.
\end{theorem}
The reference is \cite[Theorem~2.4, p.~250]{bubeck2015}; the convex-body
hypotheses are stated on p.~244 and the algorithm and feasible-point
output rule on pp.~249--250.  This theorem concerns an exact-oracle
arithmetic model.  We verify its oracle requirements directly rather
than infer them from a generic assertion about semidefinite programs.

\begin{proposition}[A feasible primal solution of certified accuracy]
\label{solver:accuracy}
For every $\nu>0$, Proposition~\ref{solver:program} admits an algorithm
returning feasible matrices $(\widehat S,\widehat X,\widehat Y_1,
\ldots,\widehat Y_a)$ and their linear objective value $v$ such that
\begin{equation}\label{solver:valuecertificate}
 v\le\mathcal E(H,C)\le v+\nu.
\end{equation}
Its number of scalar operations and exact real symmetric EVD calls is
polynomial in $\ell,D,a$ and
$\log(2+B_{\mathrm{obj}}+\nu^{-1})$.
All its input coefficients can be formed in polynomially many scalar
operations and EVD calls from $H,C,\A_1,\ldots,\A_\ell$ and $b_1,\ldots,b_a$.
\end{proposition}
\begin{proof}
Factor $C$ and form \eqref{solver:kraus}; its entries and those of every
constraint are explicit real scalar expressions.  At a candidate in
the trace-one coordinate chart, evaluate every affine PSD block.
Exact EVD either verifies all blocks are positive semidefinite or gives
a unit vector $u$ with $u^{\mathsf T}L(z)u<0$ for a violated affine
pencil $L(z)$.  The inequality $u^{\mathsf T}L(z')u\ge0$ is an explicit
separating half-space containing every feasible point.  Thus each
separation call uses polynomially many of the permitted operations.
The objective is linear, so its values and gradient are directly
available. If its gradient in the affine coordinate chart is zero,
the objective is constant and the explicitly known feasible point
\eqref{solver:interiorpoint} is already optimal; return that point.

Apply Theorem~\ref{solver:ellipsoid} to the negative of
\eqref{solver:objective}, using Lemma~\ref{solver:geometry}.  Taking
\[
 t\ge 2d^2\log\left(
 \frac{2B_{\mathrm{obj}}R_{\mathrm{out}}}
      {r_{\mathrm{in}}\min\{1,\nu\}}\right)+1
\]
suffices for the stated objective gap and for the feasible-output
threshold. The usual ellipsoid update has polynomial scalar cost
per iteration. To choose an integer count without a logarithm primitive,
write $A=2B_{\mathrm{obj}}R_{\mathrm{out}}/
(r_{\mathrm{in}}\min\{1,\nu\})$ and double $1$ until reaching
$2^j\ge A$. Then $t=2d^2j+1$ suffices, since
$\log A\le j\log2\le j$. This scan uses $O(1+\log A)$
comparisons and scalar operations. The
one-dimensional case uses interval bisection and a zero-dimensional
feasible set is immediate.  This proves the complexity claim.
\end{proof}

The certificate in \eqref{solver:valuecertificate} consists of a
feasible primal objective and the proved approximation bound for the
specified iteration count.  We do not require a separate SDP dual
matrix from the solver.  In particular, the assertion is stronger than
merely reporting small numerical feasibility residuals: the returned
point is feasible in the stipulated exact arithmetic model.

\begin{corollary}[From objective accuracy to density accuracy]
\label{solver:densityaccuracy}
Write $\varphi(S)$ for the density objective in
\eqref{solver:generalpotential}.  Suppose it is strongly concave with
modulus $c/2$ in the Hilbert--Schmidt metric, and its maximizing density
$S_*$ is faithful.  The output of Proposition~\ref{solver:accuracy}
satisfies
\begin{equation}\label{solver:densityerror}
 \|\widehat S-S_*\|_{\mathrm{op}}
 \le\|\widehat S-S_*\|_{\mathrm{HS}}
 \le2\sqrt{\nu/c}.
\end{equation}
Consequently, if $S_*\succeq\mu I_D$, choosing
$\nu\le ce^2/4$ with $e\le\mu/2$ returns a faithful density within
operator-norm error $e$.
\end{corollary}
\begin{proof}
Trace-constrained stationarity at $S_*$ and strong concavity give
$\varphi(S_*)-\varphi(S)\ge(c/4)\|S-S_*\|_{\mathrm{HS}}^2$.
This extends to boundary densities by continuity.  The linear objective
of any feasible lift of $\widehat S$ is at most $\varphi(\widehat S)$.
Combine this observation with \eqref{solver:valuecertificate}, and then
use $\|\cdot\|_{\mathrm{op}}\le\|\cdot\|_{\mathrm{HS}}$.
\end{proof}

For the square normalization one takes $c=\theta=1$; for the
rectangular implementation the square-root ridge supplies $c=\kappa$.
The density floors and strong-concavity estimates proved for the two
potentials therefore turn inverse-polynomial objective accuracy into
the density accuracy needed by the direct covariance-derivative
routine.  Finite differences in $C$ are unnecessary.  The genuine
Taylor remainder of a finite movement remains part of the walk analysis.

The SDP geometry and the separation algorithm verify the hypotheses of
the stated ellipsoid theorem. In particular, its dependence on objective
accuracy is logarithmic in this arithmetic model. The final section
combines this solver with the finite walk and bounds the number of calls.

\section{Implementing preparation, sampling, and acceptance}
\label{impl:main}

We can now discharge the numerical contracts of
Section~\ref{intro:algorithm}. The solver gives a feasible density and an
objective gap; the direct derivative formula gives the preparation test;
and the covariance and fourth-derivative bounds make finite cuts and
movements affordable. The remaining task is to count these operations
on every branch of the algorithm.

In the scalar recipe, use the explicit constants
\begin{equation}\label{impl:numerical-constants}
 s=n+D+2,\qquad
 M_C=2^{300}s^{60},\qquad M_4=2^{1040}s^{200}.
\end{equation}
Appendix~\ref{reg:section} proves that they bound the covariance second
derivative and the fourth derivative along each movement curve. Their
large exponents only serve to provide one uniform polynomial choice.
All other parameters are those in Section~\ref{intro:algorithm}.
Throughout an epoch, $\ell$ is its initial live count and the center
includes every original matrix, including frozen labels.

\subsection{The numerical recipe and complete algorithm}\label{intro:algorithm}
The complete implementation uses the analytical parameters in
Section~\ref{walk:contracts} and the derivative bounds
\eqref{impl:numerical-constants}. We collect the choices here in the
order in which they are computed. Rounding chooses the nearer sign;
a tie is resolved as $+1$.

\paragraph{Parameters.}
The following parameters are computed for $n,m\ge1$, after the empty-input test. Put $s=n+D+2$,
$L=4096$, $\delta=2^{-13}$, $a_0=(16384n)^{-1}$, and
\[
 \mu=(1024s^6)^{-1},\qquad \epsilon_A=10^{-4}.
\]
For $m\le n$ take $(q,\theta,\kappa,\chi,B)=(1/2,1,0,1,770)$.
For $m>n$, start $p=2$ and double it until $2^p\ge D/n$, and set
\[
 q=1/p,\quad \theta=\sqrt{\frac{(1-q)L^q n^{1-q}}{qD^q}},\quad
 \kappa=\chi=1/D,\quad B=2+\frac{6L^q n^{1/2-q}}{\theta q}.
\]
Since $p$ is a power of two, these powers are evaluated by repeated
square roots. The regularizer contains a term at least $2\chi\Tr\sqrt S$;
$\chi$ controls optimizer accuracy through this term's curvature.
Set
\begin{gather*}
 \tau=(B+1)^{-1},\quad
 h=\min\{a_0/(2\sqrt n),1/2,\sqrt{\tau/2},\sqrt{2^{-40}/M_4}\},\\
 K=\lceil64(B+1)+129\rceil,\quad M=K(n+1),\quad r=M+b+1,\\
 e_A=\epsilon_A/(3s^2),\qquad
 \nu_A=\min\{\epsilon_A/3,\chi e_A^2/4\}.
\end{gather*}
At an epoch with $\ell$ labels, also set
\begin{gather*}
 t_\ell=L/\sqrt\ell,\quad \eta_\ell=t_\ell/64,\quad
 \alpha_\ell=\min\{\delta/2,t_\ell/(2M_C)\},\\
 e_\ell=\min\{\mu/2,\eta_\ell\mu\delta/(12\sqrt\ell)\},\quad
 \nu_\ell=\chi e_\ell^2/4,\quad
 \omega_\ell=(\eta_\ell\delta/(3D))^2.
\end{gather*}
The derivative bounds $M_C,M_4$ and the density floor $\mu$ are proved in Appendix~\ref{reg:section};
Section~\ref{impl:runtime} proves every loop count and inverse tolerance
polynomial. Here $\chi$ is a lower bound on the square-root coefficient and
$\nu_\ell,\nu_A$ are primal objective-gap targets. Their formulas
ensure the derivative and supporting-plane accuracies proved in
Section~\ref{query:section}. The algorithms use the following routines.

\begin{itemize}[nosep,leftmargin=*]
\item $\proc{Solve}(H,C,\nu)$ returns a feasible primal SDP solution,
its density $\widehat S$, and its value $\widehat E$, with objective gap
at most $\nu$. Section~\ref{solver:main} specifies the ellipsoid implementation.
\item $\proc{Prepare}(H,C)$ discards covariance eigenvalues below $2\delta$,
recording their trace as dust; return if the range is empty.
From $\proc{Solve}(H,C,\nu_\ell)$ form
$R=\sqrt{\widehat S}$, $L_i=R\A_iR$, $N=\sum C_{ij}L_iL_j$, and
$\widehat\Gamma_{ij}=\Tr(L_i(N+\omega_\ell I)^{-1/2}L_j)$.
Restrict to $\ran C$. If its largest eigenvalue is
at most $15t_\ell/16$, return. Otherwise take a unit top eigenvector $v$,
replace $C$ by $C-\alpha_\ell vv^T$, record $\alpha_\ell$ as paid loss,
and repeat. Return the final $C$ and both accumulated trace losses.
Only $C$ changes. All spectral operations use the fixed EVD
primitive, with a fixed returned ordering and least-index tie breaking.
\item $\proc{Project}(C,V)$ first computes
$P=\mathbf 1_{[1/2,1]}(C)$ by EVD. For equations
$V=\cap_j v_j^\perp$, replace $P$ in their stated order by
$P-(Pv_j)(Pv_j)^T/(v_j^TPv_j)$ whenever the denominator is positive;
a zero denominator leaves $P$ unchanged. Use frozen-coordinate equations
in increasing original-label order, followed by the radial equation.
Return $Q=P/2$. Lemma~\ref{sq:lem:scalarshort} proves that this is
half the projection onto the legal part of the high spectral subspace.
\item $\proc{Sample}(Q)$ uses an EVD of the projection $2Q$ and keeps
its $k$ eigenvectors $u_j$ with eigenvalue one. Choose $(j,\pm)$ uniformly
from its $2k$ possibilities and return $\pm\sqrt{k/2}\,u_j$.
One uniform-real draw and a scan of equal intervals of length $1/(2k)$
implement this distribution exactly.
\end{itemize}

In Algorithm~\ref{alg:epoch}, the sampled $\xi$ is multiplied by $h$.
Equivalently, the update takes a step of length
$s_{\rm step}=h\sqrt{k/2}$ along the selected unit vector, as in
\eqref{eq:unit-direction-update}. Since $k\le n$, the parameter choice gives
\[
 \norm{h\xi}_\infty\le\norm{h\xi}_2
   =s_{\rm step}\le\frac{a_0}{2\sqrt2}<a_0.
\]
Every live coordinate has distance greater than $a_0$ from its nearer
face before movement, so the step stays inside the cube. The time
counter advances by $h^2=2s_{\rm step}^2/k$ at each movement.

\begin{algorithm}[H]
\caption{One epoch trial from a saved state $x_*$}\label{alg:epoch}
\begin{algorithmic}[1]
\Require Profile and scalars above; $\ell\ge32$ live labels $I_0$ of $x_*$;
every live coordinate has $1-|x_{*,i}|>a_0$;
$(\widehat S_*,\widehat f_*)$ from $\proc{Solve}(H(x_*),0,\nu_A)$
\State $x\gets x_*$, $C\gets I_\ell$, $T\gets0$, $Y\gets0$, $d\gets0$
\While{$T+h^2\le\tau$ and fewer than $\ell/64$ new labels have frozen and $d\le\ell/64$}
 \State $C\gets\proc{Prepare}(H(x),C)$; add its dust and paid trace losses to $d$
 \If{$d>\ell/64$} \State \textbf{break} \EndIf
 \State $V\gets\{z:z_i=0\text{ on newly frozen labels},\ z^Tx_{I_0}=0\}$
 \State $Q\gets\proc{Project}(C,V)$, $\xi\gets\proc{Sample}(Q)$
 \State $x_{I_0}\gets x_{I_0}+h\xi$, $Y\gets Y+h\A(\xi)$, $C\gets C-h^2Q$, $T\gets T+h^2$
 \State Round each live $x_i$ with $1-|x_i|\le a_0$ to its nearest sign and freeze it
\EndWhile
\State $(\widehat S,\widehat E)\gets\proc{Solve}(H(x),C,\epsilon_A/3)$
\State $\widehat\Psi\gets\widehat E-\widehat f_*-\Tr(\widehat S_*(H(x)-H(x_*)))$,
$\widehat M\gets\Tr(\widehat S_*Y)$
\If{$d\le\ell/64$, $\widehat\Psi\le16.5\sqrt\ell$, and $\widehat M\le4.5\sqrt\ell$}
 \State \Return $x$
\EndIf
\State \Return failure
\end{algorithmic}
\end{algorithm}

\begin{algorithm}[H]
\caption{Full signing: square profile first, rectangular profile when $m>n$}\label{alg:full}
\begin{algorithmic}[1]
\Require Real symmetric contractions $A_i\in\R^{m\times m}$; confidence integer $b\ge1$
\If{$n=0$ or $m=0$} \State \Return the all-positive signing \EndIf
\State Form the signed block matrices, compute the scalar recipe, and set $x\gets0$
\For{at most $n+1$ phases}
 \State Let $k$ be the current live count
 \For{at most $K$ saved epoch requests, while the live count exceeds $k/2$}
  \If{fewer than $32$ labels remain live} \State Round them to nearest signs; \Return $x$ \EndIf
  \State Save $x_*=x$; compute $(\widehat S_*,\widehat f_*)=\proc{Solve}(H(x_*),0,\nu_A)$
  \State Run at most $r$ independent trials of Algorithm~\ref{alg:epoch}; install its first accepted $x$
  \If{all trials fail} \State \Return failure \EndIf
 \EndFor
 \If{all coordinates are signs} \State \Return $x$ \EndIf
\EndFor
\State \Return failure
\end{algorithmic}
\end{algorithm}

\subsection{A covariance cap from one approximate density}

\begin{lemma}[The derivative report in \proc{Prepare}]
\label{impl:gamma-report}
Suppose $C$ has positive eigenvalues at least $2\delta$, and let
$U=\operatorname{ran}C$. If $U\ne\{0\}$, the report in
\proc{Prepare} satisfies
\[
 \|\widehat\Gamma-\Gamma\|_{\mathrm{op},U}\le\eta_\ell=t_\ell/64,
 \qquad \Gamma=D_CE(H,C).
\]
\end{lemma}
\begin{proof}
The uniform density floor in Appendix~\ref{reg:section} is $\mu$, and the
coefficient of a square-root regularizer is at least $\chi$ in either
profile. Proposition~\ref{solver:accuracy} supplies a feasible primal
solution with gap $\nu_\ell=\chi e_\ell^2/4$.
Lemma~\ref{query:gap-density} therefore gives density error at most
$e_\ell$. Apply Corollary~\ref{query:direct-contract} with
$a=\delta$ and target error $\eta_\ell$.
Its shift is precisely $\omega_\ell$; in fact the stated constants give
\[
 e_\ell=\frac{\mu}{1536\ell},\qquad
 \omega_\ell=\frac1{147456D^2\ell}<1.
\]
The formulas in \proc{Prepare} use exact EVD and exact scalar operations,
so their additional matrix-function evaluation error is zero. The
density and shift errors are each at most $\eta_\ell/3$, which is
stronger than the asserted total bound. The report is restricted to
$U$ before its largest eigenvalue is taken.
\end{proof}

\begin{lemma}[Paid covariance loss and termination]
\label{impl:preparation}
The routine \proc{Prepare} makes only positive-semidefinite decreases
of $C$ and leaves $H$ fixed. It returns either zero covariance or a
covariance with positive eigenvalues at least $2\delta$ and
\begin{equation}\label{impl:prepared-cap}
 \Gamma|_{\operatorname{ran}C}\preceq t_\ell I.
\end{equation}
Each paid decrement $C\mapsto C-\alpha_\ell vv^T$ decreases the potential
by at least $t_\ell\alpha_\ell/2$. Across an entire epoch trial there
are at most $\lfloor\ell/\alpha_\ell\rfloor$ paid decrements. The
total dust loss is at most
\begin{equation}\label{impl:dust}
 2\delta\ell=\ell/4096\le\ell/1024.
\end{equation}
In particular, every call terminates.
\end{lemma}
\begin{proof}
Removing eigenvalues below $2\delta$ is an order decrease. Its positive
survivors satisfy the stated floor. A cut has size
$\alpha_\ell\le\delta/2$, so its entire segment remains at least
$3\delta/2$ on the current range; its range therefore does not change.
All density and covariance derivative estimates apply along that segment.

At a stopped nonzero report, Lemma~\ref{impl:gamma-report} gives
\[
 \lambda_{\max}(\Gamma|_U)
 \le15t_\ell/16+t_\ell/64=61t_\ell/64<t_\ell.
\]
At a continuing report, an exact unit top eigenvector $v$ instead gives
\[
 v^T\Gamma v>15t_\ell/16-t_\ell/64=59t_\ell/64.
\]
The covariance second-derivative bound $M_C$ from
Appendix~\ref{reg:section}, together with
$\alpha_\ell\le t_\ell/(2M_C)$, yields
\begin{align*}
 E(H,C)-E(H,C-\alpha_\ell vv^T)
 &\ge\frac{59}{64}t_\ell\alpha_\ell
          -\frac{M_C\alpha_\ell^2}{2}\\
 &\ge\frac{43}{64}t_\ell\alpha_\ell
 \ge\frac12t_\ell\alpha_\ell.
\end{align*}
Dust removal cannot increase the potential, by its monotonicity in $C$.

An epoch starts from $I_\ell$. Both preparation and movement decrease
$C$ in positive-semidefinite order. Each paid cut consumes trace exactly
$\alpha_\ell$, so there can be at most $\ell/\alpha_\ell$ such cuts
over all preparation calls in that epoch. Every positive eigenvalue
discarded as dust is less than $2\delta$. Each discard permanently
reduces rank, since no later operation enlarges the covariance range.
There are at most $\ell$ discarded eigenvalues, proving
\eqref{impl:dust}. After each continuing cap test there is a paid cut;
the finite cut budget therefore forces the routine to return.
\end{proof}

The response theorem now applies at the actual returned covariance and
actual maximizing density. It gives $\operatorname{Tr}(CJ)\le B\sqrt\ell$
with the chosen profile constant $B$. In the rectangular profile the
response bound at epoch size $\ell\le n$ is no larger than this global
choice, because $q\le1/2$. Thus the covariance derivative test supplies
the curvature control required by the analysis.

\subsection{Sampling and acceptance}

The routine \proc{Project} uses an exact EVD of $C$ to form its high
spectral projection, followed by the scalar constraint scan of
Lemma~\ref{sq:lem:scalarshort}. That lemma identifies the computed matrix
with the canonical orthogonal projection, including zero denominators
and dependent constraints. There are at most $\ell+1$ coordinate and
radial equations. Forming the projection and applying the scan costs
$O((\ell+1)^3)$ scalar operations in addition to the EVD call; every
division follows a test that its denominator is positive. The final
factor $1/2$ gives the movement covariance.

\begin{lemma}[The implemented movement]\label{impl:sample}
The routine \proc{Sample} uses one exact EVD, one uniform draw, and
$O(\ell+1)$ scalar operations. Its outputs and the resulting movements
satisfy Lemmas~\ref{walk:step} and \ref{walk:invariant}. Every movement
curve meets the hypotheses of Theorem~\ref{reg:main}.
\end{lemma}
\begin{proof}
Apply EVD to the computed projection $2Q$. Its nonzero eigenvalues
are all one, and their eigenvectors form an orthonormal basis of $W$.
If there are $k$ such vectors, partition $[0,1)$ into $2k$ equal
half-open intervals, one for each vector and sign. A cumulative scan
implements exactly the uniform law proved in Lemma~\ref{walk:step}.
Lemma~\ref{walk:invariant} ensures
$\operatorname{Tr}Q>0$, cube feasibility, and radial progress.
For the derivative bound, the sampled force satisfies
\[
 \|\A(\xi)\|_{\HS}
 \le\sqrt{D\ell}\,\|\xi\|_2
 \le\ell\sqrt D\le s^2.
\]
For every $|t|\le h\le1/2$, both signs stay in the cube, so
$\|H+t\A(\xi)\|_{\op}\le n$. On the fixed range $U$ of the prepared
covariance,
$C-t^2Q\succeq(3/4)C\succeq(3\delta/2)P_U$ and
$C-t^2Q\preceq P_U$. This exceeds the floor $\gamma_0=2^{-15}$ in
Theorem~\ref{reg:main}, verifying all its hypotheses.
\end{proof}

\begin{lemma}[Acceptance accuracy]\label{impl:acceptance}
For the exact excess $\Psi$ and martingale $M_*$ defined in
\eqref{walk:excess}, the algorithm's reports satisfy
\[
 |\widehat\Psi-\Psi|\le5\epsilon_A/3,
 \qquad |\widehat M-M_*|\le\epsilon_A.
\]
\end{lemma}
\begin{proof}
Proposition~\ref{solver:accuracy} gives initial and terminal value
errors at most $\epsilon_A/3$, and Lemma~\ref{query:gap-density}
gives initial density error at most $e_A=\epsilon_A/(3s^2)$.
The trial invariant gives
$\|H-H_*\|_{\HS},\|Y\|_{\HS}\le3s^2$.
Lemma~\ref{query:plane-report} therefore bounds each pairing error
by $\epsilon_A$. Summing the appropriate errors proves the claims.
Both are smaller than $\sqrt\ell/100$; thus
Lemma~\ref{walk:acceptance} applies to these actual primal SDP reports.
\end{proof}

\subsection{Polynomial work on every execution}
\label{impl:runtime}

\begin{theorem}[Finite work and confidence amplification]
\label{impl:polynomial}
In the arithmetic model of Definition~\ref{model:main}, every execution
of Algorithms~\ref{alg:epoch} and \ref{alg:full}, including all rejected
trials, has a polynomial bound on scalar operations, exact EVD calls,
and uniform draws in $n,D,b$. The retry count $r=M+b+1$ gives total
failure probability at most $2^{-b}$.
\end{theorem}
\begin{proof}
The empty input cases have already returned. For nonempty inputs,
$n\ge1$ and $D\ge2$. We first check the magnitudes of every parameter
that occurs in a loop count or requested accuracy.

For the square profile, $B=770$. In the rectangular profile, the doubling
scan gives a dyadic $p\ge2$ with $(D/n)^{1/p}\le2$ and $p\le6D$.
Substituting the formula for $\theta$ gives
\[
 B=2+\frac{6L^{q/2}(D/n)^{q/2}}{\sqrt{q(1-q)}}
   \le2+96\sqrt p,
 \qquad q=1/p.
\]
Here $L^{q/2}\le8$, $(D/n)^{q/2}\le\sqrt2$, and $1-q\ge1/2$.
Also $\theta\le20s$ and $\theta^{-1}\le s$, directly from its formula.
Thus $B$, $\tau^{-1}$, the number of power blocks, and all scalar
regularizer coefficients are polynomially bounded. The dyadic powers
require only polynomially many scalar square roots and operations.

By their explicit definitions,
$\mu^{-1}$, $M_C$, $M_4$, and $a_0^{-1}$ are polynomial in $s$.
The minimum defining $h$ gives
\[
 h^{-2}=\max\left\{\frac{4n}{a_0^2},4,
                         \frac2\tau,2^{40}M_4\right\},
\]
so every trial has at most $N=\lfloor\tau/h^2\rfloor$ movements and
at most $N+1$ preparation calls. At epoch size $\ell\le n$,
\[
 \alpha_\ell^{-1}
 =\max\{2/\delta,2M_C/t_\ell\}
\]
is polynomial. Lemma~\ref{impl:preparation} bounds the total number of
paid cuts across these calls by $\ell/\alpha_\ell$, with at most
$\ell$ rank drops. Consequently the total number of cap reports and
spectral cleanup calls is polynomial, even when preparation consumes
too much covariance and causes a rejected trial.

The simplified identities in Lemma~\ref{impl:gamma-report} show that
$e_\ell^{-1}$ and $\omega_\ell^{-1}$ are polynomial. Since
$\chi^{-1}\le D$,
$\nu_\ell^{-1}=4/(\chi e_\ell^2)$ is polynomial as well.
The two acceptance accuracies $\epsilon_A/3$ and $\nu_A$ have polynomial
reciprocals. Proposition~\ref{solver:accuracy} therefore bounds each
SDP solve by polynomially many scalar operations and EVD calls.
Its input dimensions are polynomial, and all coefficient arrays are
formed explicitly. Computing each derivative report from the density
requires only polynomially many matrix products, traces, and EVD calls.
The legal projection and its uniform signed sampler have polynomial cost by the
preceding lemmas. There is exactly one uniform draw per movement.
This proves a common polynomial work bound for every trial.

Lemmas~\ref{walk:acceptance} and \ref{walk:phase} prove that each
requested epoch trial, given its saved state, is accepted with probability
at least $499/800$, and that at most $K$ accepted epochs halve the live
count in a phase. There are
at most $n+1$ phases, hence at most $M=K(n+1)$ saved epoch requests.
This also bounds the explicit loops in Algorithm~\ref{alg:full},
independently of success. Each request runs at most $r=M+b+1$ trials.
All requests, retries, and solver calls therefore have polynomial
total cost on every path.

Finally, Proposition~\ref{walk:assembly} applies to the verified
routine contracts. Its conditional retry argument gives
$M2^{-r}=M2^{-M-b-1}\le2^{-b}$, and its phase bound guarantees a full
signing when no request exhausts its retries. The runtime bound above
holds on every execution, without conditioning on that success.
\end{proof}

\begin{proof}[Proof of Theorem~\ref{main:combined}]
The square discrepancy guarantee is Corollary~\ref{walk:square-conclusion},
and the rectangular guarantee is Corollary~\ref{rect:conclusion}.
Proposition~\ref{solver:accuracy}, Lemmas~\ref{impl:gamma-report}--\ref{impl:acceptance},
and Theorem~\ref{reg:main} verify every routine and derivative contract
used in their common walk analysis. Theorem~\ref{impl:polynomial} bounds
the work on every execution in Definition~\ref{model:main}'s arithmetic
model and proves the requested failure probability. Taking $b=1$ gives
positive probability of a successful signing, which also establishes
the existential statement. The empty cases are returned at the first
line of the full algorithm.
\end{proof}

\section{AI Usage}
The author conceived the approach of using operator-valued
$K$-transforms, or variants based on regularizers such as Tsallis--$1/2$
rather than log-determinants, as discrepancy potentials in 2022--2023.
The semidefinite representation of the Tsallis--$1/2$ regularized
potential was subsequently conceived by an internal version of Google's
Gemini, used with a specific harness. Extensive calculations with
ChatGPT 5.5--5.6 Sol Ultra helped develop and test proof strategies,
and rule out simpler approaches, leading to the strategy presented here.
GPT 6 Astra Ultra was used to develop Lean formalizations of the
Matrix Spencer and Weaver/Kadison--Singer arguments, including separate
existence and algorithmic formulations.

\section{Acknowledgements}
The idea of using free probability to approach Weaver's discrepancy
problem, potentially algorithmically, was encouraged by inspiring
conversations with Adam W. Marcus in 2016--2017. The author is
profoundly grateful for his encouragement. Numerous conversations and
continued encouragement from Nikhil Srivastava have been invaluable
over the years, as were early discussions with Nick Ryder and Jonathan
Leake.

The author thanks Daniel A. Spielman for hosting him at Yale University
in the summer of 2019, and for extensive discussions of scalar-valued
free probability and its possible applications to matrix discrepancy.
The approach began to crystallize in the summer of 2023, during a
visit hosted by Ramon van Handel at Princeton University. Ramon's
persistent questions about concrete special cases, bottlenecks, and
obstacles pushed the author to examine the approach carefully and
eliminate many of those obstacles, strengthening his conviction that
it could succeed. The author warmly thanks Ramon for his hospitality,
many discussions of operator-valued free probability, and insistence
on an exceptionally high but necessary standard of clarity and rigor.

Finally, the author thanks Tibo from OpenAI for providing numerous
Codex resets to paid Codex subscribers, which were invaluable in
helping this project go \emph{Ultra Fast}.

\appendix
\section{Uniform derivative bounds for the finite walk}
\label{reg:section}

The algorithm moves on real matrix spaces.  Complex parameters enter only
in this appendix, where Cauchy's estimate provides explicit bounds on real
derivatives.  We first bound the unoptimized density objective and then
differentiate its stationarity equation.  The resulting constants depend
on the density floor and the retained coefficient-covariance floor.  They
do not depend on a smallest positive eigenvalue of the physical source.

We use a common notation for the two profiles.  Let $n,D\ge1$, put
$s=n+D+2$, and fix $\ell\le n$ real symmetric contractions
$\A_1,\ldots,\A_\ell$ of order $D$.  For a coefficient covariance $C$,
write $\eta_C(S)=\sum_{i,j}C_{ij}\A_iS\A_j$ as before.  Define
\begin{align}
 \mathcal R(S,C)&=2F(S,\eta_C(S))
       +\frac{\vartheta}{1-q}\operatorname{Tr}S^{1-q}
       +2\chi\operatorname{Tr}\sqrt S,\label{reg:profile}\\
 \Phi(S;H,C)&=\operatorname{Tr}(HS)+\mathcal R(S,C),
 &E(H,C)&=\max_{S\succeq0,\ \operatorname{Tr}S=1}\Phi(S;H,C).
 \label{reg:value}
\end{align}
In the square case, take $\vartheta=0$, $\chi=1$, and $q=1/2$.
In the rectangular case, take $\vartheta=\theta$, $\chi=\kappa=D^{-1}$,
and the previously chosen $q$.  In both cases the parameter bounds give
\begin{equation}\label{reg:parameterbounds}
 0<q\le\tfrac12,\qquad 0\le\vartheta\le20s,
 \qquad s^{-1}\le\chi\le1.
\end{equation}
The density objective is strictly concave and its optimizer is faithful,
as established for the two profiles.  We denote that unique optimizer by
$S_*(H,C)$.  All derivatives below use the Hilbert--Schmidt metric on
real symmetric matrices; covariance directions are restricted to the
stated coefficient support.

\begin{theorem}[The quantitative bounds used by the algorithm]
\label{reg:main}
Set
\begin{equation}\label{reg:constants}
 \mu=2^{-10}s^{-6},\qquad M_C=2^{300}s^{60},
 \qquad M_4=2^{1040}s^{200},\qquad \gamma_0=2^{-15}.
\end{equation}
For $\|H\|_{\mathrm{op}}\le n+1$ and $0\preceq C\preceq I$,
$S_*(H,C)\succeq\mu I_D$.
On every fixed coefficient support $U$ where
$\gamma_0P_U\preceq C\preceq P_U$, the optimized value is smooth and
\begin{equation}\label{reg:covariancebound}
 |D_C^2E(H,C)[V,V]|\le M_C
 \quad\text{if }V=P_UVP_U\text{ and }\|V\|_{\mathrm{HS}}\le1.
\end{equation}
Further, let $Q$ and $A$ be fixed real symmetric matrices on coefficient
and physical space, respectively, with $0\preceq Q\preceq C$ and
$\|A\|_{\mathrm{HS}}\le s^2$.  Along any interval
$I\subseteq[-1/2,1/2]$ on which
\[
 \|H+tA\|_{\mathrm{op}}\le n+1,
 \qquad \gamma_0P_U\preceq C-t^2Q\preceq P_U,
\]
the scalar function $e(t)=E(H+tA,C-t^2Q)$ satisfies
\begin{equation}\label{reg:fourthbound}
 |e^{(4)}(t)|\le M_4\qquad(t\in I).
\end{equation}
The same assertions hold when $U=\{0\}$, with covariance derivatives
interpreted on the zero-dimensional space.
\end{theorem}

\subsection{The standard analysis facts being used}

For clarity, we state the three external calculus facts and explain the
form in which they enter the proof.

\begin{theorem}[Smooth implicit function theorem]\label{reg:ift}
Let $G(p,y)$ be a smooth map between finite-dimensional real Euclidean
spaces, with $G(p_0,y_0)=0$.  If its partial derivative $D_yG(p_0,y_0)$
is invertible, then near $p_0$ the equation has a unique nearby smooth
solution $y=y(p)$.  Its derivative satisfies
\[
 D_yG(p,y(p))\,Dy(p)[h]+D_pG(p,y(p))[h]=0.
\]
\end{theorem}
This is the finite-dimensional theorem in
\cite[Appendix~B]{guilleminhaine2019}; the displayed equation is the
chain rule.  We apply it to the projected density stationarity equation.
It gives smoothness; all quantitative inverse and derivative bounds are
proved below.

\begin{theorem}[Holomorphic matrix functions]\label{reg:matrixfunctions}
If the spectrum of a complex matrix $M$ lies in an open set on which
a scalar function $f$ is holomorphic, then locally
\[
 f(M)=\frac1{2\pi i}\int_{\mathcal C}
       f(z)(zI-M)^{-1}\,dz,
\]
where $\mathcal C$ is a fixed positively oriented contour enclosing the
spectrum inside that open set.  Consequently $f(M)$ depends
holomorphically on the entries of $M$.  For the principal powers
$f(z)=z^\alpha$, the domain used here is
$\mathbb C\setminus(-\infty,0]$.
\end{theorem}
See \cite[Definition~3.3 and Section~3.1.4]{highamlin2013}.
Holomorphic dependence follows directly because the entries of the
resolvent are holomorphic wherever its determinant is nonzero; the same
contour works in a neighborhood of the base matrix.  A triangularization
also shows that $\operatorname{Tr}f(M)$ is the sum of $f$ evaluated at
the eigenvalues, with algebraic multiplicities.

\begin{theorem}[Cauchy's coefficient bound]\label{reg:cauchy}
If a scalar holomorphic function on a disc of radius $r$ has absolute
value at most $T$, its coefficient of degree $j$ has magnitude at most
$T/r^j$.  For a holomorphic function on a polydisc, applying this formula
successively in each variable bounds the coefficient of
$z_1^{j_1}\cdots z_k^{j_k}$ by $T\prod_i r_i^{-j_i}$.
\end{theorem}
The one-variable statement is Cauchy's integral formula and inequalities
\cite[Chapter~2, Corollary~4.3]{steinshakarchi2003}; the polydisc statement
follows by iterating the one-variable integral formula.

\subsection{Density floor and inverse response}

\begin{lemma}[An explicit density floor]\label{reg:floor}
Under $\|H\|_{\mathrm{op}}\le n+1$ and $0\preceq C\preceq I$,
\[
 S_*(H,C)\succeq\mu I_D.
\]
At every faithful density of trace one, the negative density Hessian
of $\Phi$ is at least $\chi/2$ on the trace-zero symmetric tangent
space.  Its inverse there has operator norm at most
$R_*=2/\chi\le2s$.
\end{lemma}
\begin{proof}
The fidelity contribution $g(S)=2F(S,\eta_C(S))$ is monotone for
positive semidefinite increments of $S$ and homogeneous of degree one.
At a faithful density it is differentiable on the fixed physical support.
Its gradient $G_S$ therefore satisfies
$G_S\succeq0$ and
\[
 \operatorname{Tr}(SG_S)=g(S)\le2\sqrt\ell\le2\sqrt n.
\]
The bound follows from $\eta_C(I)\preceq\ell I$ and the fidelity
trace bound.  Stationarity at the optimizer is
\[
 H+G_S+\vartheta S^{-q}+\chi S^{-1/2}=\lambda I_D.
\]
Pairing with $S$ gives
$\lambda\le\|H\|_{\mathrm{op}}+2\sqrt n+
\vartheta D^q+\chi\sqrt D$.
Discarding the positive terms $G_S$ and $\vartheta S^{-q}$ yields
\[
 S\succeq
 \left(\frac{\chi}
 {2\|H\|_{\mathrm{op}}+2\sqrt n+\vartheta D^q+\chi\sqrt D}
 \right)^2 I_D.
\]
Since $D^q\le\sqrt D$, the denominator is at most $25s^2$.
Together with $\chi\ge s^{-1}$, this proves the stronger lower bound
$S\succeq(625s^6)^{-1}I_D$, and hence the claimed $\mu$.

In an eigenbasis of a faithful trace-one density, write its eigenvalues
as $s_i\in(0,1]$.  The square-root derivative formula gives
\[
 -D_S^2(2\chi\operatorname{Tr}\sqrt S)[X,X]
 =\chi\sum_{i,j}
 \frac{|X_{ij}|^2}{\sqrt{s_i}\sqrt{s_j}(\sqrt{s_i}+\sqrt{s_j})}
 \ge\frac\chi2\|X\|_{\mathrm{HS}}^2.
\]
The remaining nonlinear terms are concave and contribute nonnegative
forms to the negative Hessian.  Restricting to the trace-zero tangent
preserves the lower bound, which proves the inverse estimate.
\end{proof}

\subsection{Why physical-source conditioning does not enter}

\begin{lemma}[Four derivatives of the unoptimized objective]
\label{reg:mixed}
Fix a coefficient support $U$ and
$\gamma P_U\preceq C\preceq P_U$ with $0<\gamma\le1$.
At a trace-one density $S\succeq\mu I_D$, set
\begin{equation}\label{reg:cauchyscales}
 r_0=\frac{\gamma\mu}{64},\qquad
 T_0=16D(1+\sqrt n+\vartheta+\chi),\qquad
 B_0=T_0(8/r_0)^4.
\end{equation}
Every mixed derivative of $\mathcal R$ of orders one through four is
bounded in magnitude by $B_0$ when each argument $(X,V)$ satisfies
$V=P_UVP_U$ and
$\|X\|_{\mathrm{HS}}+\|V\|_{\mathrm{HS}}\le1$.
Moreover,
\begin{equation}\label{reg:mixedpolynomial}
 B_0\le2^{85}\gamma^{-4}s^{26}
       \le2^{100}\gamma^{-4}s^{30}.
\end{equation}
\end{lemma}
\begin{proof}
Complexify the real symmetric density space and the supported real
symmetric coefficient space.  We use the complex Hilbert--Schmidt norm
for their perturbations.  For
$\|\Delta S\|_{\mathrm{HS}}+\|\Delta C\|_{\mathrm{HS}}<r_0$, define
\[
 U_S=S^{-1/2}\Delta S S^{-1/2},\qquad
 U_C=C_U^{-1/2}(\Delta C)|_U C_U^{-1/2},
 \qquad C_U=C|_U.
\]
Both relative perturbation norms are less than $1/64$.
Choose an orthonormal basis $u_1,\ldots,u_r$ of $U$, where $r=\dim U$.
In this basis use the positive square root of $C_U$ and set
\[
 B_a=\sum_{b=1}^r(C_U^{1/2})_{ba}\A(u_b),
 \qquad \A(u_b)=\sum_{i=1}^{\ell}(u_b)_i\A_i.
\]
These are real symmetric Kraus matrices.  Using this specific factor,
rather than an arbitrary Kraus representation, makes $U_C$ in the next
identity exactly the relative covariance perturbation defined above.
At the fixed base point let
\[
 \mathcal Vv=(S^{1/2}B_av)_{a=1}^r,\qquad
 M=\eta_C(S)=\mathcal V^*\mathcal V.
\]
The adjoint here belongs to the fixed base map; it is not applied to
the varying complex parameters.  Direct expansion of the covariance
indices gives the holomorphic identity
\begin{equation}\label{reg:relativefactor}
 \eta_{C+\Delta C}(S+\Delta S)
 =\mathcal V^*[(I+U_C)\otimes(I+U_S)]\mathcal V.
\end{equation}

Let $K=\operatorname{ran}\eta_C(I)$ be the fixed physical source support.
It equals the span of the ranges of the Kraus matrices and is independent
of the faithful density.  Since all coefficient perturbations stay on
$U$, every perturbed source is supported on this same $K$.
If $K=\{0\}$, the fidelity vanishes identically and the following
argument is needed only for the density powers.  Otherwise
$M_K=M|_K$ is positive definite and
$\mathcal VM_K^{-1/2}:K\to\bigoplus_{a=1}^r\mathbb C^D$ is an isometry.
Equation~\eqref{reg:relativefactor} therefore proves
\begin{equation}\label{reg:relativesource}
 \left\|M_K^{-1/2}(M_{\mathrm{pert}}-M_K)M_K^{-1/2}\right\|_{\mathrm{op}}
 <\frac2{64}+\frac1{4096}<\frac1{16},
\end{equation}
where $M_{\mathrm{pert}}$ is the perturbed source compressed to $K$.
The bound is relative to the source itself and depends only on the
density and coefficient-covariance perturbations.

We check explicitly that the matrix powers keep their principal branch.
For a positive definite $P$, any relative perturbation
$A=P^{1/2}(I+E)P^{1/2}$ with $\|E\|_{\mathrm{op}}<\alpha<1$
has numerical range in the sector
\[
 |\arg z|<\arctan\!\frac\alpha{1-\alpha}.
\]
Indeed $v^*Av/(v^*Pv)=1+w^*Ew/\|w\|^2$ for $w=P^{1/2}v$.
In particular $A$ is invertible.  Its inverse has the same sector
bound: write $v=Aw$ in $v^*A^{-1}v=w^*A^*w$.
If two such matrices are denoted by $A,B$ and $ABv=\lambda v$, then
\[
 \lambda=\frac{v^*Bv}{v^*A^{-1}v}.
\]
For the relative bounds $1/64$ and $1/16$ above, the sum of the two
sector angles is less than $\pi$.  Therefore the product has no
eigenvalue on $(-\infty,0]$.

Apply this observation to
$A=(S+\Delta S)|_K$ and $B=M_{\mathrm{pert}}$.
The compressed relative perturbation inherits the bound on $U_S$.
Indeed, with $S_K=S|_K$ and inclusion $\iota_K:K\to\mathbb C^D$,
$S^{1/2}\iota_KS_K^{-1/2}$ is an isometry, and compression conjugates
$U_S$ by that isometry.
The function $\operatorname{Tr}_K(AB)^{1/2}$ is holomorphic by
Theorem~\ref{reg:matrixfunctions}, and equals fidelity on real positive
inputs.  The compression is essential here: the density remains full
dimensional, while only the zero source directions are omitted from
this particular trace.  Since
$\|A\|_{\mathrm{op}}\le1+1/64$ and
$\|B\|_{\mathrm{op}}\le(1+1/16)n$, summing the principal roots of
the eigenvalues gives
\[
 |\operatorname{Tr}_K(AB)^{1/2}|\le2D\sqrt n.
\]
The same matrix-function theorem makes
$\operatorname{Tr}(S+\Delta S)^{1-q}$ and
$\operatorname{Tr}(S+\Delta S)^{1/2}$ holomorphic, each with absolute
value at most $2D$.  Consequently the holomorphic extension of
$\mathcal R$ has absolute value at most $T_0$ on the indicated ball.

For $1\le j\le4$ and unit directions $w_1,\ldots,w_j$ in the stated
sum norm, substitute $(S,C)+\sum_{i=1}^jz_iw_i$ and integrate on the
circles $|z_i|=r_0/(2j)$.  The polydisc lies inside the ball, and its
coefficient of $z_1\cdots z_j$ is exactly the mixed derivative
$D^j\mathcal R[w_1,\ldots,w_j]$.  Theorem~\ref{reg:cauchy} bounds it
by $T_0(2j/r_0)^j\le B_0$.  Finally,
$T_0\le2^9s^2$ and $r_0^{-1}=2^{16}\gamma^{-1}s^6$, so
$B_0\le2^{85}\gamma^{-4}s^{26}$ as claimed.
\end{proof}

\subsection{Passing the bounds through the optimizer}

On a fixed coefficient face, Lemma~\ref{reg:mixed} gives a smooth
unoptimized objective near every faithful density.  Let $P_0$ be the
orthogonal projection onto trace-zero symmetric matrices.  In an
orthonormal coordinate chart of that tangent space, stationarity is
$P_0\nabla_S\Phi=0$.  Its derivative in the density variable is the
negative definite tangent Hessian.  Lemma~\ref{reg:floor} makes it
invertible, so Theorem~\ref{reg:ift} gives a smooth local solution.
The continuous unique global optimizer coincides with that solution:
compactness and uniqueness imply continuity, and hence keep the global
optimizer in the local neighborhood.  This proves the smoothness of
$S_*(H,C)$ and $E(H,C)$ on each such face.

We first obtain \eqref{reg:covariancebound}.  At fixed $H$, let $V$ be a
supported covariance direction of Hilbert--Schmidt norm at most one.
Write $A_S=-P_0\Phi_{SS}P_0$ for the positive tangent Hessian, and
$b=P_0\Phi_{SC}[\,\cdot\,,V]$ for the corresponding tangent force,
identified with a vector by the Hilbert--Schmidt inner product.
Implicit differentiation and the envelope identity give
\begin{equation}\label{reg:covarianceenvelope}
 D_C^2E[V,V]=\Phi_{CC}[V,V]+\langle b,A_S^{-1}b\rangle.
\end{equation}
At the optimizer, the density floor allows
Lemma~\ref{reg:mixed} to bound $|\Phi_{CC}[V,V]|$ and $\|b\|$ by
$B_0$, while $\|A_S^{-1}\|\le R_*$.
For $\gamma\ge\gamma_0=2^{-15}$, the sharper inequality in
\eqref{reg:mixedpolynomial} gives $B_0\le2^{145}s^{26}$.
It follows that
\[
 |D_C^2E[V,V]|\le B_0+R_*B_0^2
 \le2^{292}s^{53}\le M_C.
\]
This proof of the covariance curvature bound uses the same uniform
local derivatives as the movement estimate; it requires no separate
bound on an inverse physical source.

For the fourth derivative, fix a legal matched path from
Theorem~\ref{reg:main} and define
\[
 \mathcal F(S,t)=\operatorname{Tr}((H+tA)S)+\mathcal R(S,C-t^2Q),
 \qquad K_*=32B_0s^2,
\]
where $B_0$ is evaluated with the lower covariance floor $\gamma$ on
this path.  Since $0\preceq Q\preceq I$ on a space of dimension at
most $n$, $\|Q\|_{\mathrm{HS}}\le\sqrt n$.
For $|t|\le1/2$, the first two derivatives of $C-t^2Q$ have norms at
most $v$ and $2v$, respectively, where $v=\sqrt n$; higher derivatives
vanish.  The chain-rule multipliers through four scalar derivatives
are bounded respectively by
\[
 1,\quad v,\quad v^2+2v,\quad v^3+6v^2,
 \quad v^4+12v^3+12v^2.
\]
Each is at most $25s^2$.  The linear-center contributions are bounded
by $s^2$, using $\|A\|_{\mathrm{HS}}\le s^2$ and
$\|H+tA\|_{\mathrm{HS}}\le\sqrt D(n+1)\le s^2$.
Thus every mixed derivative of $\mathcal F$ of total order one through
four, with unit density directions, has magnitude at most $K_*$.

Put $a_j=\|S_*^{(j)}(t)\|_{\mathrm{HS}}$ for $j=1,2,3$.
Repeated differentiation of projected stationarity gives
\begin{align*}
 a_1&\le R_*K_*,\\
 a_2&\le R_*K_*(1+a_1)^2\le4R_*^3K_*^3,\\
 a_3&\le R_*K_*\bigl((1+a_1)^3+3(1+a_1)a_2\bigr)
       \le32R_*^5K_*^5.
\end{align*}
For example, the second differentiated equation has leading term
$\mathcal F_{SS}S_*''$ and remaining terms
$\mathcal F_{SSS}[S_*',S_*']+2\mathcal F_{SSt}S_*'+\mathcal F_{Stt}$.
At the next order the terms group into a cubic expression in $(S_*',1)$
and three bilinear expressions in $(S_*'',0)$ and $(S_*',1)$.
All density derivatives lie in the trace-zero tangent, so the inverse
bound $R_*$ applies at each differentiation.

Finally, the envelope identity is $e'(t)=\mathcal F_t(S_*(t),t)$.
Differentiating it three times yields
\[
 |e^{(4)}(t)|\le
 K_*\bigl((1+a_1)^3+3(1+a_1)a_2+a_3\bigr)
 \le64R_*^5K_*^6.
\]
Using the weaker but convenient bound
$B_0\le2^{100}\gamma^{-4}s^{30}$ gives
$K_*\le2^{105}\gamma^{-4}s^{32}$ and therefore
\[
 64R_*^5K_*^6\le2^{641}\gamma^{-24}s^{197}
 \le2^{1001}s^{197}\le M_4
 \qquad(\gamma\ge2^{-15}).
\]
Together with Lemma~\ref{reg:floor} and
\eqref{reg:covarianceenvelope}, this completes the proof of
Theorem~\ref{reg:main}.

\begin{remark}[What these constants control]
The covariance bound controls a finite paid covariance cut.  The fourth
derivative bound controls the remainder after averaging the two signs
of a finite movement.  The numerical recipe chooses
\[
 h=\min\left\{\frac{a_0}{2\sqrt n},\frac12,
                 \sqrt{\frac\tau2},\sqrt{\frac{2^{-40}}{M_4}}\right\},
\]
where $a_0$ is the boundary margin and $\tau$ is the epoch duration.
Consequently $M_4h^2\le2^{-40}$.
These bounds are independent of numerical SDP accuracy and of a spectral
gap.  Improving a solver's precision does not remove this finite-step
remainder; the mesh size makes it affordable.
\end{remark}

\end{document}